\documentclass[acmsmall]{acmart}

\usepackage{booktabs} 

\usepackage[ruled]{algorithm2e} 

\SetAlFnt{\small}
\SetAlCapFnt{\small}
\SetAlCapNameFnt{\small}
\SetAlCapHSkip{0pt}
\IncMargin{-\parindent}

\usepackage[utf8]{inputenc}
\usepackage[T1]{fontenc}
\usepackage{microtype} 

\usepackage{amsmath,amsthm,amsfonts}
\usepackage{proof}
\usepackage{mathpartir}
\usepackage{hyperref}
\usepackage{stmaryrd}
\usepackage{color}
\usepackage{natbib}
\usepackage{graphicx}
\usepackage[all]{xy}
\usepackage{etoolbox}

\newtheorem{definition}{Definition}

\newtheorem{lemma}{Lemma}
\newtheorem{corollary}{Corollary}

\newtheorem*{proposition*}{Proposition}
\newtheorem*{thm*}{Theorem}

\newcommand{\rj}[7]{#1 \mid #2 \vdash #3 : #4 \sim #5 : #6 \mid #7}

\newcommand{\jrhol}[7]{#1 \mid #2 \vdash #3 : #4 \sim #5 : #6 ~\{#7\}}

\newcommand{\jimp}{\Rrightarrow}
\newcommand{\holwf}[2]{#1 \vdash #2}
\newcommand{\uj}[5]{#1 \mid #2 \vdash #3 : #4 \{#5\}}

\newcommand{\armj}[9]{#1 \mid #2 \vdash \{#3\}  #4 : #5 \sim #6 : #7
\{\!\{#8\}\!\}_{#9}}

\newcommand{\armjnoc}[7]{\vdash \{#1\}~#2 : #3 \sim #4 : #5~
\{\!\{#6\}\!\}_{#7}}

\newcommand{\aumj}[7]{#1 \mid #2 \vdash \{#3\} #4 : #5
  \{\!\{#6\}\!\}_{#7} }
\newcommand{\aumjnoctxt}[5]{\vdash \{#1 \} #2 : #3 \{\!\{ #4 \}\!\}_{#5}}
\newcommand{\aumje}[7]{#1 \mid #2 \vdash \{#3\} #4 : #5 \{\!\{#6\}\!\}_{#7} }

\newcommand{\juhol}[5]{#1 \mid #2 \vdash #3 : #4 \{ #5 \} }

\newcommand{\jhol}[3]{#1 \mid #2 \vdash #3}
\newcommand{\jlc}[3]{#1 \vdash #2 : #3}
\newcommand{\res}{{\mathbf r}}
\newcommand{\st}{{\mathbf s}}

\newcommand{\vl}{{\mathbf v}}
\newcommand{\resl}{\res_1}
\newcommand{\resr}{\res_2}
\newcommand{\vll}{\vl_1}
\newcommand{\vlr}{\vl_2}
\newcommand{\stl}{\st_1}
\newcommand{\str}{\st_2}

\newcommand{\subst}[2]{[#2/#1]}

\newcommand{\defeq}{\triangleq}

\newcommand{\loc}{\mathsf{Loc}}

\newcommand{\rname}[1]{${\sf [#1]}$}

\newcommand{\htriple}[5]{\mathcal{H}_{#1,#2}(#3,#4,#5)}
\newcommand{\hquad}[6]{\mathcal{H}_{#1,#2}(#3,#4,#5,#6)}

\newcommand{\sem}[1]{\llbracket #1 \rrbracket}

\newcommand{\semp}[1]{{(\![} #1 {]\!)}}

\newcommand{\tunit}{\mathbb{U}}
\newcommand{\tmem}{\mathbb{M}}
\newcommand{\tval}{\mathbb{V}}
\newcommand{\tbool}{\mathbb{B}}

\newcommand{\D}[1]{#1\cdot 1}
\newcommand{\qbsunit}{1}
\newcommand{\qbsmem}{M}
\newcommand{\qbsval}{V}
\newcommand{\qbsbool}{\D{\{\bot,\top\}}}
\newcommand{\qbsnat}[1]{\D{#1}}
\newcommand{\qbsposr}{[0,\infty]_{\QBS}}

\newcommand{\coe}[2]{c_{#1,#2}}

\newcommand{\monn}{{\mathsf T}} 
\newcommand{\mon}[2]{{\monn_{#1}(#2)}}

\newcommand{\eff}[1]{\mathbf{Eff}(#1)}
\newcommand{\safe}[1]{\mathsf{Safe}(#1)}

\newcommand{\unit}[1]{{\sf unit}(#1)}

\newcommand{\mletA}[2]{{\sf let}\ #1 = #2\ {\sf in}}
\newcommand{\mlet}[3]{{\sf let}\ #1 = #2\ {\sf in}\ #3}

\newcommand{\wrt}[2]{#1:= #2}
\newcommand{\rd}[1]{{\sf read}\ #1}
\newcommand{\unif}[1]{{\sf Unif}(#1)}
\newcommand{\lap}[2]{{\sf Lap}_{#1}(#2)}

\newcommand{\mfold}[3]{{\sf mfold}\ #1\ #2\ #3}

\newcommand{\casebool}[3]{{\sf if}\ #1\ {\sf then}\ #2 \ {\sf else}\ #3}

\newcommand{\ttrue}{\mathbf{tt}}
\newcommand{\ffalse}{\mathbf{ff}}
\newcommand{\sskip}{{\sf skip}}

\newcommand{\sample}[1]{{\sf sample}(#1)}

\newcommand{\pair}[2]{\langle #1, #2 \rangle}
\newcommand{\prl}[1]{\pi_1(#1)}
\newcommand{\prr}[1]{\pi_2(#1)}

\newcommand{\id}[0]{{\sf id}}

\newcommand{\inj}[1]{\langle #1 \rangle}
\newcommand{\injs}[1]{[ #1 ]}

\newcommand{\cA}{\mathcal{A}}

\newcommand{\cD}{\mathcal{D}}

\newcommand{\cL}{\mathcal{L}}

\newcommand{\cO}{\mathcal{O}}
\newcommand{\cP}{\mathcal{P}}
\newcommand{\cR}{\mathcal{R}}
\newcommand{\cS}{\mathcal{S}}
\newcommand{\cT}{\mathcal{T}}

\newcommand{\BB}{\mathbb{B}}
\newcommand{\CC}{\mathbb{C}}

\newcommand{\EE}{\mathbb{E}}

\newcommand{\NN}{\mathbb{N}}
\newcommand{\PP}{\mathbb{P}}

\newcommand{\nat}{\mathbb{N}}

\newcommand{\RR}{\mathbb{R}}
\newcommand{\bool}{\mathbb{B}}

\newcommand{\posr}{[0,\infty]}

\newcommand{\QBS}{\mathbf{QBS}}
\newcommand{\Meas}{\mathbf{Meas}}
\newcommand{\Set}{\mathbf{Set}}

\newcommand{\BRel}[1]{\mathbf{BRel}(#1)}

\newcommand{\To}{\Rightarrow}

\newcommand{\dto}{\mathbin{\dot\rightarrow}}
\newcommand{\dtimes}{\mathbin{\dot\times}}

\newcommand{\ddtimes}{\mathbin{\ddot\times}}

\newcommand{\stmtr}[1]{{#1}{{\cS}}}
\newcommand{\dstmtr}[1]{{#1}\dot{{\sf S}}}

\newcommand{\ddstmtr}[1]{\ddot{\sf S}(#1)}

\newcommand{\mprobn}{{{\cP}}} 
\newcommand{\mprob}{\mprobn} 

\newcommand{\glprobn}{{\dot\mprobn}} 
\newcommand{\glprob}[3][]{\glprobn_{#1}(#2)(#3)} 

\newcommand{\glubn}{\dot\mprobn^{\mathrm{ub}}} 
\newcommand{\glub}[3][]{\glubn_{#1}(#2)(#3)} 

\newcommand{\gldpn}{\dot\mprobn^{\mathrm{dp}}} 

\newcommand{\glexpn}{\dot\mprobn^{\mathrm{exp}}} 
\newcommand{\glexp}[3][]{\glexpn_{#1}(#2)(#3)} 

\newcommand{\mpstn}{{\stmtr\mprobn}} 
\newcommand{\mpst}{\mpstn} 

\newcommand{\lpstn}{\dstmtr{\glprobn}} 
\newcommand{\lpst}[3][]{\lpstn_{#1}(#2)(#3)} 

\newif\ifcomments
\commentsfalse

\newcommand{\tetsuya}[1]{\ifcomments\textit{\color{darkgreen}[TS]: #1}\fi}

\newcommand{\sk}[1]{\ifcomments\textit{\color{blue}[SK]: #1}\fi}
\definecolor{darkgreen}{RGB}{50,170,0}

\definecolor{darkred}{RGB}{180,50,50}

\newcommand{\ev}{\mathrm{ev}}
\newcommand{\qbszero}{\mathrm{zero}}
\newcommand{\qbssucc}{\mathrm{succ}}

\newcommand{\incl}{{{\pi_1^*I}}}
\newcommand{\mixm}[2]{m_{#1,#2}}
\newcommand{\mix}[3]{\mixm{#1}{#2:#3}}
\newcommand{\sub}[2]{\mathrm{sub}^{#1}_{#2}}

\newcommand{\dom}{\mathop{\mathrm {dom}}}

\newif\ifhideproofs
\ifhideproofs
\usepackage{environ}
\NewEnviron{hide}{}

\fi

\setcopyright{rightsretained}
\acmJournal{PACMPL}
\acmYear{2021}
\acmSubmissionID{icfp21main-p187-p}
\copyrightyear{2021}
\acmVolume{5}
\acmNumber{ICFP}
\acmArticle{93}
\acmMonth{8}
\acmPrice{}
\acmDOI{10.1145/3473598}

\begin{document}

\title [Higher-order probabilistic adversarial computations]{Higher-order probabilistic adversarial computations: Categorical semantics and program logics}


\author{Alejandro Aguirre}
\orcid{0000-0001-6746-2734}
\authornote{This research was carried
	out while the first author was affiliated to the IMDEA
	Software Institute and Universidad Polit\'ecnica de Madrid}
\affiliation{
	\institution{Aarhus University}
	\country{Denmark}
}
\email{alejandro@cs.au.dk}

\author{Gilles Barthe}
\affiliation{
	\institution{MPI-SP}
	\country{Germany}
}
\affiliation{
	\institution{IMDEA Software Institute}
	\country{Spain}
}
\email{gbarthe@mpi-sp.org}
 
\author{Marco Gaboardi}
\affiliation{
	\institution{Boston University}
	\country{USA}
}
\email{gaboardi@bu.edu}

\author{Deepak Garg}
\affiliation{
	\institution{Max Planck Institute for Software Systems}
	\country{Germany}
}
\email{dg@mpi-sws.org}

 \author{Shin-ya Katsumata}
 \affiliation{
   \institution{National Institute of Informatics}            
   \streetaddress{2-1-2 Hitotsubashi}
   \city{Chiyoda-ku}
   \state{Tokyo}
   \postcode{101-8430}
   \country{Japan}                    
 }
 \orcid{0000-0001-7529-5489}
 \email{s-katsumata@nii.ac.jp}

 \author{Tetsuya Sato}
 \affiliation{
	 \institution{Tokyo Institute of Technology}
	 \country{Japan}
}
\email{tsato@c.titech.ac.jp}

\begin{abstract}
  Adversarial computations are a widely studied class of computations
  where resource-bounded probabilistic adversaries have access to
  oracles, i.e., probabilistic procedures with private state. These
  computations arise routinely in several domains, including security,
  privacy and machine learning.

  In this paper, we develop program logics for reasoning about
  adversarial computations in a higher-order setting. Our logics are
  built on top of a simply typed $\lambda$-calculus extended with a
  graded monad for probabilities and state. The grading is used to
  model and restrict the memory footprint and the cost (in terms of
  oracle calls) of computations. Under this view, an adversary is a
  higher-order expression that expects as arguments the code of its
  oracles. We develop unary program logics for reasoning about error
  probabilities and expected values, and a relational logic for
  reasoning about coupling-based properties. All logics feature rules
  for adversarial computations, and yield guarantees that are valid
  for all adversaries that satisfy a fixed resource policy. We prove
  the soundness of the logics in the category of quasi-Borel spaces,
  using a general notion of graded predicate liftings, and we use
  logical relations over graded predicate liftings to establish the
  soundness of proof rules for adversaries. We illustrate the working
  of our logics with simple but illustrative examples.
\end{abstract}

  \maketitle
  \renewcommand{\shortauthors}{A. Aguirre, G. Barthe, M. Gaboardi, D. Garg, S. Katsumata, and T. Sato}

\section{Introduction}
Probabilistic programs occur widely in privacy, security, and other
domains where formal guarantees are required. These guarantees are
often expressed using expectations, e.g, one may want to prove
that the expected value of a randomized algorithm remains close to
some deterministic function of its input.  This can be
established by means of expectation-based methods that originate from
the works of~\citet{Kozen85} and of~\citet{Morgan96}. Another class of guarantees is concerned with
proving the probability of events; e.g., one may want to prove
that a randomized algorithm has a small probability of returning an
incorrect answer, or more generally that there is a small probability
that a bad event happens. These kinds of properties are the target of
so-called Boolean-based methods, such as the union bound logic
proposed by~\cite{BartheGGHS16}. These two approaches
are traditionally used to reason about properties concerning a single
program execution. However, many security and privacy properties are
naturally expressed by relating two program executions; we call such
properties relational properties. Relational counterparts to
expectation-based and Boolean-based methods have been proposed,
including the relational expectation-based logic
of~\citet{BartheEGHS18}, and probabilistic relational Hoare
logic~\cite{BartheGZ09}.
 


Some of these logics additionally support reasoning about
\emph{adversarial computations}, where resource-bounded but otherwise
unconstrained adversaries interact with oracles, i.e.\, probabilistic
procedures with private state. These logics view adversaries as
uninterpreted procedures, and restrict their power by adding
constraints on the memory they can read and write, and on the number
of times they can call other procedures. These constraints are
captured by a notion of valid adversary, and it is reasonably simple
to define proof rules for valid adversaries. 
%
The combination of program logics and adversary rules yield powerful
frameworks that have been used to reason about many examples,
including security of cryptographic constructions~\cite{BartheGZ09}
and stability of machine learning algorithms~\cite{BartheEGHS18}.

The aforementioned works are developed on top of a core probabilistic
imperative language. However, it is often desirable to reason about
higher-order programs, either because the programs of interest are
written in a higher-order language, or more fundamentally because the
programs manipulate higher-order objects. Unfortunately, program
logics for higher-order probabilistic languages are not as well
understood as their counterparts for imperative languages. One
potential reason for this is that denotational semantics of
higher-order probabilistic programs have been lacking. Indeed, a
classic result by~\citet{aumann1961borel} shows that the
category of Borel spaces is not Cartesian closed, and therefore it
cannot be used to interpret programs. 
Fortunately, recent works propose elegant semantics for higher-order
probabilistic programs, such as Probabilistic Coherent Spaces (or
PCoh)~\cite{DanosE11} and Quasi-Borel Spaces (or
QBS)~\cite{HeunenKSY17}.  These semantics can be used as a basis for
developing program logics, as shown for instance by~\citet{SatoABGGH19},
who develop unary and relational logics over QBS.  However, reasoning
in this system is based on an axiomatization of probabilities, and is
intricate. Moreover, this system does not support reasoning about
state and adversarial computations.

\paragraph*{Goals and technical outline}
In this paper we set out to develop a \emph{general framework} for
designing program logics that reason about resource-constrained
adversarial computations in a higher-order probabilistic language. The
reasoning principles themselves are fairly natural, and have been
considered in the first-order setting before~\cite{BartheGZ09}, but
generalizing them to the higher-order requires addressing the
following challenges:
\begin{itemize}
	\item How can we enforce the restrictions on the adversaries?
	\item Can we support relational or expectation-based logics?
	\item How can we formalize the reasoning principles into a
          \emph{common set of} proof rules? How can we prove these
          rules sound?
	\item How can we give a denotational model to these logics?
\end{itemize}


Program properties in an adversarial setting usually make some
assumptions about adversaries by restricting the number of times they
can invoke the oracle, and denying them access to the private state of
the oracle (formally, the oracle is a function with hidden local
state, passed as an argument to the adversary). In the first-order
setting this is usually addressed by restricting the syntax of
adversaries in an ad-hoc manner, but for higher-order programs a more
principled approach would involve using the type system to enforce
these restrictions. Another idea would be to use local state and some
sort of separation logic~\cite{TassarottiH19}, but it is not clear how
such features can be added to denotational models for higher-order
probabilistic programs. The solution we propose here first involves
grading a monad for global state and probabilities by two parameters
$\Sigma$ and $k$: $\Sigma$ represents the memory footprint of the
computation and $k$ represents the number of oracle calls. Thus, our
language has types of the form $\mon{\Sigma,k}{\tau}$ to represent
computations with memory footprint $\Sigma$ and at most $k$ oracle
calls. Then, we allow quantification over memory grading, which can be
seen as a lightweight form of polymorphism. We then model adversaries
as computations of second-order types, e.g.\, the type $\forall
\alpha. (\sigma \rightarrow \mon{\alpha,1}{\tau}) \rightarrow
\mon{\Sigma\cup \alpha,k}{\tau'}$ captures an adversary that has
access to an oracle of type $\sigma \rightarrow \mon{\alpha,1}{\tau}$
and that returns values of type $\tau'$. The grading ensures that the
adversary can call the oracle at most $k$ times, and separation
between adversary and oracle memories is enforced by a parametricity
property derived from the quantification in the type: the adversary
can only read and write the memory region $\Sigma$; and in particular,
it cannot access the private memory of the oracle (denoted
$\alpha$). To our knowledge, this is the first use of this form of
parametricity.

On top of this language, we develop a Boolean-based unary logic, an
expectation-based unary logic, and a Boolean-based relational logic.
The first of these logics can be used to reason about the probability
that the output of a program satisfies some assertion. Its judgments
are based on generalized Hoare triples of the form
$\{\phi\}~t:\mon{\Sigma,\epsilon}{A}~\{\!\{\psi\}\!\}_{\delta}$, with the meaning that if the
initial state satisfies $\phi$, then the final state after running $t$
satisfies $\psi$ with probability at least $1-\delta$. The logic's
interpretation is based on a graded monad lifting, which maps the
postcondition $\psi$ and the grading $\delta$ to an assertion over
probability distributions.

Crucially, soundness of this first logic does \emph{not} depend on the
concrete definition of the lifting, but only on some algebraic
properties of the lifting, so the logic can be generalized. We use
this observation to develop a second higher-order program logic for a
completely different purpose, namely, proving properties of
expectations, similar to~\citet{Morgan96}.  In this logic,
\emph{assertions are real-valued functions}, as opposed to
Boolean-valued assertions of the first logic. Remarkably, most of the
proof rules of the two logics are the same, thanks to the similar
algebraic properties of the underlying liftings. This shows how, by
exploiting similarity in the underlying liftings, we can get almost
similar proof rules to prove completely different properties with
different truth values. We believe that building two
differently-valued logics (real-valued and Boolean-valued) from common
rules is novel.

Our third logic is a relational logic that can be used to prove
properties~\cite{BartheFGGHS16} of \emph{pairs} of higher-order
probabilistic programs using couplings. Again, we exploit the
structure of liftings (couplings are particular cases of liftings),
this time for relational reasoning.

To each of the three logics we add (structurally very similar) proof
rules for reasoning about adversaries. Adversary rules combine all of
the features of our framework, and can be used to reason about the
interaction of an oracle $\cO$, whose code we know, with an adversary
$\cA$ of which we only know the type. Our type system enforces that
the adversary can only call the oracle at most $k$ times and it cannot
access the oracle's private memory. The adversary rules of all three
logics have similar structure and follow the same underlying pattern:
Assuming some invariant about the oracle's private state (which we can
discharge in our logics), derive a property of \emph{any} adversary
that can call the oracle at most $k$ times. For instance, in the first
logic above, the adversary rule says that if the oracle preserves an
invariant $\phi$ with probability at least $1 - \delta$, then any
adversary calling the oracle at most $k$ times preserves $\phi$ with
probability at least $1-k\delta$.


Next, we define a semantic model for our language and the three
logics. Just modeling the language with its higher-order nature and
probabilities is nontrivial as explained earlier.  Concretely, we
model our language in the category $\QBS$ of Quasi-Borel spaces. We
then interpret monadic types using the monad $\cT(-) \triangleq M \To
\cP( - \times M)$ for some QBS $M$ of memories, where $\cP$ denotes
the monad of probability measures over $\QBS$.

Next, we wish to build a uniform framework to model our three logics, 
their different notions of truth-values, different
liftings, and both unary and relational reasoning.  For this, we build
our theory using the notion of \emph{Heyting-valued predicates}, which
are maps from a set $X$ to a Heyting algebra $\Omega$. By
instantiating $\Omega$ differently, we are able to model our different
logics. Further, to interpret logics themselves we employ
\emph{graded monad liftings} \cite{10.1145/2535838.2535846}, which map a Heyting-valued predicate
over a set $X$ to a Heyting-valued predicate over the set of distributions over
$X$. 
We also introduce a novel concept of \emph{stateful lifting}, which
combines graded monad liftings with the state monad. This gives a
categorical semantics of our new Hoare-triple type
(c.f. \cite{NanevskiMB08}):
``$\{\phi\} t:\mon{\Sigma,\epsilon}{A} \{\!\{\psi\}\!\}_\delta$'', where
$\phi,\psi$ are $\Omega$-valued predicates, $\delta$ is a grading and
the whole type specifies properties of probabilities of state
transformers. In doing so, we carefully design a
categorical framework that unifies qualitative and quantitative
assertions using Heyting algebras, and admits interpretations
of the triples under a generic graded lifting.
Soundness of the different logics follows
\emph{uniformly} by suitably instantiating the liftings and the
Heyting algebras.

The soundness of the adversary rules needs separate proofs, since we
must show that the rules are sound for any term inhabiting the
adversary's type. This can usually be done with logical relations, but
an approach based on standard logical relations would fail here, since
it would not take into account the latent effect of the types and
their relation to the invariant. Therefore, we develop a novel logical
relation that is parametrized by the invariant we want to preserve,
and graded by the probability of failure.

\paragraph*{Contributions}
In summary, our contributions are the following:

\begin{itemize}
	\item We design a type system for a higher-order probabilistic
          language to model adversaries and restrict their
          capabilities.  This is achieved through the use of a monad
          graded by the memory footprint and the cost the
          computations, and exploiting parametricity over the memory
          usage. This novel application of parametricity allows us to
          enforce a separation between the adversary and the oracle
          memories in a setting with global state.

	\item We design three unary and relational logics to reason
          about probabilistic programs in this setting. We go beyond
          logics in which assertions are Boolean by also presenting a
          logic in which assertions are real-valued functions, whose
          expected value the logic establishes.  Assertions in our
          logics are also graded, to allow us to reason about the
          probability of failure, or the tightness of bounds. The
          logics are instances of a generic structure -- both in the
          proof rules and the semantics -- showcasing the common
          structure behind them.

	\item We introduce a notion of stateful lifting, which is used to
	interpret the triples in our judgments, from which we can
	construct a categorical model for the rest of the framework.
	This model is parametrized by a Heyting algebra of truth values
	and a graded lifting that interprets assertions. This allows us
	to have a uniform categorical model which is general enough for
	all the logics that we present. 
	
	\item We introduce rules to reason about the interaction
          between adversaries, from which we only know their type, and oracles.
          This uses the parametricity above, to
          show that an invariant is preserved, and moreover it uses
          the cost restriction on the adversary to compute the grading
          of the interaction. Soundness of these rules follows
          from a novel logical relation.
\end{itemize}

  \section{Illustrative examples}
\label{sec:motivation}

We introduce two illustrative examples, which we use to motivate our
modeling of adversaries, and later to showcase the mechanics of our
different logics. Our examples are deliberately simple; further
examples are in the conclusion and the appendix.

\paragraph*{Pollution attacks against Bloom filters~\cite{GerbetKL15}}
Bloom Filters~\cite{Bloom70} are probabilistic data structures useful
to represent sets efficiently at the cost of a loss in precision.
Informally, a Bloom Filter is a data structure with two procedures: an
insertion procedure for adding a value to the current set, and a
membership procedure to query whether a value belongs to the current
set. For simplicity, we assume that values are taken from the set
$[n]=\{0, \ldots, n-1\}$ for some $n$. A Bloom Filter represents
subsets of $[n]$ as an array $L$ of bits of fixed size $m$.  Initially
all bits in $L$ are set to 0. The insertion procedure is parametrized
by a hash function $H: ([n] \times [\ell])\rightarrow [m]$ sampled
uniformly at random, where $\ell$ is a parameter of the Bloom
Filter. The procedure ${\sf insert}(x)$ updates to 1 the value of the
array at positions $h(x,1),\ldots,h(x,\ell)$. The procedure ${\sf
  member}(x)$ computes $h(x,1),\ldots,h(x,\ell)$ and returns 1 if all
these bits are set to 1, and 0 otherwise. The main advantage of Bloom
Filters is their space-efficiency over other classical data structures
for sets. But this efficiency comes at a price: Bloom Filters may
yield false positives: a membership query may possibly return
true for a value that does not belong to the current set due to hash
collisions. Therefore, an adversary may attempt to pollute the Bloom
Filter in order to trigger false positives~\cite{GerbetKL15}. In this
paper, we consider a very simple form of pollution attacks, where an
adversary adaptively performs insertion queries with the goal to set
to 1 a maximal number of bits of the Bloom Filter. Since the adversary
is probabilistic, we use the \emph{expected} number of bits set to 1
as a measure of the adversary's success. Assuming that the Bloom
Filter is initially empty, i.e.\, all bits are set to 0, one can prove
that for \emph{every} adversary $\cA$ making at most $k$ queries to
the insertion oracle, the expected number of bits set to 1 after the
adversary returns is upper bounded by $m \cdot (1-((m-1)/m)^{\ell\cdot
  k})$.

We model this example in a simply typed calculus enriched with graded
monadic type constructors. Concretely, we model adversaries carrying a
pollution attack against a Bloom Filter as computations $\cA$ of type
$ \forall \alpha.  ([n]\to \mon{\alpha,1}{\tunit}) \to \mon{\alpha
  \cup \Sigma, k}{\tunit}$ where $\tunit$ is the unit type and by
abuse of notation we view $[n]$ as a type. The intended argument of
the adversary is the insertion oracle. The monadic type
$\mon{\alpha,1}{\tau}$ should be seen as stateful probabilistic
computations that can read and write to the set of locations $\alpha$
(but not others) and have cost $1$.  Therefore, the grading ensures
that each oracle call has cost 1, and that the adversary can make at
most $k$ calls to the oracle. The universal quantification on $\alpha$
ensures that the adversary can only read and write locations in
$\Sigma$, and that its effect on other memory locations like the
$L[i]$s is only indirect, through calls to its oracle.

We assume that hash functions are implemented as random oracles,
i.e.\, stateful probabilistic functions that lazily sample their
output when queried with a fresh input. The pseudo-code of the
insertion oracle ${\sf insert}$ is deferred to Section~\ref{sec:exl}.
Under this modeling we upper bound the success of pollution attacks
via the judgment:
\[
\bullet \mid \cA: \tau \mid \bullet \mid \bullet \vdash \{ m \cdot (1-((m-1)/m)^{\ell \cdot k})\}~{\cA~{\sf insert}}\colon{\mon{\Sigma\cup\{r,L,h\},k}{\tunit}}~
\left\{\!\left\{\textstyle \sum_{i=0}^{m-1} L[i] \right\}\!\right\}
\]
%
where $\tau\triangleq\forall \alpha.  ([n]\to \mon{\alpha,1}{\tunit})
\to \mon{\alpha \cup \Sigma, k}{\tunit}$, ${\sf insert}$ is the
insertion oracle, and $\{r,L,h\}$ are the memory locations used
by the oracle.  The adversary, represented by the variable $\cA$,
is declared in a special \emph{adversary context}. The other contexts
for grading variables, standard variables and logical assumptions are
empty (the contexts are explained in Sections~\ref{sec:language}
and~\ref{sec:uhol}). The statement on the right hand side of the
turnstile can be seen as a generalized Hoare triple, given by two
assertions (between curly braces) and a program, as in Hoare Type
Theory~\cite{NanevskiMB08}. We have a generic syntax of judgments and
a generic set of generic inference rules, that can later be
instantiated to different notions of assertions and different
interpretations.  For the specific instantiation used here
(Section~\ref{sec:exl}), the assertions are \emph{quantities} -- maps
from states to the non-negative reals -- that are known as the
pre-expectation (the one on the left) and the post-expectation (on the
right), respectively. The interpretation of such a statement is that
the \emph{expected} value of the post-expectation over the output
distribution of the program is upper bounded by the pre-expectation.

We note that pollution attacks are a very simple example. More
advanced attacks are considered by \citet{NaorY19,ClaytonPS19}, who
develop an elaborate theory of Bloom Filters and probabilistic data
structures under adversarial environments.

\paragraph*{PRF/PRP Switching Lemma}
The PRF/PRP Switching Lemma~\cite{Impagliazzo:1989} is a classical tool
in provable security. Let
$\{0,1\}^l$ denote the set of bitstrings of length
$l$.  The lemma states that the probability of a bounded adversary to 
distinguish between a pseudo-random function (PRF) and a pseudo-random
permutation (PRP) is upper bounded by
$k(k+1)/2^{l+1}$
  where $k$ is the maximal number of calls allowed to the adversary. The
  PRF/PRP Switching Lemma is a popular benchmark for computer-aided
  cryptography, so multiple formalizations are available,
  e.g.~\cite{BartheGZ09}.

  We model the adversary as a computation of type
  $\forall \alpha. (\{0,1\}^l \to \mon{\alpha,1}{\{0,1\}^l} )\to
  \mon{\Sigma\cup\alpha, k}{\{0,1\}}$
  where $\Sigma$ models the private memory of the adversary.
  Similar to the case of Bloom filters, we follow a lazy modeling of
  PRF and PRP. The pseudo-code of PRF and PRP is given below:
\begin{align*}
  {\it PRF} (x_1:\{0,1\}^l) & \defeq
  {\sf if}\ x_1 \notin \mathsf{dom}~L_1\ {\sf then}\ \{z_1=
                                        \unif{\{0,1\}^l};\hspace{-1em}&L_1[x_1] :=
                                        z_1\}; \mathsf{return}~L_1[x_1]\\
  {\it PRP} (x_2:\{0,1\}^l) & \defeq  {\sf if}\ x_2 \notin \mathsf{dom}~L_2\ {\sf then}\ 
  \{ z_2 = \unif{\{0,1\}^l \setminus (\mathsf{im}\ L_2)};\hspace{-1em}&L_2[x_2] := z_2\}; \mathsf{return}~L_2[x_2] 
\end{align*}
We show that for every adversary $\cA$ with the aforementioned type,
the statistical distance between $\cA~{\it PRF}$ and $\cA~{\it PRP}$
is upper bounded by $k(k+1)/2^{l+1}$ using an approximate
relational logic (i.e., a logic that can prove approximations rather
than equalities). 
We establish the following judgment:
{\small
\[
\bullet \mid \cA\colon\tau \mid \bullet \mid \bullet \armjnoc{\stl = \str}{\cA~{\it PRF}}{\mon{\Sigma \cup \{ L_1 \},
    k}{\{0,1\}}}{\cA~{\it PRP}}{\mon{\Sigma\cup\{ L_2 \} ,k}{\{0,1\}}}{\stl = \str \wedge
        \vll=\vlr}{k(k+1)/2^{l+1} }
\]
}
where
$\tau\triangleq \forall \alpha. (\{0,1\}^l \to \mon{\alpha,1}{\{0,1\}^l} )\to \mon{\Sigma\cup\alpha, k}{\{0,1\}}$.
This judgment has the following interpretation: if we have two memories
$\stl,\str$ that are equal and we run the computation on the left and the
computation on the right with input memories $\stl$ and $\str$ respectively,
then the output distributions are going to be close, and their statistical
distance is upper bounded by $k(k+1)/2^{l+1}$.
      
Although our proof uses an approximate logic, there is an alternative
proof that uses an exact relational logic, and the Union Bound logic.
The latter proof uses the so-called up-to-bad technique, and defaults
to the union bound logic to prove that the probability of collisions
in a PRF is upper bounded by $k(k+1)/2^{l+1}$. This is
captured in the union bound logic by the judgment:
{\small
\[
\bullet \mid \cA:\tau \mid \bullet \mid \bullet \aumjnoctxt{|dom~L_1|=0}{\cA~{\it
    PRF}}{\mon{\Sigma\cup\{L_1\},k}{\{0,1\}}}{
|dom~L_1|=|im~L_1|}{k(k+1)/2^{l+1}}
\]
}
This specification has the same syntax as the specification of the
pollution attacks, but uses different notions of predicates and a
different intrepretation (but crucially, the same set of inference
rules). The assertions are Boolean predicates, and the interpretation
of this judgment is that if the initial state satisfies
$|dom~L_1|=0$, then the final state satisfies
$|dom~L_1|=|im~L_1|$ with probability $1-\frac{k(k+1)}{2^{l+1}}$. In
other words, the judgment behaves as a Hoare triple that has some
probability of failure.

\section{Language}\label{sec:language}
We consider a core language that models higher-order, stateful,
probabilistic computations over algebraic datatypes.

\paragraph*{Syntax} The language combines the usual constructs of
$\lambda$-calculus and monadic constructs. Monadic computations are
introduced and composed by unit and let. In addition, we have
operations for sampling from a distribution in a set $\cD$ of base
distributions, and for reading or writing at a location
in a set $\loc$ of locations. We also consider a primitive $\sskip$
operation that represents an empty computation (we could also define
$\sskip$ as $\unit{*}$, where $*$ is the sole inhabitant of the unit
type), and a primitive ${\sf mfold}$ for nesting monadic computations
(the reason why we make ${\sf mfold}$ monadic will become apparent in
the next paragraph, when typing is considered). Finally, for technical
reasons that will become apparent when defining the logic, we
distinguish between adversarial variables and standard
variables. Formally, the terms of the language are given by the
following grammar:
\begin{align*}
  t, u ::=~ & x \mid \cA \mid * \mid 0 \mid S~u \mid \lambda x. u \mid
         t\ u
         \mid \pair{t}{u} \mid
         \casebool{t}{u_1}{u_2} \mid
         \prl{t} \mid \prr{t} \mid \\
         & \rd{a} \mid \wrt{a}{u} \mid \sskip \mid \unit{t} \mid \mlet{x}{t}{u} \mid \mfold{t}{u_1}{u_2} \mid
         \sample{\nu}
\end{align*}
where $x$ ranges over variables, $\cA$ ranges over adversary
variables, $a$ ranges over a set $\loc$ of memory locations and $\nu$
ranges over a set $\cD$ of distribution symbols.
We assume that each distribution has arity
$\tau_{\nu,1} \times \dots \times \tau_{\nu,|\nu|} \to \sigma_{\nu}$,
that accounts for the parameters of the distribution.
The meaning of the expressions is standard, except for the monadic fold operation for
naturals, which sequences computations in the monadic step, that is:
\begin{align*}
    \mfold{0}{t}{u} = t \qquad
    \mfold{(S~n)}{t}{u} = \mlet{x}{(\mfold{n}{t}{u})}{u~x}
\end{align*}

\paragraph*{Syntactic sugar}
In our examples we use some syntactic sugar to simplify the code. Concretely, we
will write $x = t; u$ instead of $(\lambda x. u)~t$,
$l := t$ instead of ${\sf let}~\_~{\sf in}~\wrt{l}{t}$ (i.e., we do not bind the returned value)
and ${\sf inc}~l$ instead of ${\sf let}~y=\rd{l}~{\sf in}~\wrt{l}{y+1}$, where we assume $y$
is a free variable.

\paragraph*{Effects}
We use a type-and-simple effect system to model the memory footprint
and oracle complexity of computations.  We model the memory footprint
as (an overapproximation of) the set $\Sigma$ of memory locations read
and written by a computation. In addition, our effect system supports
abstract effects and effect polymorphism. These are used essentially
to model adversaries. The grading $k$ tracks how many times an
adversary calls its oracles. For simplicity, we use a single natural
number for tracking oracle calls; however, it is possible to track
oracle calls more finely by having a number per oracle.  Semantically,
effects form an ordered commutative monoid: memory effects are modeled
using $(\PP(\loc\cup \mathcal{R}), \emptyset, \cup, \subseteq)$, where
$\mathcal{R}$ is a set of memory regions, $\PP$ is the powerset operator,
and cost is modeled using $(\NN, 0, +, \leq)$.



\paragraph*{Types}
Our language is essentially simply typed. As base types we consider the
unit type $\tunit$, booleans $\tbool$, and natural numbers, which are
indexed by either a constant natural number $K$ or by infinity, to
indicate an upper bound on the inhabitants of the type. We will simply
write $\nat$ for $\nat[\infty]$. 
On top of this we consider extended computations, which are given a
type $\mon{\Sigma,k}{\tau}$.  Here, $\tau$ is the return type, and
${\Sigma, k}$ is a grading that accounts for the memory effect and
cost of the computation. We assume all locations in memory contain
the same type $\tval$. We keep this abstract in
the current presentation, but we will instantiate it to a concrete type (e.g. $\NN,\BB,\dots$)
in the examples. Finally, we include a type $\tmem$ for memories.
These cannot be explicitly manipulated in the language, but are used
in specifications, see later in the section.

Formally, the set of types is given by the following syntax:
\begin{align*}
  \tau, \sigma & ::= B\mid \tbool \mid \nat[K] \mid \tmem\mid \tunit \mid \tval \mid \tau \to \sigma \mid \tau \times \sigma
\mid \mon{\Sigma,k}{\tau}
\mid \forall \alpha. \tau 
\end{align*}
where $B$ ranges over a set of base types, $K$ ranges over natural numbers and
the expression $\Sigma$ is built from region variables and memory locations.
Note that bounded natural types $\nat[K]$ are used to compute the grading of
monadic folds.

\begin{figure*}[!t]
\small
\begin{gather*}
\inferrule*[right=\sf Star]
{ }
{ \Xi;\Delta; \Gamma \vdash \star \colon \tunit}
\qquad
\inferrule*[right=\sf Zero]
{ }
{ \Xi;\Delta; \Gamma \vdash 0 \colon \NN[0]}
\qquad
\inferrule*[right=\sf Succ]
{ \Xi;\Delta; \Gamma \vdash t \colon \NN[K]}
{ \Xi;\Delta; \Gamma \vdash S~t \colon \NN[K+1]}
\qquad
\\
\inferrule*[right =\sf Fold]
{\Delta; \Gamma \vdash t : \NN[K] \\ \Delta; \Gamma \vdash u_1 : \mon{\Sigma,k}{\tau} \\ \Delta; \Gamma \vdash u_2 : \tau \to \mon{\Sigma',k'}{\tau}}
  {\Delta; \Gamma \vdash \mfold{t}{u_1}{u_2} : \mon{\Sigma\cup\Sigma',k+K\cdot k'}{\tau}}
\\
\inferrule*[right =\sf Unit]
  {\Xi;\Delta; \Gamma \vdash t : \tau}
  {\Xi;\Delta; \Gamma \vdash \unit{t}: \mon{\emptyset,0}{\tau}}
\qquad
\inferrule*[right =\sf Bind]
{\Xi;\Delta; \Gamma \vdash t_1 : \mon{\Sigma_1,k_1}{\tau} \\ 
\Xi;\Delta; \Gamma, x : \tau \vdash t_2 : \mon{\Sigma_2,k_2}{\sigma}}
{\Xi;\Delta; \Gamma \vdash \mlet{x}{t_1}{t_2} : \mon{\Sigma_1 \cup \Sigma_2 \cup \eff{\tau}, k_1+k_2}{\sigma}}
\\
 \inferrule*[right =\sf Read]
  {a \in \loc 
  }
  {\Xi;\Delta; \Gamma \vdash \rd{a} : \mon{\{a\},0}{\tval}}
\quad
\inferrule*[right =\sf Write]
{\Xi;\Delta; \Gamma \vdash u:\tval \\  a \in \loc }
 {\Xi;\Delta; \Gamma \vdash \wrt{a}{u} : \mon{\{a\},0}{\tunit}}
 \quad
\inferrule*[right =\sf Skip]
{\ }
{\Xi; \Delta; \Gamma \vdash \sskip : \mon{\emptyset,0}{\tunit}}
 \\
\inferrule*[right =\sf Sample]
{\Xi; \Delta; \Gamma \vdash t_i \colon \tau_{\nu,i} \qquad (\forall 1\leq i\leq |\nu|) }
{\Xi; \Delta; \Gamma \vdash \sample{\nu(t_1,\dots,t_{|\nu|})} : \mon{\emptyset,0}{\sigma_\nu}}
\qquad
\inferrule*[right =\sf Subtype]
{\Xi; \Delta; \Gamma \vdash t : \tau' \\ \Xi \vdash \tau' \preceq \tau }
{\Xi; \Delta; \Gamma \vdash t : \tau}
\\
\inferrule*[right =\sf ForAll-I]
{\Xi, \alpha;\Delta; \Gamma \vdash t : \tau \\ \alpha \not\in FV(\Delta; \Gamma)}
{\Xi;\Delta; \Gamma  \vdash t : \forall \alpha. \tau}
\qquad
\inferrule*[right =\sf ForAll-E]
{\Xi;\Delta; \Gamma \vdash t : \forall\alpha.\tau \\ 
\Sigma \in \PP(\loc) }
{\Xi;\Delta; \Gamma \vdash t : \tau\subst{\alpha}{\Sigma}}
\\
\inferrule*[right =\sf Adv]
{(\cA: \forall \alpha. (\sigma \to
  \mon{\alpha,k}{\tau})  \to \mon{\alpha \cup \Sigma, k'}{\tau'}) \in \Delta
\\ \Xi, \alpha;\Delta; \Gamma \vdash t : \sigma \to
  \mon{\Sigma',k}{\tau}} {\Xi;\Delta; \Gamma \vdash \cA~t:
  \mon{\Sigma \cup \Sigma',
    k'}{\tau'}}
\\
\inferrule*[right =\sf Adv-Inst] {(\cA: \forall \alpha. (\sigma \to
  \mon{\alpha,k}{\tau})  \to \mon{\alpha \cup \Sigma, k'}{\tau'}) \in
  \Delta \\
  \Xi;\Delta ; \Gamma  \vdash t: \mon{\Sigma_1, k_1}{\tau'} \\
  \Xi;\Delta \backslash \cA; \bullet \vdash u: \forall \alpha. (\sigma \to
  \mon{\alpha,k}{\tau})    \to \mon{\alpha \cup \Sigma, k'}{\tau'}}{\Xi;\Delta \backslash \cA;\Gamma \vdash t[u/\cA]: \mon{\Sigma_1, k_1}{\tau'} }
\end{gather*}
\caption{Selected typing rules}\label{fig:typing}
\end{figure*}

\paragraph*{Type system}
A typing judgment $\Xi ; \Delta; \Gamma \vdash t : \tau$ is a relation
between contexts, terms and types.  Contexts are triples of the form
$\Xi ; \Delta; \Gamma$, where $\Xi$ is a grading context, $\Delta$ is
an adversary context, and $\Gamma$ is a variable context. A grading
context $\Xi$ is a collection of variables $\alpha, \beta, \cdots$
representing the memory regions manipulated by the computation.  The
adversary and variable contexts are functions from a finite set of
variables (adversary and standard, respectively).  For a context
$\Gamma$ and $n$ distinct variables $x_i$ such that
$x_i\not\in\dom\Gamma$, by $\Gamma,x_1:\tau_1,\cdots,x_n:\tau_n$ we
mean the context obtained by extending $\Gamma$ with
$x_1:\tau_1\cdots x_n:\tau_n$.  We use a similar notation for
$\Delta$.

Typing rules are presented in Figure~\ref{fig:typing}. Many rules are
standard so we focus on the remaining rules.
The read and the write rules assume that locations store values of
type $\tval$. The effect of a read or write is the location itself.
The unit and bind rules act on the grading as the unit and multiplication
of the monoid from which the gradings are taken, but the bind rule also
adds the effect $\eff{\tau}$ of the type $\tau$ encapsulated by the monad.
This will be important for proving soundness of the adversary rules, and it is
defined as:
\begin{gather*}
    \eff{B}\triangleq \emptyset \qquad \eff{\tau\to\sigma} \triangleq \eff{\sigma}
\qquad \eff{\mon{\Sigma,k}{\tau}} \triangleq \Sigma\cup\eff{\tau} \qquad \eff{\forall \alpha. \tau} = \eff{\tau\subst{\alpha}{\emptyset}}
\end{gather*}

Monadic fold (for natural numbers) defines an iterator. It receives a natural
number bounded by $K$, a computation $u_1$ with cost $k$ for the zero case,
and a computation $u_2$ with cost $k$ for the successor case. The operational
semantics imply that $u_2$ will be run at most $K$ times, and $u_1$ will be
run exactly once, so we can give a bound $K\cdot k' + k$ on the total cost.
The rules for quantifier introduction and elimination are a
lightweight version of effect polymorphism. We can quantify over any
grading in $\Xi$ that does not appear free in $\Gamma$ and $\Delta$,
and we can instantiate a quantifier to any concrete memory
region. Since this does not actually have any computational content,
we choose to not reflect these rules in the term.

The adversary rule allows applying an adversary variable to an
expression with matching type. The instantiation rule for adversary
variables substituting an adversary variable by a \emph{closed}
expression of the same type. This is the only distinction between
standard and adversary variables -- adversary variables represent
closed expressions, while standard variables represent arbitrary
expressions. The reason for making this distinction will become clear
when we describe logics.

Types are ordered by subtyping $\tau' \preceq \tau$, which is used in
the \rname{Subtype} rule. Subtyping is mostly standard. On the type
$\mon{\Sigma ,k}{\tau}$, subtyping allows increasing $k$, $\Sigma$ and
weakening $\tau$. The rules for subtyping are presented in
Figure~\ref{fig:subtyping} in Appendix~\ref{ap:typing-rules}.

\paragraph*{Expressions about memories}

In previous work~\cite{ABGGS17}, the terms appearing in logical assertions and the
terms (i.e., the programs) they specify about are derived from the same grammar.  In the
current setting, program specifications contain distinguished variables
representing the state, because they need to be able to refer explicitly to
initial or final states and their contents, but we do not want programs to have
this capability. Therefore, terms appearing in logical assertions will be
derived from a grammar that extends the grammar of programs: 
\[ \tilde{t},\tilde{u} ::= \ldots \mid \tilde{t}[a] \mid \tilde{t}[a\mapsto \tilde{u}] \] 
Here, $\tilde{t}[a]$ denotes the contents of state $\tilde{t}$ at location $a$, $\tilde{t}[a\mapsto \tilde{u}]$ denotes the state resulting by
replacing the contents of location $a$ in $\tilde{t}$ by $\tilde{u}$, and the ellipsis contains
all the other term constructors.
Figure~\ref{fig:typ-mem}
in Appendix~\ref{ap:typing-rules} presents the (obvious) typing rules for these
new constructs.

\section{Higher-Order Unary Logics}
\label{sec:uhol}


In this section, we describe two program logics for our
language. Both use the same syntactic proof rules derived from
a common template, but they differ significantly in their
semantics and apply to very different verification problems. The first
one is a higher-order Union Bound Logic, in the line
of~\cite{BartheGGHS16}. This logic allows proving postconditions (for
probabilistic computations) that may not hold with an explicit
``error'' probability $\delta$. The second one is a higher-order
expectation logic, in the line
of~\cite{Kozen85,Morgan96,BartheEGHS18}. Instead of specifying
programs with qualitative assertions (that can either be true or false), this logic uses
\emph{quantitative} assertions ranging over the non-negative reals.
The logic can be used to prove bounds on the expected values of quantitative
postconditions.  Both logics have adversaries and state.

\subsection{Higher-order Union Bound Logic}
\label{sec:ubl}

The syntax of our union bound logic (HO-UBL) is shown
below. \emph{Propositions} $\phi, \psi$ are standard (intuitionistic)
HOL formulas over terms. Quantifiers range over these terms.  In
contrast, \emph{assertions} $P,Q$ denote pre- and post-conditions that
relate the language's expressions \emph{and} the current heap state.
$R$ and $f$ denote atomic propositions and atomic assertions,
respectively. $\inj{\phi}$ is an injection of propositions into
assertions. The connectives $\sqcap$ and $\sqcup$ respectively denote
conjunction ($\wedge$) and disjunction ($\vee$) at the level of
assertions.%
\footnote{There is a reason for using different symbols for these
  connectives in propositions and assertions: In the expectation logic
  (Section~\ref{sec:exl}), we want to reuse the same syntax, but give
  assertions \emph{quantitative} interpretations while retaining the
  Boolean interpretations for propositions. Using different symbols
  for the connectives prevents confusion there.}
\footnote{We can add quantifiers $\forall, \exists$ to assertions, but
  we elide them here. Our examples only use \emph{finite}
  quantification in assertions, which can be encoded using $\sqcap$
  and $\sqcup$.}
To refer to the state, assertions $P, Q$ may contain a distinguished
variable $\st$, which stands for the current state. Similarly,
propositions $\phi, \psi$ can contain a distinguished variable $\vl$
that stands for (the value of) the term being verified.

\[\begin{array}{@{}llll}
\mbox{Propositions} & \phi, \psi & ::= & R(t_1, \dots, t_n) \mid \top \mid \bot \mid \phi \wedge \psi \mid \phi \vee \psi \mid \phi \To \psi \mid \neg \phi \mid \forall x : \sigma. \phi \mid
 \exists x:\sigma. \phi \\
	\mbox{Assertions} & P, Q & ::= & f(\tilde{t_1},\dots,\tilde{t_n}) \mid \top \mid \bot \mid \inj{\phi} \mid P \sqcup Q \mid P \sqcap Q \\
 \mbox{Assumptions} & \Psi & ::= & \bullet \mid \Psi, \psi \\ 
 \mbox{Judgments} & \multicolumn{3}{l}{
   \Xi \mid \Delta \mid \jhol{\Gamma}{\Psi}{\phi}
   } \\
 & \multicolumn{3}{l}{
   \Xi \mid \Delta \mid \jhol{\Gamma}{\Psi}{P \jimp Q}
   } \\
 &
 \multicolumn{3}{l}{
   \Xi \mid \Delta \mid \juhol{\Gamma}{\Psi}{t}{\sigma}{\phi}
 }  \\
  & \multicolumn{3}{l}{
   \Xi \mid \Delta \mid \aumj{\Gamma}{\Psi}{P}{t}{\mon{\Sigma,k}{\sigma}}{Q}{\delta}
 } 
\end{array}\]

The logic has four judgments that rely on four contexts---$\Xi$,
$\Delta$ and $\Gamma$ that were described earlier---and the new
context $\Psi$, which contains logical assumptions (propositions)
ranging over the variables in $\Delta$ and $\Gamma$. Since most of our
rules do not modify or read the contexts and $\Xi$ and $\Delta$, we
omit them from most of the discussion below. This simplifies the
judgments, e.g., we write $\jhol{\Gamma}{\Psi}{\phi}$ instead of $\Xi
\mid \Delta \mid \jhol{\Gamma}{\Psi}{\phi}$.

The first judgment $\jhol{\Gamma}{\Psi}{\phi}$ is HOL's standard
entailment judgment. It means that the proposition $\psi$ holds for
all typed instantiations of the variables in $\Gamma$ satisfying all
propositions in $\Psi$. The second judgment $\jhol{\Gamma}{\Psi}{P
  \jimp Q}$ is entailment of assertions; it means that the assertion
$P$ entails the assertion $Q$ for all typed instantiations of $\Gamma$
satisfying $\Psi$.

The third judgment $\juhol{\Gamma}{\Psi}{t}{\sigma}{\phi}$ means that
for all typed instantiations of $\Gamma$ satisfying $\Psi$, the term
$t$ (of type $\sigma$) satisfies $\phi\subst{\vl}{t}$. (Recall that
$\vl$ is a distinguished variable.) In other words, the judgment
specifies a property $\phi$ of the term $t$ being verified. This
judgment's proof rules are directed by the syntax of $t$ and are taken
as-is from the prior logic UHOL~\cite{ABGGS17}. The work on UHOL also
shows that these syntax-directed rules are sound and complete relative
to HOL: $\juhol{\Gamma}{\Psi}{t}{\sigma}{\phi}$ iff
$\jhol{\Gamma}{\Psi}{\phi\subst{\vl}{t}}$. We reproduce these rules in
Appendix~\ref{ap:proof-rules}.

The fourth judgment
$\aumj{\Gamma}{\Psi}{P}{t}{\mon{\Sigma,k}{\sigma}}{Q}{\delta}$ is new
to our logic. It specifies a pre-condition $P$ and a post-condition
$Q$ for a \emph{monadic computation} $t$ of type
$\mon{\Sigma,k}{\sigma}$. Recall that a monadic computation of this
type is stateful and probabilistic --- it takes a state and produces a
\emph{distribution} on results of type $\sigma$ and final states. The
judgment means that, for any instantiation of $\Gamma$ satisfying all
propositions in $\Psi$, starting the execution of $t$ in any state $m$
that satisfies $P\subst{\st}{m}$, the final state $m'$ and result $t'$
(of type $\sigma$) satisfy $Q\subst{\st}{m'}\subst{\vl}{t'}$ \emph{with
  probability at least} $1-\delta$. Additionally, only locations in
the set $\Sigma$ are modified. In other words, the judgment represents
a standard Hoare-triple for stateful computations, but with a small
twist: the postcondition may not hold with an error probability
$\delta$. The semantics of this judgment is defined by lifting
standard Hoare triples to distributions (Section~\ref{sec:st-gr-lift}).

Formally, the pre-condition $P$ can contain the free variable $\st:
\tmem$ (where $\tmem$ is the type of memories), while $Q$ can contain
the variables $\st : \tmem$ and $\vl: \sigma$. Additionally, both may
mention variables from $\Gamma$ and the elided context $\Delta$. The
restriction that only locations in $\Sigma$ be modified during $t$'s
reduction is needed for handling adversaries as we explain soon. The
judgment also does not make use of the grade $k$ in the type
$\mon{\Sigma,k}{\sigma}$; this grade is also used for handling
adversaries.

\begin{figure*}[!htb]
\small
\begin{gather*}
%
\inferrule*[right =\sf\scriptsize MLET-U]
{\aumj{\Gamma}{\Psi}{P}{t}{\mon{\Sigma,k}{\tau}}{Q}{\delta} \\
\aumj{\Gamma, x : \tau}{\Psi}{Q\subst{\vl}{x}}{u}{\mon{\Sigma',k'}{\sigma}}{R}{\delta'} \\  x \not\in R }
{\aumj{\Gamma}{\Psi}{P}
{\mlet{x}{t}{u}}{\mon{\Sigma\cup\Sigma'\cup\eff{\tau}, k+k'}{\sigma}}{R}{\delta+\delta'}}
\\[0.3em]
 \inferrule*[right =\sf\scriptsize UNIT-U]
  {\uj{\Gamma}{\Psi}{t}{\tau}{\phi} \\
  \holwf{\Gamma, \st : \tmem}{P}}
  {\aumj{\Gamma}{\Psi}{P}
  {\unit{t}}{\mon{\emptyset,0}{\tau}}{\inj{\phi} \sqcap P}{0}}
\qquad
\inferrule*[right =\sf\scriptsize READ-U]
{~}
{\aumj{\Gamma}{\Psi}{P\subst{\vl}{\st[a]}}
{\rd{a}}{\mon{\{a\},0}{\tval}}{P}{0}}
\\[0.3 em]
\inferrule*[right =\sf\scriptsize WRITE-U]
{
 \Xi;\Gamma \vdash t : \tval  }
{\aumj{\Gamma}{\Psi}{P\subst{\st}{\st[a\mapsto t]}}
{\wrt{a}{t}}{\mon{\{a\},0}{\tunit}}{P}{0}}
\\[0.3em]
\inferrule*[right =\sf\scriptsize MCASE-U]
{\Gamma \vdash b : \bool \\
    \aumj{\Gamma}{\Psi, b=\ttrue}{P_1}{t_1}{\mon{\Sigma,k}{\tau}}{Q}{\delta}\\
\aumj{\Gamma}{\Psi,  b=\ffalse}{P_2}{t_2}{\mon{\Sigma,k}{\tau}}{Q}{\delta}}
{\aumj{\Gamma}{\Psi}{\inj{b=\ttrue} \sqcap P_1) \sqcup (\inj{b=\ffalse} \sqcap P_2)}
{\casebool{b}{t_1}{t_2}}{\mon{\Sigma,k}{\tau}}{Q}{\delta}}
\\[0.3em]
\inferrule*[right =\sf\scriptsize MFOLD-U]
{\Gamma \vdash n : \nat[K] \\
    \aumj{\Gamma}{\Psi \wedge n=0}{P}{t_1}{\mon{\Sigma,k}{\tau}}{Q}{\delta}\\
\aumj{\Gamma, x\colon \tau}{\Psi \wedge n \neq 0}{Q}{t_2}{\mon{\Sigma',k'}{\tau}}{Q}{\delta'}}
{\aumj{\Gamma}{\Psi}{P}
{\mfold{n}{t_1}{(\lambda x. t_2)}}{\mon{\Sigma\cup\Sigma',k + K\cdot k'}{\tau}}{Q}{\delta + K\cdot \delta'}}
\end{gather*}
\caption{Monadic proof rules of our higher-order union bound
  logic. These rules are reused for the higher-order expectation
  logic with a different interpretation of $\inj{\phi}$, $\sqcap$, and
  $\sqcup$.}\label{fig:muhol}
\end{figure*}
\paragraph{Monadic rules}
Figure~\ref{fig:muhol} presents the main rules of the fourth
judgment. As before, we omit the contexts $\Xi$ and $\Delta$; these
transfer unchanged from the conclusion to the premises in all
rules. All our rules are syntax-directed. The rule \textsf{UNIT-U}
applies to the term $\unit{t}$, which returns the term $t$ without
modifying the state with probability $1$. The rule just restates this
differently. Formally, if $t$ satisfies $\phi$ (first premise), then
executing $\unit{t}$ from a state satisfying $P$ results in a state
and return term that satisfy $\inj{\phi} \sqcap P$. The probability of
this \emph{not} happening is $0$. (The premise $\holwf{\Gamma, \st :
  \tmem}{P}$ just means that the assertion $P$ is a well-formed
predicate over the typed variables in $\Gamma, \st:\tmem$.)

The rule \textsf{MLET-U} for the monadic bind is a generalization of
the usual sequencing rule of Hoare logic. The error probabilities
$\delta$ and $\delta'$ are summed in the conclusion. This is easy to
see: From the first premise, with probability at least $1-\delta$, the
postcondition $Q$ of $t$ holds and, hence, from the second premise,
with probability at least $(1-\delta)-\delta'$, the postcondition $R$
holds. Hence, the error probability is at most
$\delta+\delta'$.

The rules \textsf{READ-U} and \textsf{WRITE-U} propagate heap changes
backwards, as in standard Hoare logic. We also have the rule
\textsf{MCASE-U} for conditionals of monadic type. Again, this rule
follows the rule for conditionals in Hoare logic.
The rule \textsf{MFOLD-U} applies to $\mfold{n}{t_1}{(\lambda
  x. t_2)}$. Here, $K$ is a bound on the number of iterations (first
premise). The error probability in the conclusion is the error
probability $k'$ of the iteration's body scaled by $K$. 

\begin{figure*}
\small
\begin{gather*}
\inferrule*[right =\sf\scriptsize CONSEQ-U]
{\jhol{\Gamma}{\Psi}{P \jimp P'}\\
\aumj{\Gamma}{\Psi}{P'}
{t}{\mon{\Sigma,k}{\tau}}{Q'}{\delta'}\\
\jhol{\Gamma}{\Psi}{Q' \jimp Q}\\
    \delta'\leq\delta
}
{\aumj{\Gamma}{\Psi}{P}
{t}{\mon{\Sigma,k}{\tau}}{Q}{\delta}}
\\[0.3em]
\inferrule*[right =\sf\scriptsize OR-PRE-U]
{\aumj{\Gamma}{\Psi}{P}
{t}{\mon{\Sigma,k}{\tau}}{Q}{\delta} \\
\aumj{\Gamma}{\Psi}{P'}
{t}{\mon{\Sigma,k}{\tau}}{Q}{\delta}
}
{\aumj{\Gamma}{\Psi}{P \sqcup P'}
{t}{\mon{\Sigma,k}{\tau}}{Q}{\delta}}
\\[0.3em]
\inferrule*[right =\sf\scriptsize AND-POST-U]
{\aumj{\Gamma}{\Psi}{P}
{t}{\mon{\Sigma,k}{\tau}}{Q}{\delta}\\
\aumj{\Gamma}{\Psi}{P}
{t}{\mon{\Sigma,k}{\tau}}{Q'}{\delta'}
}
{\aumj{\Gamma}{\Psi}{P}
{t}{\mon{\Sigma,k}{\tau}}{Q \sqcap Q'}{\delta + \delta'}}
\end{gather*}
\vspace{-0.4cm}
\caption{Selected structural rules of the higher-order union bound
  logic. These rules are reused for the higher-order expectation logic
  with a different interpretation of $\sqcap$, $\sqcup$ and
  $\jimp$.}\label{fig:str-muhol}
\end{figure*}

\paragraph{Structural rules}
Figure~\ref{fig:str-muhol} shows selected structural rules of our
logic. The rule of consequence, \textsf{CONSEQ-U}, allows weakening
postconditions and error probabilities, and strengthening
preconditions. The rule \textsf{OR-PRE-U} allows case analysis in the
precondition. Finally, the rule \textsf{AND-POST-U} allows splitting a
conjunction in the postcondition. Note that in this case, the error
probability $\delta+\delta'$ is the sum of the error probabilities of
the two conjuncts. This sum is a standard union bound on (error)
probabilities, which explains the name of our logic.

\paragraph{Rules for monadic primitives}
Additionally, we include rules for monadic primitives that we use in
examples. For instance, the following rule is used for typing the term
$\unif{\sigma}$, which samples a value uniformly from the
\emph{finite} type $\sigma$. The sampling does not change the state,
so the rule copies the precondition $P$ to the postcondition. The
sampled value additionally satisfies any predicate $\phi$ of
cardinality $N$ with probability $1 - N/|\sigma|$.
\[\small
\inferrule*[right =\sf\scriptsize SAMPLE-UBL]
    { \Gamma \mid \Psi \vdash P \Rrightarrow (|\{x\in \sigma \mid x \in \phi \}|\,/\,|\sigma|\, =\, \delta)  }
  {\aumj{\Gamma}{\Psi}{P}
  {\unif{\sigma}}{\mon{\emptyset,0}{\sigma}}{P\wedge \phi}{1-\delta}}
\]


\subsection*{Adversary rule}
In security applications, one often wants to prove properties of
``adversarial'' code, of which very little is known
statically. Typically, one may know or assume that the adversarial
code is closed, has a specific simple type and that it has a certain
bounded complexity, but not much else. Verification of such
unknown code, unsurprisingly, relies on parametricity properties
of the language. To this end, we need proof rules that internalize
parametric reasoning into the logic. Below we show
one such rule, \textsf{ADV-U}, which suffices for our examples. We
first explain the rule informally and then give more formal details:
\[\small
  \inferrule*
      [right= \sf\scriptsize ADV-U]
      {
        (\cA :\forall \alpha. (\sigma \to \mon{\alpha,1}{\tau}) \to \mon{\Sigma \cup \alpha, k}{\tau'}) \in \Delta \\
        \Delta \mid \aumj{ x : \sigma}{\Psi}{P}{t}{\mon{\Sigma',1}{\tau}}{P}{\delta} \quad(x \notin \Psi, P)\\
        \Delta \mid \holwf{\st:\tmem}{P} \\
        P \in {\sf Safe}(\Sigma) \\
        \sigma,\tau,\tau'\ \text{non-monadic types}
      }
      {
        \Delta \mid \aumj{\bullet}{\Psi}{P}{\cA\ (\lambda x. t)}{\mon{\Sigma\cup\Sigma',k}{\tau}}{P}{k\cdot\delta}
}
\]

\paragraph{Informal explanation}
Informally, \textsf{ADV-U} says the following. Suppose:
\begin{itemize}
\item[-] $\cA$ is an arbitrary (adversarial) \emph{closed second-order
  program} whose side effects are limited to the locations in
  $\Sigma$, and that uses its argument at most $k$ times (first
  premise),
\item[-] $\lambda x.t$ is an argument for $\cA$ such that $t$
  \emph{preserves} the assertion $P$ on memories, except with
  probability $\delta$ (second premise), and
\item[-] $P$ does not depend on the values in any locations in
  $\Sigma$ (premise $P \in {\sf Safe}(\Sigma)$, which is defined
  formally later).
\end{itemize}
Then, $\cA$ applied to $(\lambda x.t)$ preserves the assertion $P$
except with probability $k \cdot \delta$.

We can easily see why the conclusion holds. First, since $P$ depends
only on values of locations outside $\Sigma$, to violate $P$, $\cA$
must modify locations outside $\Sigma$. Next, the only way $\cA$ can
even hope to modify variables outside $\Sigma$ is by invoking its
argument $(\lambda x.t)$. This is because $\cA$'s own effects are
limited to $\Sigma$ and it is closed, so it cannot get access to other
effects due to additional substitutions. Hence, the only way for $\cA$
to violate $P$ is by invoking $\lambda x.t$. However, $t$ violates $P$
with probability at most $\delta$ and $\cA$ cannot apply $\lambda x.t$
more than $k$ times. Hence, by a straightforward union bound, $\cA$'s
chances of violating $P$ are bounded by $k \cdot \delta$, which is
exactly the conclusion.

The remarkable aspect of the rule is how little it assumes about the
adversarial expression $\cA$ -- just that $\cA$ is closed, that it
uses its argument at most $k$ times and that its side-effects are
limited to $\Sigma$. The derived conclusion -- that $P$ is preserved
except with probability $k \cdot \delta$ -- holds for \emph{any}
closed, simply typed substitution for the variable $\cA$. This is what
makes this rule very powerful and useful. For example, in security
applications, $\cA$ can model an arbitrary, unknown ``adversary'' of
bounded complexity $k$ that is given a known ``oracle'' as
argument. The rule then proves properties of any instance of the
adversary applied to a given oracle ($\lambda x.t$), without having to
verify the adversary.

\paragraph{Formal notes}
The type of $\cA$ (first premise) ensures that its argument incurs an
effect of at least $1$ unit at each use, and that the total effect of
$\cA$ is $k$. Hence, $\cA$ cannot use its argument more than $k$
times. Further, the rule insists that $\cA$ exist in $\Delta$, not
$\Gamma$. This ensures that $\cA$ represents a closed term. Finally,
the quantification over the effect set $\alpha$ ensures that $\cA$
\emph{itself} writes only to the locations in $\Sigma$.

The condition $P \in {\sf Safe}(\Sigma)$ is formally defined as
$\forall m_1\, m_2: \tmem. \; (\forall a \not\in \Sigma.\; m_1[a] =
m_2[a]) \To P\subst{\st}{m_1} \Leftrightarrow P\subst{\st}{m_2}$, and
means that $P$ is independent of the values in locations in $\Sigma$.
The condition that $\sigma$, $\tau$ and $\tau'$ be non-monadic is
a technical simplification: the rule is proven sound using a logical
relation, and terms of non-monadic types trivially inhabit the
relation. The restriction can be lifted by imposing additional
logical conditions on the argument and result value of $t$, as well
as requiring that $P$ must be also safe for the effects in $\tau'$.
Similarly, the restriction to closed adversaries can be lifted by
requiring that every free variable is instantiated to a term that
inhabits the logical relation.

We chose a particular second-order type for adversaries in this paper,
which was the most convenient for our examples. However, the soundness
argument can be used to easily derive adversary rules for different
adversaries, e.g. adversaries that accept multiple oracles or third-order
adversaries, that interact with oracles that receive functions as arguments.

In the examples it is often convenient to use a mild extension of the rule, where the
invariant $P_i$ and the error bound $\delta_i$ depend on some natural number
$i$ and we show that (1) each oracle call with precondition $P_i$ satisfies the
postcondition $P_{i+1}$ with error probability $\delta_i$ and (2) every $P_i$
implies $P_{i+1}$, from which we deduce that the adversarial computation
satisfies the postcondition $P_k$ with error probability $\delta_1 + \dots +
\delta_k$. To avoid cluttering the notation, here and in Section~\ref{sec:bool-rhol}
we present the rules without this indexing.


\subsection*{Example: Probability of collisions}
We now exercise our proof system to upper bound the probability of
collisions for all adversaries making at most $k$ queries to a PRF.
Recall that our goal is to prove
\[
\cA\colon \forall \alpha. (\{0,1\}^l {\to} \mon{\alpha,1}{\{0,1\}^l} ) {\to}
\mon{\Sigma\cup\alpha, k}{\{0,1\}} \aumjnoctxt{Empty}{\cA~{\it
    PRF}}{\mon{\Sigma,k}{\{0,1\}}}{\Phi_k}{k(k+1)/2^{l+1}}
\]
where $\Phi_i \triangleq |dom(L)| = |im(L)| \land |dom(L)|\leq i$.
 
Applying the proof rule for adversaries, it suffices to
prove that the $i$-th call of the oracle preserves the assertion
$NoColl$ with error probability at most $i/2^l$. Here we use
the fact that at the $i$-th iteration the domain of $L$ has size at
most $i$. So we have to prove:
{
\small
\[
x\colon \{0,1\}^l \aumjnoctxt{|dom(L)| = |im(L)| \land |dom(L)|\leq i}{e}{
    \mon{\alpha,1}{\{0,1\}^l}}{|dom(L)| = |im(L)| \land | dom(L) | \leq
  i+1}{i/2^l}
\]
}where $e$ is the body of $PRF$.  We work our way backwards starting
with $P_{i+1} \triangleq \{\!\{ |dom(L)| = |im(L)| \land | dom(L) |
\leq i+1 \}\!\}_{i/2^l}$ from the end of the program and compute the
precondition of each statement. Note that we keep the grading because
every precondition is the postcondition of the previous statement. The
last instruction is a return, which we can skip since our assertion
does not mention the return value.

Now we encounter the case split. The else
branch is empty, so its precondition is still $P_{i+1}$. On the then branch we
start by strengthening the postcondition to $\{\!\{ x\not\in dom(L) \wedge P_{i+1} \}\!\}_{i/2^l}$.
Then we have the assignment $L_1[x_1] = z_1$, whose precondition is
\[ \{\!\{ x_1 \not\in dom(L) \wedge |dom(L) \cup \{x_1\}| = |im(L) \cup \{z_1\}| \wedge |dom(L) \cup \{x_1\}| \leq i+1 \}\!\}_{i/2^l} \]
Now we can strengthen this to
\[  \{\!\{ z_1 \not\in im(L) \wedge x_1 \not\in dom(L) \wedge |dom(L) \cup \{x_1\}| = |im(L) \cup \{z_1\}| \wedge |dom(L) \cup \{x_1\}| \leq i+1 \}\!\}_{i/2^l} \]
which is equivalent to
\[  \{\!\{ z_1 \not\in im(L) \wedge x_1 \not\in dom(L) \wedge |dom(L)| + 1  = |im(L)| + 1 \wedge |dom(L)| + 1  \leq i+1 \}\!\}_{i/2^l} \]
and by the \textsf{SAMPLE-UBL} rule, we know that the probability of sampling
something outside $im(L)$ is at least $1-i/(2^l)$, so the precondition of this is
\[  \{\!\{ x_1 \not\in dom(L) \wedge |dom(L)| + 1 = |im(L)| + 1 \wedge |dom(L)| + 1 \leq i+1 \}\!\}_{0} \]
This is the precondition of the then branch. By the \textsf{MCASE-U} rule, the
precondition of the whole case construct is
\[  \left\{\!\left\{ \begin{array}{c}
            (x\not\in dom(L) \wedge (x\not\in dom(L) \wedge |dom(L)| + 1 = |im(L)| + 1 \wedge |dom(L)| + 1 \leq i+1)) \vee \\
            (x\in dom(L) \wedge (|dom(L)| = |im(L)| \land | dom(L) | \leq i+1)) 
\end{array}\right\}\!\right\}_{0}. \]
By strengthening, we finally get $\left\{\!\left\{ (|dom(L)| = |im(L)|
\land | dom(L) | \leq i)) \right\}\!\right\}_{0}$, which is exactly
$P_i$.

\subsection{Higher-order expectation logic}
\label{sec:exl}
Our second logic (HO-EXP) is a \emph{quantitative} (non-Boolean) higher-order
expectation logic that proves upper bounds on expected values of
functions of program results and final memories. The logic extends
expectation calculi~\cite{Morgan96,KaminskiKMO16} to the higher-order
setting.

This logic is \emph{syntactically} very similar to the higher-order
union bound logic of Section~\ref{sec:ubl} in the formulas, the
judgments and most of the proof rules, but it is very different in the
interpretation of \emph{assertions} $P, Q$. Specifically, assertions
in this logic are non-negative real-valued functions of their free
variables ($\Delta, \Gamma, \vl, \st$). The assertion connectives
$\sqcap$ and $\sqcup$ are the pointwise supremum and infimum operators
on such functions, as defined below. (Propositions $\phi, \psi$ still
have Boolean interpretations, as in the union bound logic.)

To upper bound expected values, the monadic judgment
$\aumj{\Gamma}{\Psi}{P}{t}{\mon{\Sigma,k}{\sigma}}{Q}{\delta}$ is
interpreted quantitatively, in terms of expectations.\footnote{As in
  Section~\ref{sec:ubl}, the contexts $\Xi,\Delta$ also exist but are
  elided from most of the presentation for brevity.} Specifically, the
inner judgment $\{P\}t : \mon{\Sigma,k}{\sigma}
\{\!\{Q\}\!\}_{\delta}$ means that for every state $m: \tmem$, if we
run $t$ from state $m$, then the \emph{expected value of $Q$} over all
possible final states is upper-bounded by $P\subst{\st}{m} + \delta$.
That is, $\EE_{(m',t') \sim t(m)}[Q\subst{\st}{m'}\subst{\vl}{t'}]
\leq P\subst{\st}{m}+\delta$. The whole judgment
$\aumj{\Gamma}{\Psi}{P}{t}{\mon{\Sigma,k}{\sigma}}{Q}{\delta}$ means
that this inequality holds for all substitutions for $\Gamma$ that
satisfy $\Psi$. Again, the formal semantics of this judgment is
defined by a lifting.
Conventionally, $P$ and $Q$ are respectively called the
\emph{pre-expectation} and the \emph{post-expectation} of the term
$t$. Informally, the judgment $\{P\}t : \mon{\Sigma,k}{\sigma}
\{\!\{Q\}\!\}_{\delta}$ means that the expected value of the
post-expectation is upper-bounded by the pre-expectation plus an error
$\delta$.

\paragraph{Syntax}
The logic reuses the syntax of the union bound logic
(Section~\ref{sec:ubl}), but we extend assertions with some
connectives that are specific to quantities. These new connectives are
shown in \mbox{\color{blue}\bf blue-bold} font below.
\[\begin{array}{@{}llll}
 \mbox{Assertions} & P, Q & ::= & f(t_1,\dots,t_n) \mid \top \mid \bot \mid \inj{\phi} \mid P \sqcup Q \mid P \sqcap Q \mid {\color{blue}\bf \injs{\phi}} \mid {\color{blue}\bf P+Q} \mid {\color{blue}\bf k\cdot P}
\end{array}\]
Assertions are quantities ranging over $[0,\infty]$. $f$ denotes a
function with codomain $[0,\infty]$. Assertion connectives have the
following interpretations.
\[\begin{array}{cll@{\qquad}cll}
\sem{\top} & \triangleq & 0 &
\sem{\bot} & \triangleq & \infty \\
\sem{P \sqcup Q} & \triangleq & \inf\{\sem{P}, \sem{Q}\} &
\sem{P \sqcap Q} & \triangleq & \sup\{\sem{P}, \sem{Q}\} \\
\sem{\inj{\phi}} & \triangleq &
\left\lbrace
\begin{array}{lll}
  0 & \phi \mbox{ holds}\\
  \infty & \phi \mbox{ does not hold}
\end{array}
\right. &
\sem{\injs{\phi}} & \triangleq &
\left\lbrace
\begin{array}{lll}
  1 & \phi \mbox{ holds}\\
  0 & \phi \mbox{ does not hold}
\end{array}
\right. \\
\sem{P + Q} & \triangleq & \sem{P} + \sem{Q} &
\sem{k \cdot P} & \triangleq & k \cdot \sem{P}
\end{array}\]

Note that $\bot$ is interpreted as $\infty$, not $0$. Similarly,
$\sqcap$ corresponds to supremum, not infimum. This reversal of the
usual order is due to the fact that we want to prove upper bounds. The
connective $\injs{\phi}$ is also called the Iverson
bracket~\cite{iverson62}.

The judgment
$\aumj{\Gamma}{\Psi}{P}{t}{\mon{\Sigma,k}{\sigma}}{Q}{\delta}$ is
interpreted as explained above.  The judgment $\jhol{\Gamma}{\Psi}{P
  \jimp P'}$ means that $P \geq P'$ for all instantiations of $\Gamma$
that satisfy $\Psi$.

\paragraph{Proof rules.}
The expectation logic reuses the monadic and structural proof rules of
the union bound logic as is (Figures~\ref{fig:muhol}
and~\ref{fig:str-muhol}). However, the rules' meanings are
quantitative and their soundness is completely different. We
illustrate the new meanings of some of these rules by explaining why
they are still sound.

In rule \textsf{UNIT-U}, the premise ensures that $\phi$ holds (for
the term $t$). So, $\inj{\phi}$ equals $0$ semantically, and
$\inj{\phi} \sqcap P$ is equivalent to $P$. Combined with the fact
that $\unit{t}$ returns $t$ and does not modify the memory, both with
probability $1$, the rule is trivially sound. The rule \textsf{MLET-U}
corresponds to the standard composition of random functions. In the
rule \textsf{MCASE-U}, the precondition $\inj{b=\ttrue} \sqcap P_1)
\sqcup (\inj{b=\ffalse} \sqcap P_2)$ in the conclusion is semantically
equal to $P_1$ when $b = \ttrue$ and $P_2$ when $b = \ffalse$. Hence,
the conclusion reduces to either the second or the third premise.%
\footnote{The precondition $\inj{b=\ttrue} \sqcap P_1) \sqcup
  (\inj{b=\ffalse} \sqcap P_2)$ is semantically equivalent to $([
    b=\ttrue ] \cdot P_1) + ([b=\ffalse] \cdot P_2)$. The latter is a
  more conventional way of writing the precondition~\cite{Morgan96},
  but we prefer the former because it shows the correspondence to the
  union bound logic.}

We also have a new structural rule (\textsf{LIN-EXP}) that allows
combining two different Hoare triples for the same program, relying on
the linearity of expectations.
\[\small
\inferrule*[right =\sf\scriptsize LIN-EXP]
           { \aumje{\Gamma}{\Psi}{P_1}
             {t}{\mon{\Sigma,k}\tau}{Q_1}{\delta_1}  \\
             \aumje{\Gamma}{\Psi}{P_2}
                   {t}{\mon{\Sigma,k}\tau}{Q_2}{\delta_2}
           }
           {\aumje{\Gamma}{\Psi}{P_1+P_2}
             {t}{\mon{\Sigma,k}\tau}{Q_1+Q_2}{\delta_1 + \delta_2}}
\]
           

\paragraph{Rules for monadic primitives}
Finally, we include rules for monadic primitives that we use in
examples. For instance, the rule \textsf{UNIF-EXP} below applies to
the term ${\sf Unif}(K)$, which samples from the uniform distribution
over $\{0,1,\dots,K-1\}$. For $U \subseteq \{0, \dots, K-1\}$, a value
sampled from this distribution is in $U$ with probability exactly
${|U|}/(K)$. Hence, the expected value of $\injs{\vl \in U}$ is
exactly ${|U|}/(K)$, which is the pre-expectation.
\[\small
\inferrule*
    [right =\sf\scriptsize UNIF-EXP]
    { U \subseteq \{0, \dots, K-1\} }
    {\aumje{\Gamma}{\Psi}
      {(|U|/K) \cdot P}{{\sf Unif}(K)}{\mon{\emptyset,0}{\NN[K]}}{\injs{\vl \in U} \cdot P}{0}
    }
\]

\paragraph{Adversary rule.}
The expectation logic admits the adversary rule \textsf{ADV-U} of the
union bound logic 
but with a quantitative
definition of the meta-predicate ${\sf Safe}$. Here, $P \in {\sf
  Safe}(\Sigma)$ is defined as $\forall m_1\, m_2: \tmem. \; (\forall
a \not\in \Sigma.\; m_1[a] = m_2[a]) \To P\subst{\st}{m_1} =
P\subst{\st}{m_2}$. With this change to the definition of $\sf Safe$,
the rule is sound for expectations.

\paragraph{Example.}
We consider pollution attacks on Bloom Filters motivated in
Section~\ref{sec:motivation}. We consider an arbitrary adversary $\cA$
with access to the ${\sf insert}$ oracle of a Bloom filter. The goal
of the adversary is to set as many bits in the Bloom filter to $1$ as
possible using $k$ queries to the oracle. We assume that the Bloom
filter is initially empty and, for simplicity, that it uses only one
hash function, i.e., $\ell = 1$ (our proof easily generalizes to any
$\ell$). We model the hash function as a random oracle that is sampled
lazily. The Bloom filter is implemented as a vector of $m$ bits in
locations $L[0], \ldots, L[m-1]$. The inserted elements are from the
set $[n] = \{0,\ldots,n-1\}$, and $h[0] \ldots h[s-1]$ are auxiliary
locations that hold integers. Additionally, we assume a location $r$
that holds a counter. This is a ghost variable to help us in our verification
effort, it is concretely used to make the invariant depend on the
number of previous calls.
Initially, each $L[i]$ is set to $0$, each
$h[i]$ is set to $-1$ and $r$ is set to $0$. The code of the ${\sf
  insert}$ oracle is shown below:
\[\begin{array}{rcl}
    {\sf insert} (x: [n]) & \defeq &
    \mletA{b}{\rd{h[x]}}\\
  && {\sf if}\ b \not= -1\ {\sf then}\\
  &&\qquad \mletA{y}{\unif{m}}{}\\
  &&\qquad \wrt{h[x]}{y};
  \wrt{L[y]}{1};
	{\sf inc}~r\\
  && {\sf else}~ {\sf inc}~r\\
\end{array}
\]

We want to show that the \emph{expected} number of bits any adversary
can set after making $k$ calls to the adversary is upper bounded by $ m(1-\left((m-1)/m\right)^{k})$.
For this, we prove that for any $\cA: \forall \alpha.  ([n]\to
\mon{\alpha,1}{\tunit}) \to {\mon{\alpha \cup \Sigma, k}{\tunit}}$, we
have
\[\aumjnoctxt{F}{\cA \ {\sf insert}}{\mon{\Sigma \cup\{ L, h, r\},k}{\tau}}{F}{0},\]
where the expectation $F$ is defined as
\[ F = \left({\textstyle \sum_{i \in [m]} \st[L[i]]} \right)
    \left((m-1)/m\right)^{k-\st[r]} +
  m\left(1-\left((m-1)/m \right)^{k-\st[r]}\right).
\]
The idea behind this choice of $F$ is that in the initial state where
$r$ and all $L[i]$s are $0$, $F$ equals the upper bound we want (shown
above), and after the execution, when $\st[r] = k$, $F$ equals
$\sum_{i \in [m]} \st[L[i]]$, whose expectation is what we want to
upper bound.

By the adversary rule \textsf{ADV-U} rule, we need to show that ${\sf
  insert}$ preserves $F$, i.e.,
$ \aumjnoctxt{F}{{\sf insert}~x}{\mon{\{L, h, r\},1}{\tau}}{F}{0} $.
We first use the rule \textsf{MLET-U}. Since $\rd{h[x]}$ trivially
preserves $F$, we need to show that the if-then-else preserves $F$. We
use \textsf{CONSEQ-U} to replace the pre-condition's $F$ with the
equivalent $(\inj{(b \not= -1) = \ttrue} \sqcap F) \sqcup
(\inj{(b \not= -1) = \ffalse} \sqcap F)$. Using the rule
\textsf{MCASE-U}, we then need to prove that the ``then'' and ``else''
branches preserve $F$.

The else branch is fairly straightforward. We need to show that
\[\aumjnoctxt{F}{\mletA{c}{\rd r} \ \wrt{r}{c+1}}{\mon{\{L, h , r\},1}{\tau}}{F}{0}.\]
Using the rules
\textsf{MLET-U}, \textsf{READ-U} and \textsf{WRITE-U}, we get
\[ \aumjnoctxt{F\subst{\st}{\st[r \mapsto (\st[r] + 1)]}}{\mletA{c}{\rd r} \ \wrt{r}{c+1}}{\mon{\{L , h , r\},1}{\tau}}{F}{0}.\]
Hence, by \textsf{CONSEQ-U}, it suffices to show that $F \geq
F\subst{\st}{\st[r \mapsto (\st[r] + 1)]} = F\subst{\st[r]}{(\st[r] +
  1)}$, which follows immediately because $F$ is a decreasing function
of $\st[r]$ (this uses the fact that each $\st[L[i]]$ is either $0$ or
$1$).

On the then branch, we take into account the following property of the
uniform distribution:
\begin{equation}\label{eqn:coupon-exp}
    \vdash  \left\{ ((m-1)/m)  \left({ {\textstyle \sum_{i \in [m]} \st[L[i]]} }\right) + 1 \right\} 
    \unif{m} 
    \left\{\!\left\{ 1 + {\textstyle\sum_{i\in ([m] \setminus \{\vl\})}} \st[L[i]]\right\}\!\right\}.
\end{equation}
To prove this property, we first note that the post-expectation is equal to
\[
[\vl \in \{ i \mid \st[L[i]]=1 \}] \cdot \left({\textstyle \sum_{i \in [m]} \st[L[i]]}\right) + [\vl \in \{ i \mid \st[L[i]]=0 \}] \cdot \left(1+{\textstyle \sum_{i \in [m]} \st[L[i]]}\right).
\]
By \textsf{LIN-EXP} and \textsf{UNIF-EXP}, the pre-expectation
of this with respect to $\unif{m}$ is
\[
(1/m) \left({\textstyle \sum_{i \in [m]} \st[L[i]] }\right) \left({\textstyle \sum_{i \in [m]} \st[L[i]]}\right) + (1/m) \left(m-{\textstyle\sum_{i \in [m]} \st[L[i]]} \right) \left(1+{\textstyle\sum_{i \in [m]} \st[L[i]]}\right),\]
which equals the pre-expectation of~(\ref{eqn:coupon-exp}).

We return to the proof of the then branch and reason backwards from
the end of the branch. After going backwards over $\wrt{h[x]}{y};
\wrt{L[y]}{1}; \mletA{c}{\rd r} \ \wrt{r}{c+1}$, our pre-expectation
becomes $F\subst{\st}{\left(\st[r\mapsto \st[r]+1][L[y] \mapsto
    1]\right)}$, which expands to
\[\left(1 + {\textstyle \sum_{i \in ([m] \setminus \{ y \})} \st[L[i]]}\right) \cdot ((m-1)/m)^{k-\st[r]-1} + m (1-((m-1)/m)^{k-\st[r]-1}).\]
Using~(\ref{eqn:coupon-exp}) and linearity to compute the pre-expectation
of the sampling command, which is
\[\left( (m-1)/m \cdot \left( {\textstyle \sum_{i \in [m]} \st[L[i]]}\right) + 1 \right) \cdot ((m-1)/m)^{k-\st[r]-1} + m (1-((m-1)/m)^{k-\st[r]-1}), \]
and by some rearranging of the terms, this is equal to
\[\left( {\textstyle \sum_{i \in [m]} \st[L[i]]}\right) \cdot ((m-1)/m)^{k-\st[r]} + m (1-((m-1)/m)^{k-\st[r]}),\]
which coincides with $F$. This concludes the proof.

\section{Higher-order probabilistic relational logic}
\label{sec:bool-rhol}

In this section, we present a logic (HO-RPL) to reason about relations between two
computations. The syntax of the logic is shown below, where the propositions,
assertions and assumptions have the same meaning as in Section~\ref{sec:ubl}.
Here we note that, although we keep the abstract syntax of assertions,
in this section we consider only their Boolean interpretation, where the connectives
are replaced by their usual Boolean counterparts and $\langle \phi \rangle$ is
equivalent to $\phi$. Researching a quantitative interpretation, where assertions
are interpreted as distances, is an interesting direction for future work.

\[\begin{array}{@{}llll}
\mbox{Propositions} & \phi, \psi & ::= & R(t_1, \dots, t_n) \mid \top \mid \bot \mid \phi \wedge \psi \mid \phi \vee \psi \mid \phi \To \psi \mid \neg \phi \mid \forall x : \sigma. \phi \mid
 \exists x:\sigma. \phi \\
	\mbox{Assertions} & P, Q & ::= & f(\tilde{t_1},\dots,\tilde{t_n}) \mid \top \mid \bot \mid \inj{\phi} \mid P \sqcup Q \mid P \sqcap Q \\
 \mbox{Assumptions} & \Psi & ::= & \bullet \mid \Psi, \psi \\ 
 \mbox{Judgments} & \multicolumn{3}{l}{
   \Xi \mid \Delta \mid \jhol{\Gamma}{\Psi}{\phi}
   } \\
 & \multicolumn{3}{l}{
   \Xi \mid \Delta \mid \jhol{\Gamma}{\Psi}{P \jimp Q}
   } \\
 & \multicolumn{3}{l}{
     \jrhol{\Gamma}{\Psi}{t_1}{\sigma_1}{t_2}{\sigma_2}{\phi}
 }  \\
  & \multicolumn{3}{l}{
      \armj{\Gamma}{\Psi}{P}{t}{\mon{\Sigma_1,k_1}{\sigma_1}}{t_2}{\mon{\Sigma_2,k_2}{\sigma_2}}{Q}{\delta}
 } 
\end{array}\]

We have already explained the first two judgments in previous sections.
The third form of judgment 
$\jrhol{\Gamma}{\Psi}{t_1}{\sigma_1}{t_2}{\sigma_2}{\phi}$ constitutes
the non-monadic fragment of the logic and comes from
RHOL~\cite{ABGGS17}, a logic to prove relational properties of pure
higher-order programs directed by the syntax of the programs. In these
judgments, $\holwf{\Gamma,\resl,\resr}{\phi}$ is a HOL formula depending on two
distinguished variables $\resl,\resr$ that represent the term on the left of
the judgment and the term on the right, respectively.  The interpretation is
given by the equivalence
$\jrhol{\Gamma}{\Psi}{t_1}{\sigma_1}{t_2}{\sigma_2}{\phi} \Leftrightarrow \jhol{\Gamma}{\Psi}{\phi\subst{\resl}{t_1}\subst{\resr}{t_2}}$,
which follows from the relative completeness theorem of RHOL.
We present the rules for RHOL in the Appendix.

The fourth kind of judgments is new to our presentation, and is introduced
to reason about monadic computations. These have the syntax
$\armj{\Gamma}{\Psi}{P}{t_1}{\mon{\Sigma,k}{\sigma_1}}{t_2}{\mon{\Sigma,k}{\sigma_2}}
{Q}{\delta}$ where $P$ is a Boolean-valued assertion (called the pre-condition)
well-formed in the context $\Gamma, \stl : M, \str : M$, and $Q$ is another
Boolean-valued assertion (called the post-condition) well-formed in the context
$\Gamma, \stl : M, \str : M, \vll : \sigma_1, \vlr : \sigma_2$. Here,
the variables $\stl,\str$ refer to (resp.  left or right-side)
memories, and $\vll, \vlr$ refer to (resp. left or right-side) result
values. Here $\delta$ is a quantitative bound taken in an ordered
monoid; in the simplest case, the monoid consists of a single element
$0$. For the particular interpretation presented in this section, we take the
monoid of non-negative reals with addition. Following the convention of RHOL, we assume
that the free variables of $t_1$ and $t_2$ are disjoint.

The semantics of judgments is based on the notion of statistical
distance. For a general $Q$, the meaning of the judgment depends on
the lifting defined in Example~\ref{ex:DP-lifting}. Here we give an
intuition for the case where $Q$ is of the form
$\stl=\str \sqcap \vll=\vlr$, sufficient for our examples.  If we can derive
\[\armj{\Gamma}{\Psi}{P}{t_1}{\mon{\Sigma,k}{\tau}}{t_2}{\mon{\Sigma,k}{\tau}}{\stl=\str \sqcap \vll=\vlr}{\delta}\]
then for every instantiation of $\Gamma$
satisfying $\Psi$, and every pair of initial memories $m_1,m_2\in M$,
such that $m_1,m_2\in P$, the statistical distance between the output
distributions $t_1(m_1)$ and $t_2(m_2)$ is at most $\delta$, i.e.,
for every event $S$ the absolute difference between the
probabilities of $S$ in $t_1(m_1)$ and $t_2(m_2)$ is at most $\delta$.
%
In particular, when $t_2$ is a renaming of $t_1$, $\delta=0$,
and $P$ and $Q$ define partial equivalences on memories, the judgment
enforces a form of generalized non-interference.

\paragraph*{Monadic and structural rules} Figure~\ref{fig:gen-mrhol}
presents selected monadic and structural rules. Following a pattern
that is standard for relational logics, we have 2-sided rules, such as
\rname{UNIT-R}, \rname{MLET-R}, \rname{READ-R}, \rname{WRITE-R} and
\rname{MCASE-R}, where the two expressions have the same top-level
structure, and 1-sided rules, such as \rname{L-UNIT-R} and
\rname{L-MLET-R}, which exclusively consider the top-level construct
of one expression. These generalize their unary counterparts. In
particular, the rule \rname{MCASE-R} has an extra assumption ensuring
that the two computations go to the same branch, so we only need to
prove a relation between the then branches and a relation between the
else branches. A 1-sided rule without this assumption also exists,
allowing to consider the 4 possible pairs of branches, but we do not
show it here.


\paragraph*{Rules for sampling} Our logic also features rules for
reasoning about sampling. In contrast to the other rules, these
are only valid to the particular interpretation based on statistical
distance that we present here.
We show one rule below:
\[\small
\inferrule*[right =\sf SAMPLE-R]
 { B_1 \subseteq B_2\; \text{finite}}
    {\armj{\Gamma}{\Psi}{P}
    {\unif{B_1}}{\mon{\emptyset,0}{B_1}}{\unif{B_2}}{\mon{\emptyset,0}{B_2}}{\langle \vll=\vlr \rangle \sqcap P}{|B_1|/|B_2|}}
 \]
The rule is used to compare uniform
samplings from two finite sets. There exists an alternative rule where
$\delta=0$, at the cost of weakening the postcondition; this rule is
shown in the appendix.

\begin{figure*}[!htb]
\small
\begin{gather*}
  \inferrule*[right =\sf\scriptsize UNIT-R]
  {\rj{\Gamma}{\Psi}{t_1}{\tau_1}{t_2}{\tau_2}{\phi} \\
    \holwf{\Gamma, \stl: M, \str : M}{P}
}{\armj{\Gamma}{\Psi}{P}{\unit{t_1}}{\mon{\emptyset,0}{\tau_1}}{\unit{t_2}}{\mon{\emptyset,0}{\tau_2}}{\langle \phi \rangle \sqcap P)}{0}}
  \\[0.3em]
\inferrule*[right =\sf\scriptsize MLET-R]
{\armj{\Gamma}{\Psi}{P}{t_1}{\mon{\Sigma_1,k_1}{\tau_1}}{t_2}{\mon{\Sigma_2,k_2}{\tau_2}}{Q}{\delta} \\
	\armj{\Gamma, x_1 : \tau_1, x_2 : \tau_2}{\Psi}{Q\subst{\vll}{x_1}\subst{\vlr}{x_2}}
        {u_1}{\mon{\Sigma_1',k_1'}{\sigma_1}}{u_2}{\mon{\Sigma_2',k_2'}{\sigma_2}}{R}{\delta'} \\  x_1,x_2 \not\in R }
{\armj{\Gamma}{\Psi}{P}
{\mlet{x_1}{t_1}{u_1}}{\mon{\Sigma_1\cup\Sigma_1',k_1+k_1'}{\sigma_1}}
{\mlet{x_2}{t_2}{u_2}}{\mon{\Sigma_2\cup\Sigma_2',k_2+k_2'}{\sigma_2}}{R}{\delta+\delta'}}
\\[0.3em]
\inferrule*[right =\sf\scriptsize READ-R]
{\uj{\Xi;\Gamma}{\Psi}{a_1}{\loc}{\psi}\\
\uj{\Xi;\Gamma}{\Psi}{a_2}{\loc}{\psi}}
{\armj{\Gamma}{\Psi}{P\subst{\vll}{\stl[a_1]}\subst{\vlr}{\str[a_2]}}
{\rd{a_1}}{\mon{\{a_1\},0}{\tval}}{\rd{a_2}}{\mon{\{a_2\},0}{\tval}}{P}{0}}
\\[0.3em]
\inferrule*[right =\sf\scriptsize WRITE-R]
{\Xi;\Gamma \vdash a_1 : \loc \\
 \Xi;\Gamma \vdash t_1 : \tval \\
 \Xi;\Gamma \vdash a_2 : \loc \\
 \Xi;\Gamma \vdash t_2 : \tval  }
 {\armj{\Gamma}{\Psi}{P\subst{\stl}{\stl[a_1\mapsto
       t_1]}\subst{\stl}{\stl[a_2\mapsto
       t_2]}}{\wrt{a_1}{t_1}}{\mon{\{a_1\},0}{\tunit}}{\wrt{a_2}{t_2}}{\mon{\{a_2\},0}{\tunit}}{P}{0}}
 \\[0.3em]
\inferrule*[right =\sf\scriptsize MCASE-R]
{\jrhol{\Gamma}{\Psi}{b_1}{\BB}{b_1}{\BB}{b_1 = b_2} \\
    \armj{\Gamma}{\Psi \wedge b_1 = \ttrue}{P_1}{t_1}{\mon{\Sigma_1,k_1}{\tau_1}}{t_2}{\mon{\Sigma_2,k_2}{\tau_2}}{Q}{\delta}\\
\armj{\Gamma}{\Psi \wedge b_1 = \ffalse}{P_2}{u_1}{\mon{\Sigma_1,k_1}{\tau_1}}{u_2}{\mon{\Sigma_2,k_2}{\tau_2}}{Q}{\delta} \\
P \triangleq (\langle b_1 = \ttrue \rangle \sqcap P_1 \rangle ) \sqcup (\langle b_1 = \ffalse \rangle \sqcap P_2 )}
{\Gamma \mid \Psi \vdash
    \{ P \}
    \casebool{b}{t_1}{u_1} \colon \mon{\Sigma_1,k_1}{\tau_1} \sim 
    \casebool{b}{t_2}{u_2} \colon \mon{\Sigma_2,k_2}{\tau_2} 
\{\!\{ Q \}\!\}_{\delta}}
\\[0.3em]
\inferrule*[right =\sf\scriptsize L-UNIT-R]
 {\uj{\Gamma}{\Psi}{t_1}{\tau_1}{\phi} \\
  \holwf{\Gamma, \stl: M, \str : M}{P}}
 {\armj{\Gamma}{\Psi}{P}
     {\unit{t_1}}{\mon{\emptyset,0}{\tau_1}}{\sskip}{\mon{\emptyset,0}{\tunit}}
 {\langle \phi \rangle \sqcap P}{0}}
\qquad
\\[0.3em]
\inferrule*[right =\sf\scriptsize L-MLET-R]
{\armj{\Gamma}{\Psi}{P}{t_1}{\mon{\Sigma_1,k_1}{\tau_1}}{\sskip}{\mon{\emptyset,0}{\tunit}}{Q}{\delta} \\
	\armj{\Gamma, x_1 : \tau_1, x_2 : \tunit}{\Psi}{Q\subst{\vll}{x_1}\subst{\vlr}{x_2}}
{u_1}{\mon{\Sigma_1',k_1'}{\sigma_1}}{u_2}{\mon{\Sigma_2,k_2}{\sigma_2}}{R}{\delta} \\  x_1,x_2 \not\in R }
{\armj{\Gamma}{\Psi}{P}
{\mlet{x_1}{t_1}{u_1}}{\mon{\Sigma_1\cup\Sigma_1',k_1+k_1'}{\sigma_1}}{u_2}{\mon{\Sigma_2,k_2}{\sigma_2}}{R}
{\delta+\delta'}}
 \end{gather*}
 \vspace{-0.4cm}
\caption{Relational logic: monadic rules}\label{fig:gen-mrhol}
\end{figure*}
%

\paragraph*{Adversary rule}
The adversary rule for the relational setting is similar in spirit to
the adversary rule for the unary setting. However, some mild
adjustments are needed. First, we need to modify the notion of safety
for a region $\Sigma$.  In the unary case we only required that
writing to $\Sigma$ preserves the invariant. In the relational case,
we also need to require that $\phi$ is also ``safe for reading in
$\Sigma$'', meaning that an adversary reading from two different
memories related by $\phi$ at the same location in $\Sigma$ sees the
same value. This prevents the two executions from diverging due to a
read operation:
\begin{definition}
Let $\phi$ be a predicate and $\Sigma \subseteq \loc$. We say that
$\phi \in {\sf RSafe}(\Sigma)$ iff
\begin{align*}
\forall m_1, m_2, \forall l\in \Sigma. \forall v\in\tval. \phi(m_1,m_2) &\Rightarrow 
    \phi(m_1[l\mapsto v],m_2[l\mapsto v]) \land m_1[l] = m_2[l]
\end{align*}


\end{definition}
The adversary rule also allows us to show that the outputs must be \emph{extensionally
equal}, which corresponds to the predicate ${\sf Eq}_\tau$ defined below. The reason
to use this as opposed to equality in the model is that the logical relation we use
in the soundness proof corresponds to extensional equality for non-monadic types:
\begin{align*}
    {\sf Eq}_{b}(x_1,x_2) &\triangleq x_1 = x_2  \\
    {\sf Eq}_{\tau_1\to\tau_2}(x_1,x_2) &\triangleq \forall y_1,y_2 \in \tau_1. 
        {\sf Eq}_{\tau_1}(y_1,y_2) \To {\sf Eq}_{\tau_1}(x_1~y_1, x_2~y_2)\\
    {\sf Eq}_{\tau_1\times\tau_2}(x_1,x_2) &\triangleq
        {\sf Eq}_{\tau_1}(\pi_1(x_1),\pi_1(x_2)) \wedge 
        {\sf Eq}_{\tau_2}(\pi_2(x_1),\pi_2(x_2))
\end{align*}
The adversary rule can now be stated below:
  \[\small
\inferrule*[right= \sf\scriptsize ADV-R]
{ (\cA :\forall \alpha. 
    (\sigma \to \mon{\alpha,1}{\tau}) \to \mon{\Sigma \cup \alpha, k}{\tau'})\in\Delta \\
    \holwf{\stl:M,\str:M}{\phi} \\
 \phi \in {\sf RSafe}(\Sigma) \\
 x_1\not\in FV(t_2), x_2\not\in FV(t_1) \\
 \sigma,\tau,\tau'\ \text{non-monadic types}\\\\
 \armj{ x_1 : \sigma, x_2 : \sigma}{{\sf Eq}_{\sigma}(x_1,x_2) }
    {\phi}{t_1}{\mon{\Sigma',1}{\tau}}{t_2}{\mon{\Sigma',1}{\tau}}
    {\langle {\sf Eq}_{\tau}(\vll,\vlr) \rangle \sqcap \phi}{\delta}
}
{\armj{\Delta\mid\cdot}{\cdot}{\phi}
{\cA(\lambda x_1. t_1)}{\mon{\Sigma\cup\Sigma',k}{\tau'}}
{\cA (\lambda x_2. t_2)}{\mon{\Sigma\cup\Sigma',k}{\tau'}}
{\langle{\sf Eq}_{\tau'}(\vll,\vlr)\rangle \sqcap \phi}{k\delta}}
\]
Informally, the premises of the rule state:
\begin{itemize}
\item $\phi$ is a safe for the memory region $\Sigma$;
\item if their inputs are extensionally equal and their initial
  memories are related by $\phi$, then the oracles produce equal
  outputs and final memories related by $\phi$, with error $\delta$;
      
    \item $\cA$ is an arbitrary adversary that only writes to and reads from $\Sigma$
        and that can call its argument up to $k$ times
\end{itemize}
From them, we conclude that executing the adversary with these oracles
and initial memories related by $\phi$ should yield equal values and
output memories related by $\phi$, with error $k\delta$.

\subsection*{Example: PRF/PRP Switching Lemma}
We use our logic to show that the probability that an adversary can
distinguish between a PRF and a PRF on bitstrings of fixed length $l$ is
upper bounded by $\frac{k(k+1)}{2^{l+1}}$, where $k$ is the maximal
number of queries allowed to the adversary.
As before, we consider a mild extension
of the logic where the error bound can depend on the oracle counter.
For readability, we will generally omit from our judgments the effect,
adversary and variable contexts, and drop the cost grading from the
monadic types, and omit all reasoning about the size of the domain of
$L$.  Our goal is to show:
  \[
 \armjnoc
      {\stl=\str }{\cA~{\it PRF}}{\mon{\Sigma\cup \{L\},k }{\{0,1\}}}{\cA~{\it PRP}}{
        \mon{\Sigma \cup \{L\},k}{\{0,1\}}}
      {\vll=\vlr}{k(k+1) / 2^{l+1}}
    \]
  By applying the rule \rname{ADV-R} on the strengthened judgment,
    we are left to prove:
 $$
 \begin{array}{l@{~}l@{~}}
 x_1=x_2 \vdash& \{\stl = \str \}
                e_{\it PRF} \colon \mon{\{L \}}{\{0,1\}}\sim 
                 e_{\it PRP} \colon \mon{\{  L\}}{\{0,1\}} 
     \{\!\{ \stl = \str \sqcap \vll=\vlr \}\!\}_{i/2^{l}}
  \end{array}
$$
where $e_{\it PRF}$ and $e_{\it PRP}$ denote the bodies of the PRF and
PRP oracles. We then apply the \rname{MCASE-R} rule.  In the empty
\emph{else} branch, we need to prove:
\[
  x_1=x_2 \vdash \{\stl=\str \} (\rd L)[x_1]\colon \mon{L}{\{0,1\}^l}
  \sim (\rd L)[x_2] \colon \mon{L}{\{0,1\}^l} \{\!\{\stl = \str
  \sqcap \vll=\vlr \}\!\}_{i/2^{l}}
  \]
  which is a simple application of the \rname{READ-R} rule. In the
  \emph{then} branch, we first apply the \rname{WRITE-R} rule, and
  then we are left to prove:
  \[
  \vdash
      \{\stl = \str \}  \unif{X_1}\colon\mon{}{\{0,1\}^l}\sim 
      \unif{X_2} \colon \mon{L,1}{\{0,1\}^l}
      \{\!\{\stl[x_1\mapsto\vll] = \str[x_2\mapsto\vlr] \sqcap \vll =
        \vlr \}\!\}_{i/2^{l}} 
\]
where $X_1\triangleq \{0,1\}^l$ and $X_2\triangleq \{0,1\}^l \setminus
\mathsf{im}(\vl[L])$. By the rule of consequence, this follows from
\[
  \vdash
      \{\stl = \str \}  \unif{X_1}\colon\mon{}{\{0,1\}^l}\sim 
      \unif{X_2} \colon \mon{L,1}{\{0,1\}^l}
      \{\!\{\stl= \str \sqcap \vll =
        \vlr \}\!\}_{i/2^{l}} 
\]
which we prove using the \rname{SAMPLE-R} rule.

\ifdef{\G}{\renewcommand{\G}{\Gamma}}{\newcommand{\G}{\Gamma}}
\newcommand{\g}{\gamma}

\section{Semantics}\label{sec:semantics}
Now we present the formal semantics for our system. We begin with some
background, and follow with the semantics of the language and the logics,
and their soundness theorems.

\subsection{The Category of Quasi-Borel spaces and Probability Monad}

We will assume knowledge of some concepts of category theory,
such as bi-Cartesian closed categories (bi-CCC) and strong monads,
see e.g.~\cite{maclane71} for details. In any biCCC $\CC$ in this paper, we fix a
terminal object $(1,!_X\in\CC(X,1))$, and for each pair $X,Y\in\CC$ of
objects, we fix a binary product
$(X\times Y,\pi_1,\pi_2,\langle -,-\rangle)$, a binary coproduct
$(X+Y,\iota_1,\iota_2,[-,-])$ and an exponential object
$(X \To Y,\ev,\lambda(-))$. We also equip $\CC$ with
the symmetric monoidal structure $(1,(\times),l,r,a,s)$
induced by the fixed terminal object and binary products.

We will use the category $\QBS$ of quasi-Borel
spaces~\cite{HeunenKSY17,Scibior:2017:DVH:3177123.3158148} for
modeling higher-order probabilistic programs introduced in Section
\ref{sec:language}. The category $\QBS$ is a well-pointed bi-CCC; in
fact it has small products and coproducts.  For modeling
probabilistic choice, we employ the strong monad
$(\mprob, \eta^\mprob,\mu^\mprob,\theta^\mprob)$ for probability
measures over QBSs \cite{HeunenKSY17}.  For a set $A$ and a QBS $X$,
by $(A\cdot X,\{\iota_a^{A,X}:X\to A\cdot X\}_{a\in A},[-]_{a\in A})$ we
mean the coproduct of $A$-many copies of $X$.

We write $|-| \colon \QBS \to \Set$ for the forgetful functor extracting the
carrier set of QBS. It preserves finite (actually small) products. To ease
calculation, we assume $|1|=1$ and $|X\times Y|=|X|\times |Y|$ (rather than
isomorphic). We also assume that the exponential of $\QBS$ is defined so that
$|X\To Y|=\QBS(X,Y)$.
Finally, we write $\qbsposr$ for the QBS of non-negative extended reals.

\subsection{Probabilistic State Monad}

Starting from this base, which was already presented
in~\cite{HeunenKSY17} we use the state monad transformer to construct
a strong monad given by a functor.

First, we introduce the QBS for memory states. Fix a QBS $\qbsval$
corresponding to type $\tval$. The QBS for memory states is a product
$(\qbsmem,\{\pi_a:\qbsmem\to \qbsval\}_{a\in\loc})$ of $\loc$-many
copies of $\qbsval$.
We next introduce a memory update function.  Given a QBS morphism
$f:X\to \qbsval$ computing a value from an environment, we define the
memory update $u_a(f):X\times \qbsmem\to \qbsmem$ (at
location $a\in\loc$) to be the unique morphism satisfying
$\pi_a\circ u_a(f)=f\circ\pi_1$ and
$\pi_{a'}\circ u_a(f)=\pi_{a'}\circ\pi_2$ for any $a'\neq a$.
We then define the probabilistic state monad by
$ \mpstn \triangleq \qbsmem \To \mprob( - \times \qbsmem).  $ The
unit
$  \eta^\mpstn_X  : X \to \mpst X$, multiplication
$  \mu^\mpstn_X : \mpst(\mpst X)\to \mpst X$
and strength
$\theta^\mpstn_{X,Y} : X\times \mpst(Y) \to \mpst(X\times Y)$
of this monad are defined as:
\[
  \eta^\mpstn \triangleq \lambda(\eta^\mprobn_{X\times \qbsmem}) \qquad
  \mu^\mpstn \triangleq \qbsmem\To(\mu^\mprobn\circ\mprob(ev))  \qquad
  \theta^\mpstn_{X,Y} \triangleq \lambda(\mprob\alpha^{-1}\circ\theta^\mprobn \circ (X\times\ev)\circ\alpha).
\]

\subsection{Semantics of the language}

As demonstrated by Moggi, the computational metalanguage (the simply
typed lambda calculus with monadic types) is naturally interpreted in
any CCC with a strong monad.  The semantics of the language in
Section~\ref{sec:language} follows the same pattern. To
accommodate probabilities we take the category $\QBS$ with the
probabilistic state monad $\mpstn$.



The semantics of types is defined as objects in $\QBS$, assuming we
have an object $\sem{b_i}$ for every base type $b_i \in B$. Note that
the indices of the monad and the quantification over regions are
erased at the semantic level (below, for a natural number $K$,
$\bar{K}$ denotes the set $\{0,\cdots,K\}$):
\begin{gather*}
  \sem{\tbool} \triangleq \qbsbool\quad
	\sem{\nat[K]} \triangleq \qbsnat{\bar{K}}\quad
  \sem{\tunit} \triangleq \qbsunit\quad
  \sem{\tval} \triangleq \qbsval\quad
  \sem{\tmem} \triangleq \qbsmem\quad
  \\
  \sem{\sigma\to\tau} \triangleq \sem{\sigma}\To\sem{\tau} \quad
  \sem{\sigma\times\tau} \triangleq \sem{\sigma}\times\sem{\tau}\quad
  \sem{\mon{\Sigma,k}{\sigma}} \triangleq \mpstn(\sem{\sigma}) \quad
  \sem{\forall\alpha. \tau} \triangleq \sem{\tau}
\end{gather*}

This categorical semantics erases the effect annotations $\Sigma,k$
of the monadic type $\mon{\Sigma,k}{\tau}$ and the universal
quantification $\forall\alpha.\tau$ over regions, which only play a role
in proving soundness of the adversary rules. Adversary variables
are placeholders for closed terms, and do not play any special role in the
semantics. We therefore give a semantics of the language without
contexts $\Xi$ and $\Delta$. 

We interpret the subtyping relation $\Xi\vdash\tau\preceq\tau'$ as a
coercion morphism $\coe\tau{\tau'}:\sem\tau\to\sem{\tau'}$.  Most of
its definition is routine, except for the case of natural number
type: for $K\le K'$, $\coe{\nat[K]}{\nat[K']}$ is defined to be the
evident morphism $\qbsnat{\bar K}\to\qbsnat{\bar{K'}}$.

Semantics of a context $\G$ is given by the Cartesian product of the
interpretation of types in $\G$. For convenience, we fix a product
$(\sem\G,\{\pi_x^\G:\sem\G\to\sem{\G(x)}\}_{x\in\dom(\G)})$ for each
context $\G$.  For a context $\G,x_1:\tau_1,\cdots,x_n:\tau_n$, by
$\mixm\G{x_1:\tau_1,\cdots,x_n:\tau_n}:
\sem\G\times(\sem{\tau_1}\times\cdots\times\sem{\tau_n}) \to
\sem{\G,x_1:\tau_1,\cdots,x_n:\tau_n}$ we mean the evident isomorphism
in $\QBS$. Also, for a
well-typed term $\G\vdash t:\tau$ and $x\not\in\dom\G$, we define the
substitution morphism
$\sub{\G\vdash t:\tau}{x}:\sem\G\to\sem{\G,x:\tau}$ to be the
composite
$\mix\G x\tau\circ \langle id_{\sem\G},\sem{\G\vdash t:\tau}\rangle$.

\begin{figure}[t]
\small
  \begin{align*}
    \sem{\G\vdash \unit{t} : \mon{\Sigma,k}{\sigma}}
    &\triangleq
      \eta^\mpstn \circ \sem{\G\vdash t : \sigma} \\
    \sem{\G\vdash \mlet{x}{t}{u} : \mon{\Sigma,k}{\tau}}
    & \triangleq
      \sem{\G,x:\sigma\vdash u : \mon{\Sigma,k}{\tau}}^{\#\mpstn} \circ \mix\G x\sigma\circ \theta^\mpstn \circ \langle id_{\sem\G}, \sem{\G\vdash t : \mon{\Sigma,k}{\sigma}} \rangle \\
    \sem{\G\vdash \rd{a} : \mon{\Sigma,k}{\tval}}
    &\triangleq
      \lambda(\eta^\mprob \circ \langle \pi_a, id \rangle \circ \pi_2)\\
    \sem{\G\vdash \wrt{a}{t} : \mon{\Sigma,k}{\tunit}}
    &\triangleq 
      \lambda(\eta^\mprob \circ \langle !,\id\rangle\circ u_a(\sem{\G\vdash t: \tval}))\\
      \sem{\G\vdash \sample{\nu(t_1,\dots,t_k)} : \mon{\Sigma,k}{\sigma_{\nu}}}
    &\triangleq
      \lambda(\vartheta^\mprob \circ \langle 
      \sem{\nu}\circ 
      \langle \sem{\G\vdash t_1\colon \tau_{\nu,1}},\dots,\sem{\G\vdash t_k\colon \tau_{\nu, |\nu|}}\rangle\circ\pi_1, \pi_2 \rangle)\\
    \sem{\G\vdash\mfold{n}{t_1}{t_2}\colon \mon{\Sigma\cup\Sigma',k + K\cdot k'}{\sigma}}
    &\triangleq \\
    &{} \hspace{-10em}{\it fold}_{\qbsnat K,\sem{\mon{\Sigma,k}{\sigma}}} \circ \langle \sem{\G\vdash n \colon \NN[K]}, \sem{\G\vdash t_1\colon\mon{\Sigma,k}{\sigma}}, {\it Kl}_{\sem\G, \sem{\mon{\Sigma',k'}{\sigma}}} \circ \sem{\G\vdash t_2\colon \sigma\to\mon{\Sigma',k'}{\sigma}} \rangle
  \end{align*}
  \vspace*{-1.5em}
  \caption{Semantics of the language}
  \label{fig:semnon}
\end{figure}
Well-typed terms $\G\vdash t :\sigma$ are
interpreted as a morphism in $\QBS(\sem{\G},\sem{\sigma})$. 
The interpretation of monadic types can be found in Figure
\ref{fig:semnon}; the interpretation of the non-monadic fragment
is standard and deferred to
Appendix \ref{sec:interpb}. In the Figure,
$(-)^{\#\mpstn}$ denotes the Kleisli lifting of $\mpstn$;
${\it Kl}_{X,Y} \colon (X \To \mpstn Y) \to (\mpstn X \To \mpstn Y)$
denotes the internal Kleisli lifting;
$\vartheta^\mprobn_{X,Y} : \mprob(X) \times Y \to \mprob(X\times Y)$
is the co-strength, a transformation analogous to the strength but
with swapped arguments; and
${\it fold}_{\qbsnat K,X} \colon \qbsnat K \times X \times
(X \To X) \to X$ denotes the iterator over the natural
numbers up to $K$. We assume that every distribution $\nu$
with arity
$\tau_{\nu,1}\times\dots\times\tau_{\nu,|\nu|} \to \sigma_{\nu}$ has
an interpretation $\sem{\nu}$ of the proper type
$\sem{\tau_{\nu,1}}\times\dots\times\sem{\tau_{\nu,|\nu|}} \to\mprob
\sem{\sigma_{\nu}}$.  This semantics is sound in the following sense:
\begin{theorem}
  Let $\Xi\mid\Delta\mid\G \vdash t \colon \sigma$ be a well-typed
  term and $\emptyset\vdash t_i:\Delta(\cA_i)$ be closed terms
  given for each $\cA_i\in\dom(\Delta)$.  Then
  $\sem{\G \vdash t[t_i/\cA_i]_{\alpha\in\dom\Delta} \colon
    \sigma} \in \QBS(\sem{\G},\sem{\sigma})$.
\end{theorem}

\subsection{Heyting-Valued Predicates over QBSs}

\renewcommand{\O}{\Omega}
\renewcommand{\S}{{\bf 2}}
\newcommand{\bO}{{\bf\Omega}}
\newcommand{\UPred}[2]{{\bf UP}^{#1}_{#2}}

Formulas and assertions are interpreted in the same way as {\em
  predicates over QBSs}. Recall that a complete Heyting algebra is a
complete lattice $\bO=(\O,\sqsubseteq)$ (whose meet and join are
denoted by $\sqcap$ and $\sqcup$ respectively) together with a {\em
  pseudo-complement} operator $\Rrightarrow$.
\begin{definition}
  Let $\bO=(\O,\sqsubseteq)$ be a complete Heyting algebra.  An
  $\bO$-valued predicate on a QBS $X$ is a function 
  of type
  $|X|\to\O$. Define $\UPred\bO X\triangleq\Set(|X|,\O)$
  to mean the set of $\bO$-valued predicates on $X$.
\end{definition}
By the pointwise order, $\UPred\bO X$ is again a complete Heyting
algebra.  We define $\S\triangleq\{\bot\sqsubseteq\top\}$ to mean the
Sierpinski space complete Heyting algebra. For $x\in |X|$ and
$P\in\UPred\S X$, we say that $x$ satisfies $P$, denoted by
$x\models P$, if and only if $P(x)=\top$.

Every $\S$-valued predicate can be converted into a $\bO$-valued
predicate.  Define $I: \UPred\S X\to\UPred\bO X$ by $I(P)(x)=\top_\bO$
if $P(x)=\top_\S$ and $I(P)(x)=\bot_\bO$ if $P(x)=\bot_\S$. This is a
complete Heyting algebra homomorphism, that is, a function preserving
all joins, all meets and pseudo-complements.

We introduce a generalized inverse image operation for $\bO$-valued
predicates. For a QBS-morphism $f:Y\to X$, define
$f^*:\UPred\bO X\to\UPred\bO Y$ by $f^*P=P\circ|f|$. This is also a
complete Heyting algebra homomorphism. We also introduce a notation: for a
QBS morphism $f:X\to Y$ and $P\in\UPred\bO X$ and $Q\in\UPred\bO Y$,
we write $f:P\dto Q$ to mean the inequality $P\sqsubseteq f^*Q$ in
$\UPred\bO X$.  When $\bO=\S$, $f:P\dto Q$ is equivalent to stating that for
any $\gamma$, $\gamma\models P$ implies $|f|(\gamma)\models Q$.

We also define the partial application of an $\bO$-valued predicate with
an environment.  Let $\Gamma,x_1:\tau_1,\cdots,x_n:\tau_n$ be a
context.  For a predicate
$P\in\UPred\bO{\sem{\G,x_1:\tau_1,\cdots,x_n:\tau_n}}$ and
$\g\in|\sem\G|$, by
$P_\gamma\in\UPred\bO{\sem{\tau_1}\times\cdots\times\sem{\tau_n}}$ we
mean the predicate
$ P_\gamma(p)=P\circ
|\mixm\Gamma{x_1:\tau_1,\cdots,x_n:\tau_n}|(\gamma,p).$
For $P\in\UPred\bO X$ and $Q\in\UPred\bO Y$, we define $P\dtimes Q=
\pi_1^*P\sqcap \pi_2^*Q$.

Note that we have chosen predicates to be morphisms in $\Set$, rather
than morphisms in $\QBS$. This allows us to avoid reasoning about
measurability when defining predicates and writing specifications, while
still having a model that works as intended when the predicates are
measurable.



\subsection{Strong Graded Liftings of the Probability Measure Monad}
\label{sec:st-gr-lift}


We introduce a concept called {\em strong graded lifting} of strong
monads.  The following definition is a specialization of the one in
\cite{10.1145/2535838.2535846} to $\bO$-valued predicates.
\begin{definition}[Heyting-valued strong graded lifting of strong
  monad]
  Let $(E,\le,0,+)$ be a partially ordered monoid.  An {\em
    $\bO$-valued strong $E$-graded lifting} of $\mprobn$ is a family
    of functions $\dot{\cP}_X \colon E\to(\UPred\bO X\To\UPred\bO{\mprob X})$,
  implicitly indexed by $X\in\QBS$, satisfying:
  \begin{align*}
    e\le e' & \implies \glprob e P\sqsubseteq \glprob {e'}P &
    \eta^\mprob_X &:P\dto\glprob 0 P \\
    \mu^\mprob_X & :\glprob e {\glprob{e'}P}\dto\glprob{e + e'}P &
    \theta^\mprob_X & :P\dtimes \glprob e Q\dto\glprob e{P\dtimes Q}.
  \end{align*}
\end{definition}

The following is an informal explanation of liftings, ignoring for the
moment the gradings. The second and third conditions specify how the
unit and multiplication of the lifting interact with the unit and
multiplication of the base monad: for a predicate $P$ over $X$ and
$x \in P$, then $\eta(x) \in \glprob 0 P$, and if
$x\in \glprob e {\glprob{e'} P}$, then $\mu(x) \in
\glprob{e + e'}Q$. The fourth condition specifies a similar interaction
with the strength. These properties are used in proving the soundness of
the rules of our logics.  At the level of liftings, gradings can be seen
as some additional specification, or as adding quantitative
information to the specification. For instance, in HO-UBL use the
grading on a lifting to specify the probability with which a
computation may fail to satisfy the specification. This is the
intuition behind the first condition in the definition: it allows
weakenings of the grading of a lifted predicate, i.e., if $e\le e'$
then $\glprob eQ \To \glprob {e'}Q$. We now present some examples of
liftings that we will use in our soundness proofs:

\begin{example}[Lifting for union bounds]\label{ex:lifting-UB}
  Inspired from the lifting for the union bound introduced
  in~\cite[Section 9.1]{SatoABGGH19}, we give a $\S$-valued strong
  $([0,\infty],\leq,+,0)$-graded lifting $\glubn$ of $\mprobn$:
  \begin{align*}
    \glub[X]\delta P(\nu)=\top
    & \iff
      \forall f \in \QBS(X,\{0,1\}\cdot 1),P\sqsubseteq |f|~.~
      \Pr_{x\sim\nu}[f(x)=1] \geq 1-\delta.
  \end{align*}
  This can be constructed by the graded
  $\top\top$-lifting~\cite{10.1145/2535838.2535846}.
  Morally, we want $\glub[X]\delta P(\nu)$ to hold if the probability of sampling
  a value from $\mu$ that satisfies $P$ is at least $1-\delta$. However, we cannot
  compute this probability directly because $P$ may not be measurable. Instead,
  we need to quantify over all the measurable $f$ above $P$.
  We can then show that if $P$ is indeed measurable (i.e. $P=|P_0|$ for some QBS-morphism $P_0$),
  $\Pr_{x \sim \nu}[P(x)=0] \leq \delta$ holds for every
  $\nu \in \glub\delta P$. That is, for measurable predicates, the
  lifting behaves as intended.
\end{example}

\begin{example}[Lifting for expectations]\label{ex:lifting-EXP}
    We introduce a $([0,\infty], \geq)$-valued strong $([0,\infty],\leq,+,0)$-graded lifting
$\dot\mprob^{\mathrm{Exp}}$ of $\mprob$. Note that the predicates take values
in the Heyting algebra $[0,\infty]$ \emph{with reversed order} (i.e. $x \sqsubseteq y$ iff $y \leq x$),
which will be used to reason about upper bounds. We define:
\begin{align*}
\glexp[X]\delta P(\nu) &= 
\sup\{
S_{\delta + \delta'}(f^\sharp \nu)
~|~
f \colon X \to \mprob[0,\infty],
\sup\{ S_{\delta'}(f(i))~|~ P(i) {<} S_{\delta'}(f(i))  \} {<} S_{\delta + \delta'}(f^\sharp \nu)
\}\\
&\qquad\text{where }
S_\delta (\mu) = \max(0,\EE_{r \sim \mu}[r] - \delta).
\end{align*}
This can be constructed by the graded $\top\top$-lifting~\cite{10.1145/2535838.2535846}
since $\delta \leq \delta' \implies S_\delta \sqsubseteq S_{\delta'}$.
Intuitively, the lifting $\dot\mprob^{\mathrm{Exp}}$ gives an upper bound of 
expected value of $P$ under a distribution $\nu$ with margin of error $\delta$.
In the general case, where $P$ is not measurable, we get instead
an upper bound on the expected value of any $f \colon X \dot\to [0,\infty]$ measurable in $\QBS$ 
such that $|f| \leq P$ (i.e. $P \sqsubseteq |f|$). 
That, is we obtain $\dot\mprob^{\mathrm{Exp}}_X(\delta)(P)(\nu) \sqsubseteq \EE_{x \in \nu}[f(x)]-\delta$.
\end{example}


\subsection{Combining liftings and state transformer monads}

The material from the previous section allows us to model predicates over
the monad $\mprob$, but we need to extend it to
model predicates over the probabilistic state monad $\mpstn$ that models
computations in our language.
The same approach of finding a lifting of $\mpstn$  does not work
directly because it would not allow us to include the specification about
states. Such a lifting would map a $\bO$-valued predicate over $X$ to a $\S$-valued predicate over
$\mpst X=\qbsmem \To \mprob(X\times \qbsmem)$, but this does not match the shape of triples
in our logics.
We actually need to lift a pair of $\bO$-valued predicates over $\qbsmem$ (the
precondition) and over $X\times \qbsmem$ (the postcondition) into a $\S$-valued predicate over
$\mpst X$.

Therefore, we need to find a different construction.  Assume there
exists a $\bO$-valued strong $(E,\le,0,+)$-graded lifting $\glprobn$ of
$\mprob$. For each QBS $X$, we define a function
$
  \lpst [X] {-} {-,-}
  :
  E\to
  (\UPred\bO \qbsmem \times\UPred\bO {X\times \qbsmem}\To\UPred\S{\mpst X})
$
by
$
  f\models\lpst [X] e {P,Q}
  \iff
  f : P\dto\glprob \delta Q.
$
Recall that $|\mpst X|=\QBS(\qbsmem,\mprob(X\times\qbsmem))$.  We call
$\lpstn$ a stateful lifting. This can be seen as a transformer that
takes a $\bO$-valued strong $E$-graded lifting $\glprobn$ of $\mprob$
and returns a stateful lifting of the probabilistic state transformer
monad $\stmtr{\mprob}$.  In plain words, $\lpstn$ maps an $\bO$-valued
precondition $P\in\UPred\bO \qbsmem$ and an $\bO$-valued postcondition
$Q\in\UPred\bO{X\times \qbsmem}$ to the computations in $\mpst(X)$
that send initial memories in $P$ to distributions over
$X\times \qbsmem$ satisfying the lifted predicate $\glprob eQ$. In a
way, this can be seen as the set of computations
$f\colon \mon{\Sigma,e}{\sigma}$ satisfying the generalized Hoare
triple $\{ P \} f:\mon{\Sigma,k}\sigma \{\!\{Q\}\!\}$.

This operator is not an $E-$graded lifting, because it does not
have the appropriate type.
However, properties of $\bO$-valued strong $E$-graded
liftings can be extended to $\lpstn$ as stated below:
\begin{lemma}\label{lem:gr-lift}
  Let $\glprobn$ be an $\bO$-valued strong $(E,\le,1,\cdot)$-graded lifting
  of $\mprobn$. Let $f \in \QBS(X\times\qbsmem,\mprob(Y\times\qbsmem))$, and $P\in\UPred\S X$, $Q\in\UPred\bO\qbsmem$,
  $R\in\UPred\bO{X\times\qbsmem}$ and $S\in\UPred\bO{Y\times\qbsmem}$
  be predicates. The following holds:
  \begin{align*}
    &\eta^\mpstn_X:P\dto \lpst[X] 0 {Q,IP\dtimes Q}
    \\
    &f:R\dto \glprob e S \implies
    (\lambda(f))^{\#\mpstn}:\lpst[X]{e'}{Q,R}\dto\lpst[Y]{e' + e}{Q,S}
    \\
    &\theta^{\mpstn}_{X,Y}: P\dtimes\lpst[Y]e{Q,R}\dto \lpst[X\times Y]e{Q,\incl P\sqcap R}
  \end{align*}
\end{lemma}

These consequences can be explained informally by using the language of Hoare
logic:
\begin{itemize}
\item The first consequence states that if $x \in \phi$, then
  $\{Q\}~\eta(x)~\{\!\{\phi \sqcap Q\}\!\}_0$ is a valid generalized
  Hoare triple for any $Q$.
\item The second consequence gives us a way to sequence computations as
  in Hoare logic. It states that if $f$ satisfies
  $\{ Q(x,-) \}~f(x)~\{\!\{ R \}\!\}_e$ for every argument $x$, and
  $t$ satisfies $\{P\}~t~\{\!\{Q\}\!\}_{e'}$, then
  $\{P\}~\mlet{x}{t}{f}~\{\!\{R\}\!\}_{e+e'}$ is a valid generalized
  Hoare triple. Here $Q(x,-)$ is the set of memories $m$ such that
  $(x,m)\in Q$.
\item The third consequence states that if $x \in \phi$ and we have $t$ such
  that $\{P\}~t~\{\!\{Q\}\!\}_e$ then
  $\{P\}~\theta(x,t)~\{\!\{\phi \sqcap Q\}\!\}_e$ is a valid
  generalized Hoare triple.
\end{itemize}



\subsection{Soundness of the unary logics}

We interpret formulas, assertions and entailment relations in the
logic by complete Heyting algebras over QBSs.  We first develop the
semantics of HOL judgements of the form
$\bullet\mid\bullet\mid\G\mid\Psi\vdash\phi$, which is simply denoted
by $\G\mid\Psi\vdash\phi$. We then interpret
an open judgement $J=\Xi\mid\Delta\mid\G\mid\Psi \vdash \phi$
as the conjunction of all closed instantiations
$\bullet~|~\bullet~|~\G'\vdash \Psi'\vdash\Phi'$ of $J$. Here,
each $\alpha\in\Xi$ is instanciated with some subset of $\loc$, and
$\mathcal A\in\Delta$ is instanciated with a closed term of type
$\Delta(\mathcal A)$. The semantics of open judgements
of UHOL and HO-UBL are similarly defined.
This interpretation is well-behaved with respect to substitution. In
particular, the substitution $\phi[t/x]$ of $x$ by a term $t$ of
appropriate type can be interpreted by the inverse image
$
\sem{\G \vdash \phi[t/x] }
= (\sub{\G\vdash t:\tau} x)^*
\sem{\G,x\colon\tau \vdash \phi }.
$
The soundness results of~\citet{ABGGS17} for the
base logics HOL and UHOL can be recovered in this
setting, but we defer it to the appendix.

We interpret HO-UBL using the lifting $\dstmtr\glubn$ of the
probabilistic state monad constructed from the lifting for the union
bound logic $\glubn$ as in Example~\ref{ex:lifting-UB}.
The soundness result is stated as:
\begin{proposition}\label{prop:ho-ubl-sound}
  Let $\aumj{\G}{\Psi}{P}{t}{\mon{\Sigma,k}{\tau}}{Q}{\delta}$ be a
  derivable HO-UBL judgment without the adversary rule.
  Then, for any
  $\g\in|\sem{\G}|$, $\g\models\sem{\holwf{\G}{\textstyle\bigwedge \Psi}}$
  implies
  \begin{align*}
    &
      \sem{\G\vdash t : \mon{\Sigma,k}{\tau}}(\g)
      \models
      \dstmtr{\glubn}(\delta)(\sem{\G,\st:\tmem\vdash P}_\g,\sem{\G,\vl:\tau,\st:\tmem\vdash Q}_\g).
  \end{align*}
\end{proposition}

Analogously, we interpret HO-EXP using the lifting $\dstmtr\glexpn$ of the
probabilistic state monad constructed from the lifting for expectations
bound logic $\glexpn$ as in Example~\ref{ex:lifting-EXP}. Most of the
proof of the previous result can be reused, and only the rules for sampling
and linearity need separate proofs.
The soundness result is stated as:
\begin{proposition}\label{prop:ho-expn-sound}
  Let $\aumj{\G}{\Psi}{P}{t}{\mon{\Sigma,k}{\tau}}{Q}{\delta}$ be a
  derivable HO-UBL judgment without the adversary rule.
  Then, for any
  $\g\in|\sem{\G}|$, $\g\models\sem{\holwf{\G}{\textstyle\bigwedge \Psi}}$
  implies
  \begin{align*}
    &
      \sem{\G\vdash t : \mon{\Sigma,k}{\tau}}(\g)
      \models
      \dstmtr{\glexpn}(\delta)(\sem{\G,\st:\tmem\vdash P}_\g,\sem{\G,\vl:\tau,\st:\tmem\vdash Q}_\g).
  \end{align*}
\end{proposition}

\subsection{Semantics for the relational logics}

\renewcommand{\BRel}[3]{{\bf BR}^{#1}_{#2,#3}}

Let $\bO=(\O,\sqsubseteq)$ be a complete Heyting algebra.  To
interpret relational logics, we first define the concept of
$\bO$-valued binary relation between two QBSs $X,Y$. They are simply
$\bO$-valued predicates over product QBS $X\times Y$. We thus define
$\BRel\bO XY\triangleq\UPred\bO{X\times Y}$. For $\S$-valued binary
relation $P\in\BRel\S XY$ and $(x,y)\in |X\times Y|=|X|\times |Y|$, we say
that $(x,y)$ satisfies $P$ (denoted by $(x,y)\models P$) if
$P(x,y)=\top$.

We routinely extend the development in the previous section to
$\bO$-valued binary relations. For QBS-morphisms $f:X\to Y$ and
$f':X'\to Y'$, we define the pullback operation
$(f,f')^*:\BRel\bO{X'}{Y'}\to\BRel\bO{X}{Y}$ to be $(f\times
f')^*$. We write $(f,f'):P\dto Q$ to mean $P\sqsubseteq (f,f')^*Q$.


We introduce the concept of {\em Heyting-algebra valued strong graded
  relational lifting}.
\begin{definition}
  Let $(E,\le,0,+)$ be a partially ordered monoid.  An {\em
    $\bO$-valued strong $E$-graded relational lifting} of $\mprobn$ is a family
  of functions $\glprob[X,Y]--:E\to(\BRel\bO XY\To\BRel\bO{\mprob X}{\mprob Y})$,
  implicitly indexed by $X,Y\in\QBS$, satisfying:
  \begin{align*}
    e\le e' & \implies \glprob e P\sqsubseteq \glprob {e'}P &
    (\eta^\mprob_X,\eta^\mprob_Y) &:P\dto\glprob 0 P \\
    (\mu^\mprob_X,\mu^\mprob_Y) & :\glprob e {\glprob{e'}P}\dto\glprob{e + e'}P &
                                                                                     (\theta^\mprob_X,\theta^\mprob_Y) & :P\dtimes \glprob e Q\dto\glprob e{P\dtimes Q}.
  \end{align*}
\end{definition}
\begin{example}[Relational lifting for differential privacy] \label{ex:DP-lifting}
Inspired from \cite{DBLP:journals/entcs/Sato16}, we construct a $\S$-valued strong graded relational
lifting for {\em differential privacy} by a graded analogue of the codensity
lifting~\cite{DBLP:journals/lmcs/KatsumataSU18}.
The grading monoid is the product partially ordered monoid $(\posr,\le,0,+)^2$.
\begin{align*}
\dot\mprob_{X,Y}^{\mathrm{dp}}(\epsilon,\delta)(P)(\nu_1,\nu_2)
= \top
&
\iff
\forall
(f,g) \colon P \dot\to S(\epsilon',\delta')~.~
(f^\sharp \nu_1,g^\sharp \nu_2) \models S(\epsilon+\epsilon',\delta+\delta') \\
&\text{ where }
(\nu_1,\nu_2) \models S(\epsilon,\delta)
\iff 
\Pr_{b \sim \nu_1}[b = 0] \leq e^\epsilon \Pr_{b \sim \nu_2}[b = 0] + \delta.
\end{align*}
Any morphism $\chi_S \colon X \to \{0,1\}$ in $\QBS$ standing for the ``measurable'' subset $S$ of $X$, we have
$\dot\mprob^{\mathrm{dp}}_{X,X}(\epsilon,\delta)(\mathrm{Eq}_X)(\mu_1,\mu_2) = \top \implies \Pr_{x \sim \mu_1}[x \in S] \leq e^\epsilon \Pr_{x \sim \mu_2}[x \in S] + \delta$ since $S(0,0) \models (\chi_S,\chi_S)$.
A strong $(\posr,\le,0,+)$-graded lifting describing $\epsilon$-differential privacy can be given by
$\dot\mprob_{X,Y}^{\mathrm{dp}}(\epsilon,0)$.
\end{example}
We next introduce the stateful relational lifting, based on a
$\bO$-valued $(E,\le,+,0)$-graded relational lifting $\glprobn$ of
$\mprobn$.  It is a function
$\lpst[X,Y]{-}{-,-}:E\to(\BRel\bO
XY\times\BRel\bO{X\times\qbsmem}{Y\times\qbsmem} \To\BRel\S{\mpst
  X}{\mpst Y})$ defined for each $X,Y\in\QBS$ by:
\begin{displaymath}
  (f,f')\models\lpst [X,Y]e  {P,Q}
  \iff
  (f,f'):P\dto \glprob e Q.
\end{displaymath}
Lemma~\ref{lem:gr-lift} can then be generalized accordingly.
We omit the details.

\subsection{Soundness of the relational logics}


The semantics of the relational logics are a generalization of the
semantics of the unary logics. We defer soundness of RHOL to the
appendix.
We interpret the monadic rules in the category of relations over QBS,
with the $\S$-valued strong $\posr^2$-graded relational lifting $\gldpn$,
which induces $\dstmtr{\gldpn}$ as in the previous section.
The soundness result is stated below.
Its proof is by induction on the derivation and is largely independent
of the choice of a specific lifting:
\begin{proposition}\label{prop:ho-rpl-sound}
  Let
  $\armj{\G}{\Psi}{P}{t_1}{\mon{\Sigma,k}{\tau_1}}{t_1}{\mon{\Sigma,k}{\tau_2}}{Q}{\delta}$
  be a derivable HO-PRL judgment without the \rname{ADV-R} rule.  Then
  for any $\g\in|\sem{\G}|$,
  $\gamma\models\sem{\holwf{\G}{\bigwedge \Psi}}$ implies 
  \begin{displaymath}
    (\sem{\G\vdash t_1: \mon{\Sigma,k}{\tau_1}}(\g),\sem{\G\vdash t_2:
      \mon{\Sigma,k}{\tau_2}}(\g))
    \models
    \lpst{0,\delta}{\sem P_\g,\sem Q_\g},
  \end{displaymath}
  where $\sem P_\g \triangleq \sem{\G,\st_1:\tmem,\st_2:\tmem\vdash
    P}_\g$ and  $Q_\g\triangleq\sem{\G,\st_1:\tmem,\vl_1:\sigma,\st_2:\tmem,\vl_2:\sigma\vdash
    Q}_\g$.
\end{proposition}

\subsection{Soundness of the adversary rules}

To prove soundness of the adversary rules we will use the technique of
logical relations.
Logical predicates and relations~\cite{plotkin1973lambda} are a technique used in programming
language theory to prove properties such as strong normalization or contextual
equivalence.  The idea of logical relations (or predicates) is that they allow us
to prove that all inhabitants of a certain type $\tau$ satisfy a particular property $\cL(\tau)$
that is defined inductively on the structure of types, rather than terms.

For instance, in the unary case, we define a logical predicate
$\cL_{\phi,\delta}(\cdot)$ indexed by an assertion $\phi$ over memories
and a real $\delta \geq 0$. For every type $\sigma$,
$\cL_{\phi,\delta}(\sigma)$ corresponds to a set of closed terms. We defer the details
of this definition to the appendix, here it suffices to know that
$\cL_{\phi,\delta}(\mon{\Sigma,k}{\tau})$ is the set of computations
that preserve the invariant $\phi$ with error probability $k\cdot\delta$
and that return a result in $\cL_{\phi,\delta}(\tau)$
(i.e., they satisfy the triple $\{\phi\}~\_~\{\!\{\phi \wedge \vl \in \cL_{\phi,\delta}(\tau)\}\!\}_{k\cdot\delta}$),
and as usual, if $t \in \cL_{\phi,\delta}(\sigma\to\tau)$ and $u \in \cL_{\phi,\delta}(\sigma)$
then $t~u \in \cL_{\phi,\delta}(\tau)$.
Then we prove a  Basic Lemma: any closed term $t$ of type $\tau$ inhabits the
predicate $\cL_{\phi,\delta}(\tau)$ if $\phi\in\safe{\eff{\tau}}$.  This has a
rather natural interpretation: if $\phi$ does not depend on any location in
$\eff{\tau}$, then it must be preserved after running $t$.

The adversary rule \rname{ADV-U} can then be proven sound from the Basic Lemma.
By inspecting its premises, we know that $\cA$ inhabits the logical relation
$\cL_{P,\delta}(\forall \alpha. (\sigma \to \mon{\alpha,1}{\tau}) \to \mon{\Sigma \cup \alpha, k}{\tau'})$,
because $P \in {\sf Safe}(\Sigma)$. We also have that $\lambda x. t$ inhabits the logical relation 
$\cL_{P,\delta}(\sigma \to \mon{\Sigma',1}{\tau})$,
because we have a derivation of this fact (note that the Basic Lemma cannot be applied, because $P$ may
not be safe for $\Sigma'$). Then, we can conclude that running $\cA$ with
$\lambda x. t$ as argument inhabits the logical relation 
$\cL_{P,\delta}(\mon{\Sigma \cup \alpha, k}{\tau'})$, and therefore must preserve $P$.
The techniques then generalize to the relational case, where we define a
logical relation for every type, and then we prove a Basic Lemma for it. Soundness
of \rname{ADV-R} is a consequence of this Lemma.

\begin{proposition}
    The \rname{ADV-U} and \rname{ADV-R} rules are sound.
\end{proposition}



\section{Related work}\label{sec:related}


\subsubsection*{Reasoning about adversaries}
\citet{Garg10,Jia15} develop first- and higher-order program logics to
reason about safety properties of first-order concurrent and stateful
programs interacting with adversaries. Both provide rules to
reason about adversaries, morally similar to ours.
Our context of adversary variables representing closed programs
traces lineage to a similar idea based on comonads
in~\citet{Jia15}. \citet{Devriese16} develop semantic principles to
reason about adversaries, cast in terms of parametricity properties of
side-effects, an idea they call ``effect parametricity''. They use
these principles to verify code that uses object capabilities. No
syntactic proof rules are developed. These works cover only the
boolean, deterministic, unary setting.

Closer to our work, \citet{BartheGZ09} define a probabilistic
relational Hoare logic (pRHL) for reasoning about the security of
cryptographic constructions. Their logic applies to a probabilistic
imperative language with adversarial calls and features a proof rule
for adversaries. Our rule for the relational, non-quantitative setting
closely matches their rule.  \citet{BarbosaBGKS21} formalize a
resource-aware module system used in EasyCrypt to reason about
adversaries. There are commonalities between their approach and ours:
they view an adversary as a functor, whereas we view an adversary as
an expression of second-order type. However, the technicalities are
very different, since they build their system on top of an imperative
language. A further difference is that they account for the
computational cost of adversaries, which we left aside in this
work. Other similar approaches for reasoning about adversaries include
Computational Indistinguishability Logic~\citep{BartheDKL10}, and
state-separating proofs~\citep{BrzuskaDFKK18}. However, these
approaches are developed in an abstract mathematical setting, not in
the context of a programming language.

\citet{BartheFGGHS16} define an adversary rule for reasoning about
differential privacy in a quantitative variant of pRHL; their rule
uses bounds on the number of oracle queries to derive privacy bounds
of adversarial computations from privacy bounds of
oracles. \citet{BartheGGHS16} define a Union Bound logic to reason
about accuracy of adversarial computations for a similar
language. However, their proof rule is restricted to adversaries
without oracles.
We are not aware of any prior work on adversarial computations in the
quantitative setting.

\subsubsection*{Program logics for probabilistic computations}
We relate our program logics to existing approaches
for reasoning about probabilistic computations. For brevity, we only discuss
    approaches not discussed before. 
\citet{Kozen85} introduces expectation-based reasoning for a core
probabilistic programming language. \citet{Morgan96} define a weakest
pre-expectation calculus. \citet{Aguirre20} develop a variant of the
calculus for relational properties. \citet{KaminskiKMO16} show how
similar ideas can be used for reasoning about expected cost. All these
works share the setting of a probabilistic imperative
language. \citet{AguirreK20} show that expectation-based reasoning
remains sound in a higher-order setting, but their semantics is based
on set theory, not Quasi-Borel spaces, so they cannot model continuous
distributions. They also do not provide proof systems.

There exist adaptations of (approximate) probabilistic relational
Hoare logic in the higher-order setting, starting
from~\cite{BFGSSZ14}.  However, these adaptations have a
set-theoretical or topos of trees semantics and only support
distributions over discrete base types. \citet{SatoABGGH19} introduce
an expressive logic for a language similar to ours but without state
and adversary. Their model is also based on QBS.
\citet{TassarottiH19} develop a logic to prove relational properties
of higher-order programs that combine probabilities and non-determinism.
They do not support all the kinds of reasoning we do, and the relations
they can prove are between a program and an specification, rather
than between two programs.
\citet{MaillardHRM20} define a framework, embedded in a
relational dependent type theory, for defining and reasoning about
program logics for general monadic effects.  
While their work is based on Dijkstra monads, ours is more closely
related to Hoare monads~\cite{NanevskiMB08, NanevskiBG13}. Our work
extends Hoare monads to support Heyting-valued predicates,
probabilistic programs, grading and adversarial reasoning.

\subsubsection*{Program equivalence}
There is a very large body of methods for proving program
equivalence, and in particular contextual equivalence, in higher-order
languages with state, probabilities, and effects; see
e.g.~\cite{stark98,BentonHN14,JungSSSTBD15,BizjakB15,CrubilleL15,MatacheS19}.
Many of these methods have been applied to reason about security and
privacy, using the natural view of adversaries as contexts. These
methods are not comparable with ours: our relational logic can prove a
richer set of specifications (for instance, the postcondition needs
not be an equivalence relation). However, they cannot establish some
basic equivalences, e.g. swapping of two sampling instructions, due to
the specific way the logic constructs couplings. We also conjecture
that our logics are easier to extend to richer settings, such as
multi-stage and multi-adversary security notions (see
e.g.~\cite{RistenpartSS11}).  Finally, these methods cannot be used to
reason about unary properties.

\section{Concluding remarks}

We conclude the paper with a discussion of additional examples that
can be handled by the three logics we have presented (and by small
extensions to the logics), and a discussion of how we can extend our
framework with unbounded recursion.

\paragraph*{Other examples}
%
HO-UBL can verify the \emph{accuracy} of differentially private
mechanisms such as the Sparse Vector algorithm~\cite{DR14}, since
accuracy can be formulated as the probability that the noisy answer is
close to the actual answer. We have already worked out this example
but, for reasons of space, we defer it to the appendix.

The bounded leakage model is a model of leakage-resilient cryptography
in which the adversary is given access to a leakage oracle which takes
as input a function with a small codomain and returns the output of
this function applied to the secret state. A (partially formalized)
proof of security of a pseudo-random generator in the bounded leakage
model is given in~\cite{bbmo14}. HO-PRL can be used to verify this
proof, using either a first- or a higher-order representation of
leakage.

Other examples can be verified with extensions to our logics that can
also be proved sound in our framework. For instance, we can support
a slightly different relational logic in which the Hoare quadruple
is indexed by a pair $(\epsilon, \delta)$, and interpreted using the lifting from
Example~\ref{ex:DP-lifting}. With this logic, we
can study differential privacy of mechanisms such
as the exponential mechanism on non-numerical queries~\cite{DR14},
which uses a scoring function that assigns positive values to all
possible input/output pairs. Prior work \citep{BartheKOB12},
has verified this mechanism using a first-order representation of
scoring. However, we can verify a higher-order representation
of this mechanism, where the scoring function is passed as an argument
to the mechanism.

We can also use HO-EXP to verify examples based on the weakest
pre-expectation calculus~\cite{Morgan96}. One caveat is that many of
these examples use arbitrary while loops, which our language does not
currently support. This extension would require extending the model as
discussed at the end of this section.
Other examples, e.g., stability of machine learning algorithms, would
require developing a logic for relational pre-expectations. Yet
others, e.g. cryptography, would require enriching our logics with
additional proof principles that embed notions of cryptographic
reductions.  In the long run, it would be interesting to support these
formalisms with an implementation to mechanize
examples.



\paragraph*{Unbounded recursion}
Our language provides bounded recursion via the monadic fold. An
interesting follow-up would be to extend our language with unbounded
monadic recursion. For this, we would also need to change the semantic
model. One possibility is to use the recently proposed category
$\omega\QBS$~\cite{VakarKS19} to interpret types.

\begin{acks}
  S.K. was supported by ERATO HASUO Metamathematics for
  Systems Design Project (No.~\grantnum{ERATO}{JPMJER1603}),
  \grantsponsor{JST}{Japan Science and
    Technology Agency}{}.
	T.S. was supported by JSPS KAKENHI Grant Number 20K19775, Japan.
  M.G. was supported by NSF awards  CCF-2040222 and CCF-1718220.
\end{acks}

\bibliographystyle{ACM-Reference-Format}
\bibliography{refs}


\begin{thebibliography}{48}


\ifx \showCODEN    \undefined \def \showCODEN     #1{\unskip}     \fi
\ifx \showDOI      \undefined \def \showDOI       #1{#1}\fi
\ifx \showISBNx    \undefined \def \showISBNx     #1{\unskip}     \fi
\ifx \showISBNxiii \undefined \def \showISBNxiii  #1{\unskip}     \fi
\ifx \showISSN     \undefined \def \showISSN      #1{\unskip}     \fi
\ifx \showLCCN     \undefined \def \showLCCN      #1{\unskip}     \fi
\ifx \shownote     \undefined \def \shownote      #1{#1}          \fi
\ifx \showarticletitle \undefined \def \showarticletitle #1{#1}   \fi
\ifx \showURL      \undefined \def \showURL       {\relax}        \fi
\providecommand\bibfield[2]{#2}
\providecommand\bibinfo[2]{#2}
\providecommand\natexlab[1]{#1}
\providecommand\showeprint[2][]{arXiv:#2}

\bibitem[\protect\citeauthoryear{Aguirre, Barthe, Gaboardi, Garg, and
  Strub}{Aguirre et~al\mbox{.}}{2017}]%
        {ABGGS17}
\bibfield{author}{\bibinfo{person}{Alejandro Aguirre}, \bibinfo{person}{Gilles
  Barthe}, \bibinfo{person}{Marco Gaboardi}, \bibinfo{person}{Deepak Garg},
  {and} \bibinfo{person}{Pierre{-}Yves Strub}.}
  \bibinfo{year}{2017}\natexlab{}.
\newblock \showarticletitle{A relational logic for higher-order programs}.
\newblock \bibinfo{journal}{\emph{{PACMPL}}} \bibinfo{volume}{1},
  \bibinfo{number}{{ICFP}} (\bibinfo{year}{2017}),
  \bibinfo{pages}{21:1--21:29}.
\newblock
\urldef\tempurl%
\url{https://doi.org/10.1145/3110265}
\showDOI{\tempurl}


\bibitem[\protect\citeauthoryear{Aguirre, Barthe, Hsu, Kaminski, Katoen, and
  Matheja}{Aguirre et~al\mbox{.}}{2021}]%
        {Aguirre20}
\bibfield{author}{\bibinfo{person}{Alejandro Aguirre}, \bibinfo{person}{Gilles
  Barthe}, \bibinfo{person}{Justin Hsu}, \bibinfo{person}{Benjamin~Lucien
  Kaminski}, \bibinfo{person}{Joost-Pieter Katoen}, {and}
  \bibinfo{person}{Christoph Matheja}.} \bibinfo{year}{2021}\natexlab{}.
\newblock \showarticletitle{A Pre-Expectation Calculus for Probabilistic
  Sensitivity}.
\newblock \bibinfo{journal}{\emph{Proc. ACM Program. Lang.}}
  \bibinfo{volume}{5}, \bibinfo{number}{POPL}, Article \bibinfo{articleno}{52}
  (\bibinfo{date}{Jan.} \bibinfo{year}{2021}), \bibinfo{numpages}{28}~pages.
\newblock
\urldef\tempurl%
\url{https://doi.org/10.1145/3434333}
\showDOI{\tempurl}


\bibitem[\protect\citeauthoryear{Aguirre and Katsumata}{Aguirre and
  Katsumata}{2020}]%
        {AguirreK20}
\bibfield{author}{\bibinfo{person}{Alejandro Aguirre} {and}
  \bibinfo{person}{Shin{-}ya Katsumata}.} \bibinfo{year}{2020}\natexlab{}.
\newblock \bibinfo{title}{Weakest preconditions in fibrations}.
  (\bibinfo{year}{2020}).
\newblock
\newblock
\shownote{Accepted at MFPS'20.}


\bibitem[\protect\citeauthoryear{Aumann et~al\mbox{.}}{Aumann
  et~al\mbox{.}}{1961}]%
        {aumann1961borel}
\bibfield{author}{\bibinfo{person}{Robert~J Aumann} {et~al\mbox{.}}}
  \bibinfo{year}{1961}\natexlab{}.
\newblock \showarticletitle{Borel structures for function spaces}.
\newblock \bibinfo{journal}{\emph{Illinois Journal of Mathematics}}
  \bibinfo{volume}{5}, \bibinfo{number}{4} (\bibinfo{year}{1961}),
  \bibinfo{pages}{614--630}.
\newblock


\bibitem[\protect\citeauthoryear{Barbosa, Barthe, Grégoire, Koutsos, and
  Strub}{Barbosa et~al\mbox{.}}{2021}]%
        {BarbosaBGKS21}
\bibfield{author}{\bibinfo{person}{Manuel Barbosa}, \bibinfo{person}{Gilles
  Barthe}, \bibinfo{person}{Benjamin Grégoire}, \bibinfo{person}{Adrien
  Koutsos}, {and} \bibinfo{person}{Pierre-Yves Strub}.}
  \bibinfo{year}{2021}\natexlab{}.
\newblock \bibinfo{title}{Mechanized Proofs of Adversarial Complexity and
  Application to Universal Composability}.
\newblock \bibinfo{howpublished}{Cryptology ePrint Archive, Report 2021/156}.
\newblock
\newblock
\shownote{\url{https://eprint.iacr.org/2021/156}.}


\bibitem[\protect\citeauthoryear{Barthe, Daubignard, Kapron, and
  Lakhnech}{Barthe et~al\mbox{.}}{2010}]%
        {BartheDKL10}
\bibfield{author}{\bibinfo{person}{Gilles Barthe}, \bibinfo{person}{Marion
  Daubignard}, \bibinfo{person}{Bruce~M. Kapron}, {and}
  \bibinfo{person}{Yassine Lakhnech}.} \bibinfo{year}{2010}\natexlab{}.
\newblock \showarticletitle{Computational indistinguishability logic}. In
  \bibinfo{booktitle}{\emph{Proceedings of the 17th {ACM} Conference on
  Computer and Communications Security, {CCS} 2010, Chicago, Illinois, USA,
  October 4-8, 2010}}, \bibfield{editor}{\bibinfo{person}{Ehab Al{-}Shaer},
  \bibinfo{person}{Angelos~D. Keromytis}, {and} \bibinfo{person}{Vitaly
  Shmatikov}} (Eds.). \bibinfo{publisher}{{ACM}}, \bibinfo{pages}{375--386}.
\newblock
\urldef\tempurl%
\url{https://doi.org/10.1145/1866307.1866350}
\showDOI{\tempurl}


\bibitem[\protect\citeauthoryear{Barthe, Espitau, Gr{\'{e}}goire, Hsu, and
  Strub}{Barthe et~al\mbox{.}}{2018}]%
        {BartheEGHS18}
\bibfield{author}{\bibinfo{person}{Gilles Barthe}, \bibinfo{person}{Thomas
  Espitau}, \bibinfo{person}{Benjamin Gr{\'{e}}goire}, \bibinfo{person}{Justin
  Hsu}, {and} \bibinfo{person}{Pierre{-}Yves Strub}.}
  \bibinfo{year}{2018}\natexlab{}.
\newblock \showarticletitle{Proving expected sensitivity of probabilistic
  programs}.
\newblock \bibinfo{journal}{\emph{Proc. {ACM} Program. Lang.}}
  \bibinfo{volume}{2}, \bibinfo{number}{{POPL}} (\bibinfo{year}{2018}),
  \bibinfo{pages}{57:1--57:29}.
\newblock
\urldef\tempurl%
\url{https://doi.org/10.1145/3158145}
\showDOI{\tempurl}


\bibitem[\protect\citeauthoryear{Barthe, Fong, Gaboardi, Gr{\'{e}}goire, Hsu,
  and Strub}{Barthe et~al\mbox{.}}{2016a}]%
        {BartheFGGHS16}
\bibfield{author}{\bibinfo{person}{Gilles Barthe},
  \bibinfo{person}{No{\'{e}}mie Fong}, \bibinfo{person}{Marco Gaboardi},
  \bibinfo{person}{Benjamin Gr{\'{e}}goire}, \bibinfo{person}{Justin Hsu},
  {and} \bibinfo{person}{Pierre{-}Yves Strub}.}
  \bibinfo{year}{2016}\natexlab{a}.
\newblock \showarticletitle{Advanced Probabilistic Couplings for Differential
  Privacy}. In \bibinfo{booktitle}{\emph{Proceedings of the 2016 {ACM} {SIGSAC}
  Conference on Computer and Communications Security, Vienna, Austria, October
  24-28, 2016}}, \bibfield{editor}{\bibinfo{person}{Edgar~R. Weippl},
  \bibinfo{person}{Stefan Katzenbeisser}, \bibinfo{person}{Christopher
  Kruegel}, \bibinfo{person}{Andrew~C. Myers}, {and} \bibinfo{person}{Shai
  Halevi}} (Eds.). \bibinfo{publisher}{{ACM}}, \bibinfo{pages}{55--67}.
\newblock


\bibitem[\protect\citeauthoryear{Barthe, Fournet, Gr{\'{e}}goire, Strub, Swamy,
  and B{\'{e}}guelin}{Barthe et~al\mbox{.}}{2014a}]%
        {BFGSSZ14}
\bibfield{author}{\bibinfo{person}{Gilles Barthe},
  \bibinfo{person}{C{\'{e}}dric Fournet}, \bibinfo{person}{Benjamin
  Gr{\'{e}}goire}, \bibinfo{person}{Pierre{-}Yves Strub},
  \bibinfo{person}{Nikhil Swamy}, {and} \bibinfo{person}{Santiago~Zanella
  B{\'{e}}guelin}.} \bibinfo{year}{2014}\natexlab{a}.
\newblock \showarticletitle{Probabilistic relational verification for
  cryptographic implementations}. In \bibinfo{booktitle}{\emph{{POPL} 2014}},
  \bibfield{editor}{\bibinfo{person}{Suresh Jagannathan} {and}
  \bibinfo{person}{Peter Sewell}} (Eds.).
\newblock


\bibitem[\protect\citeauthoryear{Barthe, Gaboardi, Gr{\'{e}}goire, Hsu, and
  Strub}{Barthe et~al\mbox{.}}{2016b}]%
        {BartheGGHS16}
\bibfield{author}{\bibinfo{person}{Gilles Barthe}, \bibinfo{person}{Marco
  Gaboardi}, \bibinfo{person}{Benjamin Gr{\'{e}}goire}, \bibinfo{person}{Justin
  Hsu}, {and} \bibinfo{person}{Pierre{-}Yves Strub}.}
  \bibinfo{year}{2016}\natexlab{b}.
\newblock \showarticletitle{A Program Logic for Union Bounds}. In
  \bibinfo{booktitle}{\emph{43rd International Colloquium on Automata,
  Languages, and Programming, {ICALP} 2016, July 11-15, 2016, Rome, Italy}}
  \emph{(\bibinfo{series}{LIPIcs}, Vol.~\bibinfo{volume}{55})},
  \bibfield{editor}{\bibinfo{person}{Ioannis Chatzigiannakis},
  \bibinfo{person}{Michael Mitzenmacher}, \bibinfo{person}{Yuval Rabani}, {and}
  \bibinfo{person}{Davide Sangiorgi}} (Eds.). \bibinfo{publisher}{Schloss
  Dagstuhl - Leibniz-Zentrum fuer Informatik}, \bibinfo{pages}{107:1--107:15}.
\newblock
\urldef\tempurl%
\url{https://doi.org/10.4230/LIPIcs.ICALP.2016.107}
\showDOI{\tempurl}


\bibitem[\protect\citeauthoryear{Barthe, Gr{\'{e}}goire, and
  B{\'{e}}guelin}{Barthe et~al\mbox{.}}{2009}]%
        {BartheGZ09}
\bibfield{author}{\bibinfo{person}{Gilles Barthe}, \bibinfo{person}{Benjamin
  Gr{\'{e}}goire}, {and} \bibinfo{person}{Santiago~Zanella B{\'{e}}guelin}.}
  \bibinfo{year}{2009}\natexlab{}.
\newblock \showarticletitle{Formal certification of code-based cryptographic
  proofs}. In \bibinfo{booktitle}{\emph{{POPL} 2009, Savannah, GA, USA, January
  21-23, 2009}}.
\newblock
\urldef\tempurl%
\url{https://doi.org/10.1145/1480881.1480894}
\showDOI{\tempurl}


\bibitem[\protect\citeauthoryear{Barthe, K{\"o}pf, Mauborgne, and Ochoa}{Barthe
  et~al\mbox{.}}{2014b}]%
        {bbmo14}
\bibfield{author}{\bibinfo{person}{Gilles Barthe}, \bibinfo{person}{Boris
  K{\"o}pf}, \bibinfo{person}{Laurent Mauborgne}, {and}
  \bibinfo{person}{Mart{\'i}n Ochoa}.} \bibinfo{year}{2014}\natexlab{b}.
\newblock \showarticletitle{{Leakage Resilience against Concurrent Cache
  Attacks}}. In \bibinfo{booktitle}{\emph{Proc. 3rd Conference on Principles of
  Security and Trust (POST '14)}}. \bibinfo{publisher}{Springer}.
\newblock


\bibitem[\protect\citeauthoryear{Barthe, K{\"{o}}pf, Olmedo, and
  B{\'{e}}guelin}{Barthe et~al\mbox{.}}{2012}]%
        {BartheKOB12}
\bibfield{author}{\bibinfo{person}{Gilles Barthe}, \bibinfo{person}{Boris
  K{\"{o}}pf}, \bibinfo{person}{Federico Olmedo}, {and}
  \bibinfo{person}{Santiago~Zanella B{\'{e}}guelin}.}
  \bibinfo{year}{2012}\natexlab{}.
\newblock \showarticletitle{Probabilistic relational reasoning for differential
  privacy}. In \bibinfo{booktitle}{\emph{{POPL} 2012, Philadelphia,
  Pennsylvania, USA, January 22-28, 2012}}.
\newblock
\urldef\tempurl%
\url{https://doi.org/10.1145/2103656.2103670}
\showDOI{\tempurl}


\bibitem[\protect\citeauthoryear{Benton, Hofmann, and Nigam}{Benton
  et~al\mbox{.}}{2014}]%
        {BentonHN14}
\bibfield{author}{\bibinfo{person}{Nick Benton}, \bibinfo{person}{Martin
  Hofmann}, {and} \bibinfo{person}{Vivek Nigam}.}
  \bibinfo{year}{2014}\natexlab{}.
\newblock \showarticletitle{Abstract effects and proof-relevant logical
  relations}. In \bibinfo{booktitle}{\emph{The 41st Annual {ACM}
  {SIGPLAN-SIGACT} Symposium on Principles of Programming Languages, {POPL}
  '14, San Diego, CA, USA, January 20-21, 2014}},
  \bibfield{editor}{\bibinfo{person}{Suresh Jagannathan} {and}
  \bibinfo{person}{Peter Sewell}} (Eds.). \bibinfo{publisher}{{ACM}},
  \bibinfo{pages}{619--632}.
\newblock
\urldef\tempurl%
\url{https://doi.org/10.1145/2535838.2535869}
\showDOI{\tempurl}


\bibitem[\protect\citeauthoryear{Bizjak and Birkedal}{Bizjak and
  Birkedal}{2015}]%
        {BizjakB15}
\bibfield{author}{\bibinfo{person}{Ales Bizjak} {and} \bibinfo{person}{Lars
  Birkedal}.} \bibinfo{year}{2015}\natexlab{}.
\newblock \showarticletitle{Step-Indexed Logical Relations for Probability}. In
  \bibinfo{booktitle}{\emph{FoSSaCS 2015, London, UK, April 11-18, 2015.
  Proceedings}}.
\newblock


\bibitem[\protect\citeauthoryear{Bloom}{Bloom}{1970}]%
        {Bloom70}
\bibfield{author}{\bibinfo{person}{Burton~H. Bloom}.}
  \bibinfo{year}{1970}\natexlab{}.
\newblock \showarticletitle{Space/Time Trade-offs in Hash Coding with Allowable
  Errors}.
\newblock \bibinfo{journal}{\emph{Commun. {ACM}}} \bibinfo{volume}{13},
  \bibinfo{number}{7} (\bibinfo{year}{1970}), \bibinfo{pages}{422--426}.
\newblock
\urldef\tempurl%
\url{https://doi.org/10.1145/362686.362692}
\showDOI{\tempurl}


\bibitem[\protect\citeauthoryear{Brzuska, Delignat{-}Lavaud, Fournet, Kohbrok,
  and Kohlweiss}{Brzuska et~al\mbox{.}}{2018}]%
        {BrzuskaDFKK18}
\bibfield{author}{\bibinfo{person}{Chris Brzuska}, \bibinfo{person}{Antoine
  Delignat{-}Lavaud}, \bibinfo{person}{C{\'{e}}dric Fournet},
  \bibinfo{person}{Konrad Kohbrok}, {and} \bibinfo{person}{Markulf Kohlweiss}.}
  \bibinfo{year}{2018}\natexlab{}.
\newblock \showarticletitle{State Separation for Code-Based Game-Playing
  Proofs}. In \bibinfo{booktitle}{\emph{Advances in Cryptology - {ASIACRYPT}
  2018 - 24th International Conference on the Theory and Application of
  Cryptology and Information Security, Brisbane, QLD, Australia, December 2-6,
  2018, Proceedings, Part {III}}} \emph{(\bibinfo{series}{Lecture Notes in
  Computer Science}, Vol.~\bibinfo{volume}{11274})},
  \bibfield{editor}{\bibinfo{person}{Thomas Peyrin} {and}
  \bibinfo{person}{Steven~D. Galbraith}} (Eds.). \bibinfo{publisher}{Springer},
  \bibinfo{pages}{222--249}.
\newblock
\urldef\tempurl%
\url{https://doi.org/10.1007/978-3-030-03332-3\_9}
\showDOI{\tempurl}


\bibitem[\protect\citeauthoryear{Clayton, Patton, and Shrimpton}{Clayton
  et~al\mbox{.}}{2019}]%
        {ClaytonPS19}
\bibfield{author}{\bibinfo{person}{David Clayton}, \bibinfo{person}{Christopher
  Patton}, {and} \bibinfo{person}{Thomas Shrimpton}.}
  \bibinfo{year}{2019}\natexlab{}.
\newblock \showarticletitle{Probabilistic Data Structures in Adversarial
  Environments}. In \bibinfo{booktitle}{\emph{Proceedings of the 2019 {ACM}
  {SIGSAC} Conference on Computer and Communications Security, {CCS} 2019,
  London, UK, November 11-15, 2019}},
  \bibfield{editor}{\bibinfo{person}{Lorenzo Cavallaro},
  \bibinfo{person}{Johannes Kinder}, \bibinfo{person}{XiaoFeng Wang}, {and}
  \bibinfo{person}{Jonathan Katz}} (Eds.). \bibinfo{publisher}{{ACM}},
  \bibinfo{pages}{1317--1334}.
\newblock
\urldef\tempurl%
\url{https://doi.org/10.1145/3319535.3354235}
\showDOI{\tempurl}


\bibitem[\protect\citeauthoryear{Crubill{\'{e}} and Lago}{Crubill{\'{e}} and
  Lago}{2015}]%
        {CrubilleL15}
\bibfield{author}{\bibinfo{person}{Rapha{\"{e}}lle Crubill{\'{e}}} {and}
  \bibinfo{person}{Ugo~Dal Lago}.} \bibinfo{year}{2015}\natexlab{}.
\newblock \showarticletitle{Metric Reasoning about {\(\lambda\)}-Terms: The
  Affine Case}. In \bibinfo{booktitle}{\emph{30th Annual {ACM/IEEE} Symposium
  on Logic in Computer Science, {LICS} 2015, Kyoto, Japan, July 6-10, 2015}}.
  \bibinfo{publisher}{{IEEE} Computer Society}, \bibinfo{pages}{633--644}.
\newblock
\urldef\tempurl%
\url{https://doi.org/10.1109/LICS.2015.64}
\showDOI{\tempurl}


\bibitem[\protect\citeauthoryear{Danos and Ehrhard}{Danos and Ehrhard}{2011}]%
        {DanosE11}
\bibfield{author}{\bibinfo{person}{Vincent Danos} {and} \bibinfo{person}{Thomas
  Ehrhard}.} \bibinfo{year}{2011}\natexlab{}.
\newblock \showarticletitle{Probabilistic coherence spaces as a model of
  higher-order probabilistic computation}.
\newblock \bibinfo{journal}{\emph{Inf. Comput.}} \bibinfo{volume}{209},
  \bibinfo{number}{6} (\bibinfo{year}{2011}), \bibinfo{pages}{966--991}.
\newblock
\urldef\tempurl%
\url{https://doi.org/10.1016/j.ic.2011.02.001}
\showDOI{\tempurl}


\bibitem[\protect\citeauthoryear{Devriese, Birkedal, and Piessens}{Devriese
  et~al\mbox{.}}{2016}]%
        {Devriese16}
\bibfield{author}{\bibinfo{person}{Dominique Devriese}, \bibinfo{person}{Lars
  Birkedal}, {and} \bibinfo{person}{Frank Piessens}.}
  \bibinfo{year}{2016}\natexlab{}.
\newblock \showarticletitle{Reasoning about Object Capabilities with Logical
  Relations and Effect Parametricity}. In \bibinfo{booktitle}{\emph{{IEEE}
  European Symposium on Security and Privacy (EuroS{\&}P)}}.
  \bibinfo{pages}{147--162}.
\newblock


\bibitem[\protect\citeauthoryear{Dwork and Roth}{Dwork and Roth}{2014}]%
        {DR14}
\bibfield{author}{\bibinfo{person}{Cynthia Dwork} {and} \bibinfo{person}{Aaron
  Roth}.} \bibinfo{year}{2014}\natexlab{}.
\newblock \showarticletitle{The Algorithmic Foundations of Differential
  Privacy}.
\newblock \bibinfo{journal}{\emph{Foundations and Trends in Theoretical
  Computer Science}} \bibinfo{volume}{9}, \bibinfo{number}{3--4}
  (\bibinfo{year}{2014}), \bibinfo{pages}{211--407}.
\newblock
\urldef\tempurl%
\url{http://dx.doi.org/10.1561/0400000042}
\showURL{%
\tempurl}


\bibitem[\protect\citeauthoryear{Garg, Franklin, Kaynar, and Datta}{Garg
  et~al\mbox{.}}{2010}]%
        {Garg10}
\bibfield{author}{\bibinfo{person}{Deepak Garg}, \bibinfo{person}{Jason
  Franklin}, \bibinfo{person}{Dilsun~Kirli Kaynar}, {and}
  \bibinfo{person}{Anupam Datta}.} \bibinfo{year}{2010}\natexlab{}.
\newblock \showarticletitle{Compositional System Security with
  Interface-Confined Adversaries}.
\newblock \bibinfo{journal}{\emph{Electr. Notes Theor. Comput. Sci.}}
  \bibinfo{volume}{265} (\bibinfo{year}{2010}), \bibinfo{pages}{49--71}.
\newblock


\bibitem[\protect\citeauthoryear{Gerbet, Kumar, and Lauradoux}{Gerbet
  et~al\mbox{.}}{2015}]%
        {GerbetKL15}
\bibfield{author}{\bibinfo{person}{Thomas Gerbet}, \bibinfo{person}{Amrit
  Kumar}, {and} \bibinfo{person}{C{\'{e}}dric Lauradoux}.}
  \bibinfo{year}{2015}\natexlab{}.
\newblock \showarticletitle{The Power of Evil Choices in Bloom Filters}. In
  \bibinfo{booktitle}{\emph{45th Annual {IEEE/IFIP} International Conference on
  Dependable Systems and Networks, {DSN} 2015, Rio de Janeiro, Brazil, June
  22-25, 2015}}. \bibinfo{publisher}{{IEEE} Computer Society},
  \bibinfo{pages}{101--112}.
\newblock
\urldef\tempurl%
\url{https://doi.org/10.1109/DSN.2015.21}
\showDOI{\tempurl}


\bibitem[\protect\citeauthoryear{Heunen, Kammar, Staton, and Yang}{Heunen
  et~al\mbox{.}}{2017}]%
        {HeunenKSY17}
\bibfield{author}{\bibinfo{person}{Chris Heunen}, \bibinfo{person}{Ohad
  Kammar}, \bibinfo{person}{Sam Staton}, {and} \bibinfo{person}{Hongseok
  Yang}.} \bibinfo{year}{2017}\natexlab{}.
\newblock \showarticletitle{A convenient category for higher-order probability
  theory}. In \bibinfo{booktitle}{\emph{32nd Annual {ACM/IEEE} Symposium on
  Logic in Computer Science, {LICS} 2017, Reykjavik, Iceland, June 20-23,
  2017}}. \bibinfo{publisher}{{IEEE} Computer Society}, \bibinfo{pages}{1--12}.
\newblock
\urldef\tempurl%
\url{https://doi.org/10.1109/LICS.2017.8005137}
\showDOI{\tempurl}


\bibitem[\protect\citeauthoryear{Impagliazzo and Rudich}{Impagliazzo and
  Rudich}{1989}]%
        {Impagliazzo:1989}
\bibfield{author}{\bibinfo{person}{R. Impagliazzo} {and} \bibinfo{person}{S.
  Rudich}.} \bibinfo{year}{1989}\natexlab{}.
\newblock \showarticletitle{Limits on the provable consequences of one-way
  permutations}. In \bibinfo{booktitle}{\emph{21st Annual ACM Symposium on
  Theory of Computing, 1989}}. \bibinfo{publisher}{ACM}, \bibinfo{address}{New
  York}, \bibinfo{pages}{44--61}.
\newblock


\bibitem[\protect\citeauthoryear{Iverson}{Iverson}{1962}]%
        {iverson62}
\bibfield{author}{\bibinfo{person}{Kenneth~E. Iverson}.}
  \bibinfo{year}{1962}\natexlab{}.
\newblock \bibinfo{booktitle}{\emph{A Programming Language}}.
\newblock \bibinfo{publisher}{John Wiley \& Sons, Inc.},
  \bibinfo{address}{USA}.
\newblock
\showISBNx{0471430145}


\bibitem[\protect\citeauthoryear{Jia, Sen, Garg, and Datta}{Jia
  et~al\mbox{.}}{2015}]%
        {Jia15}
\bibfield{author}{\bibinfo{person}{Limin Jia}, \bibinfo{person}{Shayak Sen},
  \bibinfo{person}{Deepak Garg}, {and} \bibinfo{person}{Anupam Datta}.}
  \bibinfo{year}{2015}\natexlab{}.
\newblock \showarticletitle{A Logic of Programs with Interface-Confined Code}.
  In \bibinfo{booktitle}{\emph{{IEEE} 28th Computer Security Foundations
  Symposium (CSF)}}. \bibinfo{pages}{512--525}.
\newblock


\bibitem[\protect\citeauthoryear{Jung, Swasey, Sieczkowski, Svendsen, Turon,
  Birkedal, and Dreyer}{Jung et~al\mbox{.}}{2015}]%
        {JungSSSTBD15}
\bibfield{author}{\bibinfo{person}{Ralf Jung}, \bibinfo{person}{David Swasey},
  \bibinfo{person}{Filip Sieczkowski}, \bibinfo{person}{Kasper Svendsen},
  \bibinfo{person}{Aaron Turon}, \bibinfo{person}{Lars Birkedal}, {and}
  \bibinfo{person}{Derek Dreyer}.} \bibinfo{year}{2015}\natexlab{}.
\newblock \showarticletitle{Iris: Monoids and Invariants as an Orthogonal Basis
  for Concurrent Reasoning}. In \bibinfo{booktitle}{\emph{Proceedings of the
  42nd Annual {ACM} {SIGPLAN-SIGACT} Symposium on Principles of Programming
  Languages, {POPL} 2015, Mumbai, India, January 15-17, 2015}},
  \bibfield{editor}{\bibinfo{person}{Sriram~K. Rajamani} {and}
  \bibinfo{person}{David Walker}} (Eds.). \bibinfo{publisher}{{ACM}},
  \bibinfo{pages}{637--650}.
\newblock
\urldef\tempurl%
\url{https://doi.org/10.1145/2676726.2676980}
\showDOI{\tempurl}


\bibitem[\protect\citeauthoryear{Kaminski, Katoen, Matheja, and
  Olmedo}{Kaminski et~al\mbox{.}}{2016}]%
        {KaminskiKMO16}
\bibfield{author}{\bibinfo{person}{Benjamin~Lucien Kaminski},
  \bibinfo{person}{Joost{-}Pieter Katoen}, \bibinfo{person}{Christoph Matheja},
  {and} \bibinfo{person}{Federico Olmedo}.} \bibinfo{year}{2016}\natexlab{}.
\newblock \showarticletitle{Weakest Precondition Reasoning for Expected
  Run-Times of Probabilistic Programs}, Vol.~\bibinfo{volume}{9632}.
  \bibinfo{pages}{364--389}.
\newblock
\urldef\tempurl%
\url{https://doi.org/10.1007/978-3-662-49498-1\_15}
\showDOI{\tempurl}


\bibitem[\protect\citeauthoryear{Katsumata}{Katsumata}{2014}]%
        {10.1145/2535838.2535846}
\bibfield{author}{\bibinfo{person}{Shin{-}ya Katsumata}.}
  \bibinfo{year}{2014}\natexlab{}.
\newblock \showarticletitle{Parametric Effect Monads and Semantics of Effect
  Systems}. In \bibinfo{booktitle}{\emph{Proceedings of the 41st ACM
  SIGPLAN-SIGACT Symposium on Principles of Programming Languages}} (San Diego,
  California, USA) \emph{(\bibinfo{series}{POPL ’14})}.
  \bibinfo{publisher}{Association for Computing Machinery},
  \bibinfo{address}{New York, NY, USA}, \bibinfo{pages}{633–645}.
\newblock
\showISBNx{9781450325448}
\urldef\tempurl%
\url{https://doi.org/10.1145/2535838.2535846}
\showDOI{\tempurl}


\bibitem[\protect\citeauthoryear{Katsumata, Sato, and Uustalu}{Katsumata
  et~al\mbox{.}}{2018}]%
        {DBLP:journals/lmcs/KatsumataSU18}
\bibfield{author}{\bibinfo{person}{Shin{-}ya Katsumata},
  \bibinfo{person}{Tetsuya Sato}, {and} \bibinfo{person}{Tarmo Uustalu}.}
  \bibinfo{year}{2018}\natexlab{}.
\newblock \showarticletitle{Codensity Lifting of Monads and its Dual}.
\newblock \bibinfo{journal}{\emph{Log. Methods Comput. Sci.}}
  \bibinfo{volume}{14}, \bibinfo{number}{4} (\bibinfo{year}{2018}).
\newblock
\urldef\tempurl%
\url{https://doi.org/10.23638/LMCS-14(4:6)2018}
\showDOI{\tempurl}


\bibitem[\protect\citeauthoryear{Kozen}{Kozen}{1985}]%
        {Kozen85}
\bibfield{author}{\bibinfo{person}{Dexter Kozen}.}
  \bibinfo{year}{1985}\natexlab{}.
\newblock \showarticletitle{A Probabilistic {PDL}}.
\newblock  \bibinfo{volume}{30}, \bibinfo{number}{2} (\bibinfo{year}{1985}),
  \bibinfo{pages}{162--178}.
\newblock


\bibitem[\protect\citeauthoryear{MacLane}{MacLane}{1971}]%
        {maclane71}
\bibfield{author}{\bibinfo{person}{Saunders MacLane}.}
  \bibinfo{year}{1971}\natexlab{}.
\newblock \bibinfo{booktitle}{\emph{Categories for the Working Mathematician}}.
\newblock \bibinfo{publisher}{Springer-Verlag}, \bibinfo{address}{New York}.
\newblock
\newblock
\shownote{Graduate Texts in Mathematics, Vol. 5.}


\bibitem[\protect\citeauthoryear{Maillard, Hritcu, Rivas, and Muylder}{Maillard
  et~al\mbox{.}}{2020}]%
        {MaillardHRM20}
\bibfield{author}{\bibinfo{person}{Kenji Maillard}, \bibinfo{person}{Catalin
  Hritcu}, \bibinfo{person}{Exequiel Rivas}, {and} \bibinfo{person}{Antoine~Van
  Muylder}.} \bibinfo{year}{2020}\natexlab{}.
\newblock \showarticletitle{The next 700 relational program logics}.
\newblock \bibinfo{journal}{\emph{Proc. {ACM} Program. Lang.}}
  \bibinfo{volume}{4}, \bibinfo{number}{{POPL}} (\bibinfo{year}{2020}),
  \bibinfo{pages}{4:1--4:33}.
\newblock


\bibitem[\protect\citeauthoryear{Matache and Staton}{Matache and
  Staton}{2019}]%
        {MatacheS19}
\bibfield{author}{\bibinfo{person}{Cristina Matache} {and} \bibinfo{person}{Sam
  Staton}.} \bibinfo{year}{2019}\natexlab{}.
\newblock \showarticletitle{A Sound and Complete Logic for Algebraic Effects}.
  In \bibinfo{booktitle}{\emph{Foundations of Software Science and Computation
  Structures - 22nd International Conference, {FOSSACS} 2019, Held as Part of
  the European Joint Conferences on Theory and Practice of Software, {ETAPS}
  2019, Prague, Czech Republic, April 6-11, 2019, Proceedings}}
  \emph{(\bibinfo{series}{Lecture Notes in Computer Science},
  Vol.~\bibinfo{volume}{11425})}, \bibfield{editor}{\bibinfo{person}{Mikolaj
  Bojanczyk} {and} \bibinfo{person}{Alex Simpson}} (Eds.).
  \bibinfo{publisher}{Springer}, \bibinfo{pages}{382--399}.
\newblock
\urldef\tempurl%
\url{https://doi.org/10.1007/978-3-030-17127-8\_22}
\showDOI{\tempurl}


\bibitem[\protect\citeauthoryear{Morgan, {McIver}, and Seidel}{Morgan
  et~al\mbox{.}}{1996}]%
        {Morgan96}
\bibfield{author}{\bibinfo{person}{Carroll Morgan}, \bibinfo{person}{Annabelle
  {McIver}}, {and} \bibinfo{person}{Karen Seidel}.}
  \bibinfo{year}{1996}\natexlab{}.
\newblock \showarticletitle{Probabilistic Predicate Transformers}.
\newblock  \bibinfo{volume}{18}, \bibinfo{number}{3} (\bibinfo{year}{1996}),
  \bibinfo{pages}{325--353}.
\newblock


\bibitem[\protect\citeauthoryear{Nanevski, Banerjee, and Garg}{Nanevski
  et~al\mbox{.}}{2013}]%
        {NanevskiBG13}
\bibfield{author}{\bibinfo{person}{Aleksandar Nanevski},
  \bibinfo{person}{Anindya Banerjee}, {and} \bibinfo{person}{Deepak Garg}.}
  \bibinfo{year}{2013}\natexlab{}.
\newblock \showarticletitle{Dependent Type Theory for Verification of
  Information Flow and Access Control Policies}.
\newblock \bibinfo{journal}{\emph{{ACM} Trans. Program. Lang. Syst.}}
  \bibinfo{volume}{35}, \bibinfo{number}{2} (\bibinfo{year}{2013}),
  \bibinfo{pages}{6:1--6:41}.
\newblock
\urldef\tempurl%
\url{https://doi.org/10.1145/2491522.2491523}
\showDOI{\tempurl}


\bibitem[\protect\citeauthoryear{Nanevski, Morrisett, and Birkedal}{Nanevski
  et~al\mbox{.}}{2008}]%
        {NanevskiMB08}
\bibfield{author}{\bibinfo{person}{Aleksandar Nanevski},
  \bibinfo{person}{J.~Gregory Morrisett}, {and} \bibinfo{person}{Lars
  Birkedal}.} \bibinfo{year}{2008}\natexlab{}.
\newblock \showarticletitle{Hoare type theory, polymorphism and separation}.
\newblock \bibinfo{journal}{\emph{J. Funct. Program.}} \bibinfo{volume}{18},
  \bibinfo{number}{5-6} (\bibinfo{year}{2008}), \bibinfo{pages}{865--911}.
\newblock
\urldef\tempurl%
\url{https://doi.org/10.1017/S0956796808006953}
\showDOI{\tempurl}


\bibitem[\protect\citeauthoryear{Naor and Yogev}{Naor and Yogev}{2019}]%
        {NaorY19}
\bibfield{author}{\bibinfo{person}{Moni Naor} {and} \bibinfo{person}{Eylon
  Yogev}.} \bibinfo{year}{2019}\natexlab{}.
\newblock \showarticletitle{Bloom Filters in Adversarial Environments}.
\newblock \bibinfo{journal}{\emph{{ACM} Trans. Algorithms}}
  \bibinfo{volume}{15}, \bibinfo{number}{3} (\bibinfo{year}{2019}),
  \bibinfo{pages}{35:1--35:30}.
\newblock
\urldef\tempurl%
\url{https://doi.org/10.1145/3306193}
\showDOI{\tempurl}


\bibitem[\protect\citeauthoryear{Pitts and Stark}{Pitts and Stark}{1998}]%
        {stark98}
\bibfield{author}{\bibinfo{person}{Andrew Pitts} {and} \bibinfo{person}{Ian
  Stark}.} \bibinfo{year}{1998}\natexlab{}.
\newblock \showarticletitle{Operational Reasoning for Functions with Local
  State}.
\newblock In \bibinfo{booktitle}{\emph{Higher Order Operational Techniques in
  Semantics}}, \bibfield{editor}{\bibinfo{person}{Andrew Gordon} {and}
  \bibinfo{person}{Andrew Pitts}} (Eds.). \bibinfo{publisher}{Publications of
  the Newton Institute, Cambridge University Press}, \bibinfo{pages}{227--273}.
\newblock
\urldef\tempurl%
\url{http://www.inf.ed.ac.uk/~stark/operfl.html}
\showURL{%
\tempurl}


\bibitem[\protect\citeauthoryear{Plotkin}{Plotkin}{1973}]%
        {plotkin1973lambda}
\bibfield{author}{\bibinfo{person}{Gordon Plotkin}.}
  \bibinfo{year}{1973}\natexlab{}.
\newblock \bibinfo{title}{Lambda-definability and logical relations}.
\newblock
\newblock


\bibitem[\protect\citeauthoryear{Ristenpart, Shacham, and Shrimpton}{Ristenpart
  et~al\mbox{.}}{2011}]%
        {RistenpartSS11}
\bibfield{author}{\bibinfo{person}{Thomas Ristenpart}, \bibinfo{person}{Hovav
  Shacham}, {and} \bibinfo{person}{Thomas Shrimpton}.}
  \bibinfo{year}{2011}\natexlab{}.
\newblock \showarticletitle{Careful with Composition: Limitations of the
  Indifferentiability Framework}. In \bibinfo{booktitle}{\emph{Advances in
  Cryptology - {EUROCRYPT} 2011 - 30th Annual International Conference on the
  Theory and Applications of Cryptographic Techniques, Tallinn, Estonia, May
  15-19, 2011. Proceedings}} \emph{(\bibinfo{series}{Lecture Notes in Computer
  Science}, Vol.~\bibinfo{volume}{6632})},
  \bibfield{editor}{\bibinfo{person}{Kenneth~G. Paterson}} (Ed.).
  \bibinfo{publisher}{Springer}, \bibinfo{pages}{487--506}.
\newblock
\urldef\tempurl%
\url{https://doi.org/10.1007/978-3-642-20465-4\_27}
\showDOI{\tempurl}


\bibitem[\protect\citeauthoryear{Sato}{Sato}{2016}]%
        {DBLP:journals/entcs/Sato16}
\bibfield{author}{\bibinfo{person}{Tetsuya Sato}.}
  \bibinfo{year}{2016}\natexlab{}.
\newblock \showarticletitle{Approximate Relational Hoare Logic for Continuous
  Random Samplings}. In \bibinfo{booktitle}{\emph{The Thirty-second Conference
  on the Mathematical Foundations of Programming Semantics, {MFPS} 2016,
  Carnegie Mellon University, Pittsburgh, PA, USA, May 23-26, 2016}}
  \emph{(\bibinfo{series}{Electronic Notes in Theoretical Computer Science},
  Vol.~\bibinfo{volume}{325})}, \bibfield{editor}{\bibinfo{person}{Lars
  Birkedal}} (Ed.). \bibinfo{publisher}{Elsevier}, \bibinfo{pages}{277--298}.
\newblock
\urldef\tempurl%
\url{https://doi.org/10.1016/j.entcs.2016.09.043}
\showDOI{\tempurl}


\bibitem[\protect\citeauthoryear{Sato, Aguirre, Barthe, Gaboardi, Garg, and
  Hsu}{Sato et~al\mbox{.}}{2019}]%
        {SatoABGGH19}
\bibfield{author}{\bibinfo{person}{Tetsuya Sato}, \bibinfo{person}{Alejandro
  Aguirre}, \bibinfo{person}{Gilles Barthe}, \bibinfo{person}{Marco Gaboardi},
  \bibinfo{person}{Deepak Garg}, {and} \bibinfo{person}{Justin Hsu}.}
  \bibinfo{year}{2019}\natexlab{}.
\newblock \showarticletitle{Formal verification of higher-order probabilistic
  programs: reasoning about approximation, convergence, Bayesian inference, and
  optimization}.
\newblock \bibinfo{journal}{\emph{{PACMPL}}} \bibinfo{volume}{3},
  \bibinfo{number}{{POPL}} (\bibinfo{year}{2019}),
  \bibinfo{pages}{38:1--38:30}.
\newblock
\urldef\tempurl%
\url{https://dl.acm.org/citation.cfm?id=3290351}
\showURL{%
\tempurl}


\bibitem[\protect\citeauthoryear{\'{S}cibior, Kammar, V\'{a}k\'{a}r, Staton,
  Yang, Cai, Ostermann, Moss, Heunen, and Ghahramani}{\'{S}cibior
  et~al\mbox{.}}{2017}]%
        {Scibior:2017:DVH:3177123.3158148}
\bibfield{author}{\bibinfo{person}{Adam \'{S}cibior}, \bibinfo{person}{Ohad
  Kammar}, \bibinfo{person}{Matthijs V\'{a}k\'{a}r}, \bibinfo{person}{Sam
  Staton}, \bibinfo{person}{Hongseok Yang}, \bibinfo{person}{Yufei Cai},
  \bibinfo{person}{Klaus Ostermann}, \bibinfo{person}{Sean~K. Moss},
  \bibinfo{person}{Chris Heunen}, {and} \bibinfo{person}{Zoubin Ghahramani}.}
  \bibinfo{year}{2017}\natexlab{}.
\newblock \showarticletitle{Denotational Validation of Higher-order Bayesian
  Inference}.
\newblock \bibinfo{journal}{\emph{Proc. ACM Program. Lang.}}
  \bibinfo{volume}{2}, \bibinfo{number}{POPL}, Article \bibinfo{articleno}{60}
  (\bibinfo{date}{Dec.} \bibinfo{year}{2017}), \bibinfo{numpages}{29}~pages.
\newblock
\showISSN{2475-1421}
\urldef\tempurl%
\url{https://doi.org/10.1145/3158148}
\showDOI{\tempurl}


\bibitem[\protect\citeauthoryear{Tassarotti and Harper}{Tassarotti and
  Harper}{2019}]%
        {TassarottiH19}
\bibfield{author}{\bibinfo{person}{Joseph Tassarotti} {and}
  \bibinfo{person}{Robert Harper}.} \bibinfo{year}{2019}\natexlab{}.
\newblock \showarticletitle{A Separation Logic for Concurrent Randomized
  Programs}.
\newblock \bibinfo{journal}{\emph{Proc. ACM Program. Lang.}}
  \bibinfo{volume}{3}, \bibinfo{number}{POPL}, Article \bibinfo{articleno}{64}
  (\bibinfo{date}{Jan.} \bibinfo{year}{2019}), \bibinfo{numpages}{30}~pages.
\newblock
\urldef\tempurl%
\url{https://doi.org/10.1145/3290377}
\showDOI{\tempurl}


\bibitem[\protect\citeauthoryear{V{\'{a}}k{\'{a}}r, Kammar, and
  Staton}{V{\'{a}}k{\'{a}}r et~al\mbox{.}}{2019}]%
        {VakarKS19}
\bibfield{author}{\bibinfo{person}{Matthijs V{\'{a}}k{\'{a}}r},
  \bibinfo{person}{Ohad Kammar}, {and} \bibinfo{person}{Sam Staton}.}
  \bibinfo{year}{2019}\natexlab{}.
\newblock \showarticletitle{A domain theory for statistical probabilistic
  programming}.
\newblock \bibinfo{journal}{\emph{{PACMPL}}} \bibinfo{volume}{3},
  \bibinfo{number}{{POPL}} (\bibinfo{year}{2019}),
  \bibinfo{pages}{36:1--36:29}.
\newblock
\urldef\tempurl%
\url{https://dl.acm.org/citation.cfm?id=3290349}
\showURL{%
\tempurl}


\end{thebibliography}

\newpage
\appendix
\section{Subtyping and Additional Typing Rules}
\label{ap:typing-rules}

We present here the subtyping rules (Figure~\ref{fig:subtyping}) and the
typing rules for expressions about memories (Figure~\ref{fig:typ-mem}).

\begin{figure*}[!htb]
\small
\begin{gather*}
    \inferrule*[]
    { }
    {\Xi \vdash B \preceq B}
    \qquad
    \inferrule*[]
    {\Xi \vdash \tau\preceq\tau' \\ \Xi \vdash \tau'\preceq\tau''}
    {\Xi \vdash \tau\preceq\tau''}
    \qquad
    \inferrule*[]
    {K \leq K'}
    {\Xi \vdash\nat[K] \preceq \nat[K']}
    \qquad
    \inferrule*[]
    {\Xi \vdash \Sigma\subseteq\Sigma' \\\\ k \leq k' \\ \Xi \vdash \tau\preceq\tau' }
    {\Xi \vdash\mon{\Sigma
        ,k}{\tau}\preceq\mon{\Sigma',k'}{\tau'}}
    \\
    \inferrule*[]
    {\Xi \vdash \tau \preceq \tau' \\ \Xi \vdash\sigma\preceq\sigma'}
    {\Xi \vdash \tau\times\sigma \preceq \tau'\times\sigma'}
    \qquad
    \inferrule*[]
    {\Xi \vdash\tau' \preceq \tau \\ \Xi \vdash\sigma\preceq\sigma'}
    {\Xi \vdash \tau\To\sigma \preceq \tau'\To\sigma'}
    \qquad
    \inferrule*[]
    {\Xi \vdash \tau \preceq \tau'}
    {\Xi \vdash\forall\alpha. \tau \preceq \forall\alpha.\tau'}
  \end{gather*}
\caption{Subtyping rules. Here, $\Xi\vdash \Sigma\subseteq\Sigma'$
for $\Xi=\alpha_1,\dots,\alpha_n$ if for every
$\Sigma_1,\dots,\Sigma_n\subseteq\loc$ we have
$\Sigma\subst{\alpha_1}{\Sigma_1}\dots\subst{\alpha_n}{\Sigma_n}
\subseteq
\Sigma'\subst{\alpha_1}{\Sigma_1}\dots\subst{\alpha_n}{\Sigma_n}$.}
\label{fig:subtyping}
\end{figure*}

\begin{figure*}[!htb]
\small
\begin{gather*}
    \inferrule*[]
	{\Gamma \vdash \tilde{t} \colon \tmem \\ a\in\loc}
	{\Gamma \vdash \tilde{t}[a] \colon \tval}
\qquad
    \inferrule*[]
	{\Gamma \vdash \tilde{t} \colon \tmem \\ \Gamma \vdash \tilde{u} \colon \tval \\ a\in\loc}
	{\Gamma \vdash \tilde{t}[a \mapsto \tilde{u}] \colon \tmem}
\end{gather*}
\caption{Typing rules for memory access}\label{fig:typ-mem}
\end{figure*}

\section{Additional Proof Rules}\label{ap:proof-rules}

We first present the standard well-formed rules for HOL assertions in Figure~\ref{fig:hol-wf}.
The extended Hoare triples and quadruples from our logics can be internalized into HOL.
To this end, we add a predicate former $\htriple{\sigma}{\delta}{P}{t}{Q}$ that internalizes
the monadic judgments into. This predicates are well-formed when
$P$ is a well-formed assertion (which may contain a variable for the state $\st\colon\tmem$),
$Q$ is a well-formed assertion (which may contain variables for the state $\st\colon\tmem$,
and the result $\vl\colon\sigma$) and $\vdash t \colon \mon{\Sigma,k}{\sigma}$.
The interpretation is equivalent to the corresponding monadic judgment. 
We add rules to switch between systems in Figure~\ref{fig:hol-triples}.  
\sk{Here omitting the grade makes the
   definition of $H$ ill-formed.  I think $H$ needs to be parameterised
   like $H_{\sigma,\Sigma,k,\delta,P,Q}$, and defined as a predicate on
   $T_{\Sigma,k}(\sigma)$.  Here, $\sigma$ is a type,
   $\Sigma\subseteq P(Loc)$, $k\in Nat$, $\delta\in[0,\infty]$, $P$ is
   a $\Omega$-predicate on $\{s:\tmem\}$, and $Q$ is a $\Omega$-predicate
   on $\{s:\tmem,v:\sigma\}$.  }

We also introduce a relational analog of this predicate, via a predicate
former $\hquad{\sigma,\tau}{k\delta}{P}{t_1}{t_2}{Q}$. Morally, this is valid whenever

\begin{figure}[ht]
\small
\begin{gather*}
    \inferrule*{R \subseteq t_1 : \tau_1 \times\dots\times t_k : \tau_k \\
    \Gamma \vdash t_1 : \tau_1 \dots \Gamma \vdash t_k : \tau_k}
    {\holwf{\Gamma}{R(t_1,\dots,t_k)}}
    \qquad
    \inferrule*{\holwf{\Gamma}{\phi_1}\\ \holwf{\Gamma}{\phi_2}}
        {\holwf{\Gamma}{\phi_1 \wedge \phi_2}}
    \qquad
    \inferrule*{\holwf{\Gamma}{\phi_1}\\ \holwf{\Gamma}{\phi_2}}
        {\holwf{\Gamma}{\phi_1 \To \phi_2}}
    \\
    \inferrule*{\holwf{\Gamma, x : \tau}{\phi}}
        {\holwf{\Gamma}{\forall x. \phi}}
    \qquad
    \inferrule*{\holwf{\Gamma, x : \tau}{\phi}}
        {\holwf{\Gamma}{\exists x. \phi}}
\end{gather*}
\caption{Selected well-formedness rules of HOL}\label{fig:hol-wf}
\end{figure}

\begin{figure}[ht]
\small
\begin{gather*}
    \inferrule*[Right = $\mathcal{H}_I$]
    {\jhol{\Gamma}{\Psi}{\phi\subst{\res}{t}} \\
    \holwf{\Gamma,\st\colon M}{P}}
    {\jhol{\Gamma}{\Psi}{\htriple{\sigma}{0}{\phi}{\unit{t}}{P \wedge \phi}}}
    \\
    \inferrule*[Right = $\mathcal{H}_E$]
    {\jhol{\Gamma}{\Psi}{\htriple{\sigma}{\delta}{P}{t}{Q}} \\
        \jhol{\Gamma,x\colon \sigma}{\Psi}{\htriple{\tau}{\delta'}{Q\subst{\vl}{x}}{u}{R}}
    }
    {\jhol{\Gamma}{\Psi}{\htriple{\tau}{\delta+\delta'}{P}{\mlet{x}{t}{u}}{Q}}}
\end{gather*}
\caption{HOL rules for internalized triples}\label{fig:hol-triples}
\end{figure}

\begin{figure}[ht]
\small
\begin{gather*}
    \inferrule*[Right = Ax]
    {}
    {\jhol{\Gamma}{\Psi,\phi}{\phi}}
    \qquad
    \inferrule*[Right = $\land_I$]
    {\jhol{\Gamma}{\Psi}{\phi_1}\\ \jhol{\Gamma}{\Psi}{\phi_2}}
    {\jhol{\Gamma}{\Psi}{\phi_1 \wedge \phi_2}} \\
\qquad
    \inferrule*[Right = $\land_E$]{\jhol{\Gamma}{\Psi}{\phi_1 \wedge \phi_2}}
    {\jhol{\Gamma}{\Psi}{\phi_1}}
\qquad
    \inferrule*[Right = $\To_I$]
    {\jhol{\Gamma}{\Psi, \phi_1}{\phi_2}}
    {\jhol{\Gamma}{\Psi}{\phi_1 \To \phi_2}} \\
\qquad
    \inferrule*[Right = $\To_E$]
    {\jhol{\Gamma}{\Psi}{\phi_1 \To \phi_2} \\
    \jhol{\Gamma}{\Psi}{\phi_1}}
    {\jhol{\Gamma}{\Psi}{\phi_2}}
\end{gather*}
\caption{Selected HOL rules}\label{fig:hol}
\end{figure}

\begin{figure*}[ht]
\small
\begin{gather*}
\inferrule*[right = \sf U-VAR]
 {\jlc{\Gamma}{x}{\sigma} \\ \jhol{\Gamma}{\Psi}{\phi\subst{\res}{x}}}
 {\juhol{\Gamma}{\Psi}{x}{\sigma}{\phi}}
\qquad
\inferrule*[right= \sf U-ABS]
  {\juhol{\Gamma,x:\tau}{\Psi,\phi'}{t}{\sigma}{\phi}}
  {\juhol{\Gamma}{\Psi}{\lambda x:\tau. t}{\tau \to \sigma}{\forall x. \phi'
    \Rightarrow \phi\subst{\res}{\res\ x}}}
\\
\inferrule*[right= \sf U-APP]
      {\juhol{\Gamma}{\Psi}{t}{\tau\to \sigma}{\forall x.
        \phi'\subst{\res}{x}\Rightarrow\phi\subst{\res}{\res\ x}}\\ \juhol{\Gamma}{\Psi}{u}{\tau}{\phi'}}
  {\juhol{\Gamma}{\Psi}{t\ u}{\sigma}{\phi\subst{x}{u}}}
  \\
\inferrule*[right = \sf U-SUB]
    {\juhol{\Gamma}{\Psi}{t}{\sigma}{\phi'} \\
           \jhol{\Gamma}{\Psi}{\phi'\subst{\res}{t} \Rightarrow \phi\subst{\res}{t}}}
    {\juhol{\Gamma}{\Psi}{t}{\sigma}{\phi}}
    \\
\inferrule*[right= \sf U-CASE]
{\juhol{\Gamma}{\Psi}{t}{\sigma}{b \Rightarrow \phi} \\
\juhol{\Gamma}{\Psi}{u}{\sigma}{\neg b \Rightarrow \phi}}
        {\juhol{\Gamma}{\Psi}{\casebool{b}{t}{u}}{\sigma}{\phi}}
\end{gather*}
\caption{Selected UHOL rules}\label{fig:uhol}
\end{figure*}

\begin{figure*}[ht]
\small
\begin{gather*}
\inferrule*[right= \sf ABS]
      {\jrhol{\Gamma,x_1:\tau_1,x_2:\tau_2}{\Psi,\phi'}{t_1}{\sigma_1}{t_2}{\sigma_2}{\phi}}
{\jrhol{\Gamma}{\Psi}{\lambda x_1:\tau_1. t_1}{\tau_1 \to \sigma_1}{\lambda x_2:\tau_2. t_2}{\tau_2\to \sigma_2}{\forall x_1,x_2. \phi' \Rightarrow \phi\subst{\resl}{\resl\ x_1}\subst{\resr}{\resr\ x_2}}}
\\
\inferrule*[right = \sf APP]
      {
\jrhol{\Gamma}{\Psi}{t_1}{\tau_1\to \sigma_1}{t_2}{\tau_2\to \sigma_2}{
          \forall x_1,x_2. \phi'\subst{\resl}{x_1}\subst{\resr}{x_2}\Rightarrow \phi\subst{\resl}{\resl\ x_1}\subst{\resr}{\resr\ x_2}}\\
\jrhol{\Gamma}{\Psi}{u_1}{\tau_1}{u_2}{\tau_2}{
          \phi'}
}
{\jrhol{\Gamma}{\Psi}{t_1 u_1}{\sigma_1}{t_2 u_2}{\sigma_2}{\phi\subst{x_1}{u_1}\subst{x_2}{u_2}}}
\\
\inferrule*[right = \sf VAR]
      {\jlc{\Gamma}{x_1}{\sigma_1} \\ \jlc{\Gamma}{x_2}{\sigma_2} \\ \jhol{\Gamma}{\Psi}{\phi\subst{\resl}{x_1}\subst{\resr}{x_2}}}
{\jrhol{\Gamma}{\Psi}{x_1}{\sigma_1}{x_2}{\sigma_2}{\phi}}
\\
\inferrule*[right = \sf ABS{-}L]
      {\jrhol{\Gamma,x_1:\tau_1}{\Psi, \phi'}{t_1}{\sigma_1}{t_2}{\sigma_2}{\phi}}
{\jrhol{\Gamma}{\Psi}{\lambda x_1:\tau_1. t_1}{\tau_1 \to \sigma_1}{t_2}{\sigma_2}{\forall x_1. \phi' \Rightarrow \phi \subst{\resl}{\resl\ x_1}}}
\\
\inferrule*[right = \sf APP{-}L]
      {\jrhol{\Gamma}{\Psi}{t_1}{\tau_1\to \sigma_1}{u_2}{\sigma_2}{\forall x_1. \phi'\subst{\resl}{x_1} \Rightarrow \phi\subst{\resl}{\resl\ x_1}}\\
          \juhol{\Gamma}{\Psi}{u_1}{\sigma_1}{\phi'}}
{\jrhol{\Gamma}{\Psi}{t_1 u_1}{\sigma_1}{u_2}{\sigma_2}{\phi\subst{x_1}{u_1}}}
\\
\inferrule*[right = \sf VAR{-}L]
      {\phi\subst{\resl}{x_1} \in \Psi \\  \resr\not\in\ FV(\phi) \\
       \jlc{\Gamma}{t_2}{\sigma_2}}
      {\jrhol{\Gamma}{\Psi}{x_1}{\sigma_1}{t_2}{\sigma_2}{\phi}}
\end{gather*}
\caption{Selected RHOL rules}\label{fig:rhol}
\end{figure*}

\section{Proofs of soundness of the adversary rules}

\subsection{Soundness of the \rname{ADV-U} rule}

To prove soundness of the adversary rules we will use logical relations.
Logical predicates and relations~\cite{plotkin1973lambda} are a technique used in programming
language theory to prove properties such as strong normalization or contextual
equivalence.  The idea of logical relations (or predicates) is that they allow us
to prove that all inhabitants of a certain type satisfy a particular property
that is defined inductively on the structure of types, rather than terms.

%
%


We now define an
indexed logical predicate $\cL_{\phi,\delta}$ mapping types to predicates
(more concretely, it maps a type $\tau$ to a set of closed terms of type $\tau$).
%
The logical predicate is indexed by an invariant $\phi$, which is a predicate
over memories and a real
$\delta\in[0,1]$ and it is defined as follows:
\begin{align*}
    \cL_{\phi,\delta}(B) &\triangleq \{ b \colon B \} \\
    \cL_{\phi,\delta}(\sigma\to\tau) &\triangleq
        \begin{cases}
            \{ t \colon \sigma\to\tau \mid 
                \forall x \colon \sigma. x \in \cL_{\phi,\delta}(\sigma)
            \To (t~x)\in \cL_{\phi,\delta}(\tau)\} & \text{if } \eff{\sigma}\subseteq\eff{\tau} \\
            \{ t \colon \sigma\to\tau \mid \forall x \colon \sigma. (t~x)\in\cL_{\phi,\delta}(\tau)\}&\text{otherwise}
        \end{cases}\\
    \cL_{\phi,\delta}(\sigma\times\tau) &\triangleq \{t \colon \sigma\times\tau \mid
    \pi_1(t) \in \cL_{\phi,\delta}(\sigma)\wedge \pi_2(t) \in \cL_{\phi,\delta}(\tau) \} \\
    \cL_{\phi,\delta}(\mon{\Sigma,k}{\sigma}) &\triangleq
    \{ t\colon\mon{\Sigma,k}{\sigma}\mid \htriple{\sigma}{k\delta}{\phi}{t}{\vl\in\cL_{\phi,\delta}(\sigma) \wedge \phi}\} \\
    \cL_{\phi,\delta}(\forall \alpha.\tau) &\triangleq
    \{ t \colon \forall \alpha.\tau \mid \forall \Sigma. t
    \in \cL_{\phi,\delta}(\tau\subst{\alpha}{\Sigma}) \}
\end{align*}

%

The definition of the logical predicate involves two subtleties.
First note that there are two different definitions of the logical
predicate for arrow types, depending on whether the effect of the
argument is contained in the effect of the result. The idea is that if
it is not, then the argument is ignored, so we do not need to require
that it satisfies the logical predicate.  Otherwise, we get the usual
definition: a function satisfies the logical predicate for $\sigma\to\tau$
if arguments that satisfy the predicate for $\sigma$ get mapped to results
satisfying the predicate the predicate for $\tau$.
The second subtlety is that the definition ignores the first grading
of the monad. A different definition, without indexing the predicate
by $\phi$ and defining instead:
\[
    \cL_{\delta}(\mon{\Sigma,k}{\sigma}) \triangleq
    \{ t\colon\mon{\Sigma,k}{\sigma}\mid \forall \phi\in{\sf Safe}(\Sigma).\htriple{\sigma}{k\delta}{\phi}{t}{\vl\in\cL_{\phi,\delta}(\sigma) \wedge \phi}\} \\
\]

would impose overly strong conditions on monadic types that appear in
argument position.  Namely, it would force us to prove that they
preserve \emph{all} the invariants that are safe for a given region
$\Sigma$, but we only know that the oracle preserves a particular
invariant $\phi$.
Note however that the grading $k$ is used to scale the grading of the lifting.
The grading $\Sigma$ of the monad is used in the premise of the
Basic Lemma, which we now state, using the notion of safety that
we defined in Section~\ref{sec:ubl}.

\begin{lemma}[Basic lemma]\label{lem:basic-unary}
    Let $\dot\cP$ be as above, and assume that for every $\nu\in\Delta$,
    $\sample{\nu_\sigma} \in \dot\cP_0(\cL_{\phi,\delta}(\sigma))$.
    Let $\vdash t \colon \sigma$ be a closed term, and $\holwf{\st\colon M}{\phi}$
    such that $\phi\in{\sf Safe}(\eff{\sigma})$. Then, for all $\delta\geq 0$,
    $t \in \cL_{\phi,\delta}(\sigma)$.
\end{lemma}

\begin{proof}
    We actually prove a generalization of the Basic Lemma, which makes
    the cases of abstraction and application easier to handle. We omit
    the adversary context from the proof, but note that this extends easily
    since they must be replaced by closed terms of the appropriate type.

\begin{lemma}
    Let $\Xi; \Gamma,\Gamma'\vdash t : \sigma$ be a well-typed term and $\phi$ a predicate
    over memories such that $\phi\in{\sf Safe}(\eff{\sigma}\setminus\Xi)$. 
    Assume further that for every $(x:\tau)\in\Gamma'$,
    $\eff{\tau}\not\subseteq\eff{\sigma}$ and that $\Xi \vdash \xi$ is an instantiation
    of the context $\Xi$. Let also $\Gamma\xi\vdash \gamma$, $\Gamma'\xi\vdash\gamma'$ be instantiations
    of the typing contexts. 
    If $\gamma$ is such that $\gamma(x_i)\in \cP_{\phi,\delta}(\sigma_i)$
    for every $(x_i\colon \sigma_i)\in\Gamma$, then 
    $ t\gamma\gamma' \in \cP_{\phi,\delta}(\sigma\xi)$.
\end{lemma}

This can be proven by induction on the typing derivation. We show a few cases,
omitting the $\Xi$ context when irrelevant.

    \begin{itemize}

        \item Variable. We have $\Gamma, \Gamma'\vdash  x:\sigma$, and
            by assumption, $x$ cannot be in $\Gamma'$. Therefore,
            also by assumption, $\sem{x\gamma} \in \cL_{\phi,\delta}(\sigma)$.

        \item Abstraction. Assume $\phi \in \safe{\eff{\sigma\to\tau}}$.
            If $\eff{\sigma}\subseteq \eff{\tau}$, then also $\phi \in \safe{\eff{\sigma}}$.
            We apply I.H., and we have that
            \[t\gamma[x\mapsto u]\gamma' \in \cL_{\phi,\delta}(\tau)\] 
            for all $u \in \cL_{\phi,\delta}(\sigma)$. 
            Therefore,
            $(\lambda x.t)\gamma\gamma' \in \cL_{\phi,\delta}(\sigma,\tau)$.

            If $\eff{\sigma}\not\subseteq \eff{\tau}$, then by I.H., for all $u \colon \sigma$,
            $t\gamma(\gamma'[x\mapsto u]) \in \cL_{\phi,\delta}(\tau)$. Therefore,
            $(\lambda x.t)\gamma\gamma' \in \cL_{\phi,\delta}(\sigma,\tau)$ too.

        \item Application. By I.H., $t\gamma\gamma' \in \cL_{\phi,\delta}(\sigma\to\tau)$.
            If $\eff{\sigma}\subseteq\eff{\tau}$, then
            by I.H. we also have
            $u\gamma\gamma' \in \cL_{\phi,\delta}(\sigma)$, and therefore by
            definition
            $(t~u)\gamma\gamma' \in \cL_{\phi,\delta}(\sigma)$.
            Otherwise, by definition of the logical predicate we have again
            $(t~u)\gamma\gamma' \in \cL_{\phi,\delta}(\sigma)$.

        \item Unit. WLOG we can assume that $\Gamma'=\emptyset$. We have that
            $\phi\in{\sf safe}(\eff{\cT_{\Sigma}(\sigma)})$, so also 
            $\phi\in{\sf safe}(\eff\sigma))$. By I.H.,
            $t\gamma \in \cL_{\phi,\delta}(\sigma)$.
            By the properties of $\dot\cP$ :
            \[ (\unit{t})\gamma \in \htriple{\sigma}{0}{\phi}{-}{\vl\in\cL_{\phi,\delta}(\sigma) \wedge \phi}\]

        \item Bind. WLOG we can assume that $\Gamma'=\emptyset$. Our premises are
            $\Gamma\vdash t \colon \cT_\Sigma(\tau)$ and $\Gamma, x\colon \tau \vdash u \colon \cT_{\Sigma'}(\sigma)$,
            and by assumption, $\Phi\in{\sf Safe}(\Sigma\cup\Sigma'\cup\eff{\tau}\cup\eff{\sigma})$.
            Therefore, $\Phi\in{\sf Safe}(\Sigma\cup\eff{\tau})$ and $\Phi\in{\sf Safe}(\Sigma'\cup\eff{\sigma})$.
            So we can apply I.H. to both premises. For the first one, we have that
            \[ t\gamma \in  \htriple{\sigma}{k\delta}{\phi}{-}{\vl\in\cL_{\phi,\delta}(\sigma) \wedge \phi}\]
            and for the second one, we have that for all $e \in \cL_{\phi,\delta}(\sigma)$,
            \[ u\gamma[x\mapsto e] \in \htriple{\tau}{k'\delta}{\phi}{-}{\vl\in\cL_{\phi,\delta}(\tau) \wedge \phi}\]
            By properties of the lifting, we get
            \[ (\mlet{x}{t}{u})\gamma \in \htriple{\tau}{(k+k')\delta}{\phi}{-}{\vl\in\cL_{\phi,\delta}(\tau) \wedge \phi}\]

        \item Forall introduction. By assumption, $\Xi,\alpha; \Gamma \vdash t \colon \tau$,
            and $\phi\in{\sf Safe}(\eff{\forall \alpha. \tau}\setminus \Xi)$,
            so $\phi\in{\sf Safe}(\eff{\tau}\setminus (\Xi \cup \alpha))$.
            For all $\Xi \vdash \xi$ and $\Sigma$, by I.H.,
            $t\gamma \in \cL_{\phi,\delta}(\tau\xi[\alpha\mapsto\Sigma])$,
            and therefore
            $(\Lambda \alpha. t)\gamma \in \cL_{\phi,\delta}(\forall \alpha. \tau)$.

        \item Forall elimination. By assumption, $\Xi;\Gamma\vdash t\colon \forall \alpha.\tau$
            and $\phi \in {\sf Safe}(\eff{\tau\subst{\alpha}{\Sigma}}\setminus \Xi)$.
            Recall that $\eff{\forall \alpha. \tau} = \eff{\tau\subst{\alpha}{\emptyset}}$.
            Since $\eff{\tau\subst{\alpha}{\emptyset}} \subseteq \eff{\tau\subst{\alpha}{\Sigma}}\setminus \Xi$,
            then also $\phi \in {\sf Safe}(\eff{\forall \alpha. \tau} \setminus \Xi)$
            (in other words, if the effect is smaller, then the predicate is still safe).
            Therefore, we can apply I.H., and we get that for $\Xi\vdash \xi$
            $t\gamma \in \cL_{\phi,\delta}(\forall \alpha. \tau\xi)$.
            From this, we can conclude that
            $t[\Sigma]\gamma \in  \cL_{\phi,\delta}(\tau\subst{\alpha}{\Sigma}\xi)$.

    \end{itemize}

\end{proof}

And from this, we can conclude:

\begin{corollary}
    The \rname{ADV-U} rule is sound.
\end{corollary}

\begin{proof}
    By the premise on $\cA$, safety of $\phi$  and the Basic Lemma, we can prove that
    \[
            \cA \in \cL_{\phi,\delta}(\forall \alpha. (\sigma \to \cT_{\alpha,1}(\tau)) \to \cT_{\Sigma \cup \alpha, k}(\tau'))
    \]
    On the other hand, by the assumption on the oracle we have that:
    \[\lambda x. t 
    \in \cL_{\phi,\delta}(\sigma\to\cT_{\Sigma',1}(\tau))\]
    From this, we can derive:
    \[
            \cA[\Sigma'](\lambda x. t) \in \cL_{\phi,\delta}(\cT_{\Sigma \cup \Sigma', k}(\tau'))
    \]
    and the conclusion follows directly.
\end{proof}

\subsection*{Soundness of the \rname{ADV-R} rule}

We generalize the logical predicates to the relational case.
For an invariant $\holwf{\stl\colon M, \str\colon M}{\phi}$
and $\delta\in[0,\infty]$ we define a logical relation,
i.e., a map from a type $\sigma$ to pairs of closed terms of type $\sigma$ as follows:
\begin{align*}
    \cR_{\phi,\delta}(B) &\triangleq \{ b_1, b_2 \colon B  \mid b_1 = b_2 \} \\
    \cR_{\phi,\delta}(\sigma\to\tau) &\triangleq
                                                \begin{cases}
            \left\{ t_1,t_2 \colon \sigma\to\tau \left|
                \begin{array}{l}
                \forall u_1, u_2 \colon \sigma. (u_1,u_2) \in \cR_{\phi,\delta}(\sigma)\\
                \To(t_1~u_1, t_2~u_2)\in \cR_{\phi,\delta}(\tau)
                \end{array}\right.\right\}
            \vspace*{.5mm}
            & \text{if}\ \eff{\sigma}\subseteq\eff{\tau} \\
            \{ t_1,t_2 \colon \sigma\to\tau \mid \forall u_1,u_2 \colon \sigma. 
            (t_1~u_1,t_2~u_2)\in\cR_{\phi,\delta}(\tau)\} &\text{otherwise}
        \end{cases}\\
    \cR_{\phi,\delta}(\sigma\times\tau) &\triangleq 
    \left\{t_1,t_2 \colon \sigma\times\tau \left|
        \begin{array}{l}
        (\pi_1(t_1),\pi_1(t_2)) \in \cR_{\phi,\delta}(\sigma)\\
        \wedge\; (\pi_2(t_1),\pi_2(t_2)) \in \cR_{\phi,\delta}(\tau)
        \end{array}\right.\right\}\\
        \cR_{\phi,\delta}(\cT_{\Sigma,k}(\sigma)) &\triangleq
    \{ t_1,t_2\colon\cT_{\Sigma,k}(\sigma)\mid 
    \hquad{\sigma}{k\delta}{\phi}{t_1}{t_2}{(\vll,\vlr)\in\cR_{\phi,\delta}(\sigma) \wedge \phi}\} \\
        \cR_{\phi,\delta}(\forall \alpha.\tau) &\triangleq
    \{ t_1,t_2 \colon \forall \alpha.\tau \mid \forall \Sigma. 
    (t_1[\Sigma], t_2[\Sigma]) \in \cR_{\phi,\delta}(\tau\subst{\alpha}{\Sigma}) \}
\end{align*}
The definition is analogous to the unary case. Note that the relation at
the base types changes, we now require equality to ensure that the computations
have the same control flow.
%
%
We state now the Basic Lemma:
\begin{lemma}[Basic lemma]\label{lem:basic-rel}
    Let $\ddot\cP$ be as above, and assume that for every $\nu\in\cD$,
    $(\sample{\nu},\sample{\nu}) \in \ddot\cP^{\it dp}_{(0,0)}(\cR_{\phi,\delta}(B))$.
    Let $\vdash t : \sigma$ be a closed term and $\stl\colon M, \str\colon M\vdash \phi$
    such that $\phi\in{\sf RSafe}(\eff{\sigma})$. 
    Then, for all $\epsilon,\delta\geq 0$, $(t,t)\in\cR_{\phi,\delta}(\sigma)$
\end{lemma}

\begin{proof}
    The proof is very similar to the unary case, by first stating a generalization
    and then proving it by induction on the typing judgment.
%
\end{proof}


By instantiating the Basic Lemma at the type of adversaries, we get
the following:
\begin{corollary}
    The \rname{ADV-R} rule is sound.
\end{corollary}

\section{Additional examples}

\subsection{Example: Accuracy for differentially private mechanisms}

Differential privacy~\cite{DR14} is a family of techniques focused on preventing
queries from databases from revealing private data about the entries in the database.
Concretely, we want to have plausible deniability that a concrete entry is in
the database.  This is usually achieved by making the queries through a
\emph{mechanism}, an algorithm that adds randomness to the output of the
queries before releasing them. Privacy comes with the tradeoff of accuracy;
the randomness that differentially private mechanisms add means that the result
of the query cannot be exact. Nonetheless, we can often give bounds about how
large the error can be.

In this example we will verify an accuracy bound for a differentially private
mechanism known as Sparse Vector (SV) algorithm. This algorithm is used to make $k$
numerical queries (i.e., queries whose output is a number) to a database, and
answering for how many of them the result is above some threshold $T$. It can
also be modelled in an online manner~\cite{BartheGGHS16}, in which an adversary
makes the queries one by one and observes the result (whether it is above or
below the threshold) before deciding which query to make next.

As in many differentially private mechanisms, the randomness in the SV
algorithm comes from a Laplace distribution, which has the following rule:

\[
\inferrule*[right =\sf LAP-UBL]
    { \Gamma \mid \Psi \vdash t \colon \RR}
  {\aumj{\Gamma}{\Psi}{P}
  {\lap{\epsilon}{t}}{\mon{\emptyset,0}{\RR}}{P\wedge |\vl - t| \leq \dfrac{1}{\epsilon} \log\dfrac{1}{\delta}}{\delta}}
\]

What this rule is saying is: (1) sampling does not change the state, and (2) with probability $1-\delta$,
the value sampled from $\lap{\epsilon}{t}$ will be within an interval of radius $\dfrac{1}{\epsilon} \log\dfrac{1}{\delta}$
centered at $t$. In other words, it tells us how accurate it is to replace $t$ by $\lap{\epsilon}{t}$.

We model the SV algorithm by having an oracle that provides access to the database.
The codes of the oracle and the main algorithm are given below:

\[\begin{array}{l} 	  \\
	\cO(q\colon \RR \to \RR):\\
	\quad x = {\sf evalQ}(q);\\
	\quad \mletA{y}{\lap{\epsilon/4}{x}}\\
	\quad \mletA{t'}{\rd{t}} \\
	\quad z = \casebool{y \geq t'}{\ttrue}{\ffalse};\\
	\quad a[r] := z; \\
	\quad b[r] := x; \\
	\quad {\sf inc}~r; \\
	\quad \unit{z};
	\end{array}
\qquad
\begin{array}{l} 	  \\
	{\bf mainSV}(T \colon \RR):\\
	\quad \mlet{x}{\lap{\epsilon/2}{T}}\\
	\quad t := x; \\
	\quad r := 0; \\
	\quad \cA(\cO);
	\end{array}
\]

The oracle runs the query $q$, adds Laplace noise to its result, checks whether it is above
the threshold $t$, and returns its result to the adversary. For the verification process,
we add some ghost variables: $r$ records the round number, $b$ is an array containing the
noiseless result of the queries, and $a$ records whether the noisy queries are above the threshold.

In the main procedure, we simply initialize the $t$ variable by adding noise to the threshold
$T$, set the auxiliary variables to $0$ and then instantiate the adversary with
the oracle $\cO$. We want to show the following accuracy bound
\[ \vdash \{ \top \} {\bf mainSV}(T) \colon \mon{\{a,b,r,t\}, 1}{\RR} \{\!\{ \forall i \leq k. \Phi(i) \}\!\}_{\beta} \]
where $\Phi$ is defined as:
\[ \Phi(i) \triangleq (\st[a[i]]=\ttrue \Rightarrow \st[b[i]] \geq T - \dfrac{6}{\epsilon}\log\dfrac{k+1}{\beta})
	\wedge (\st[a[i]]=\ffalse \Rightarrow \st[b[i]] \leq T + \dfrac{6}{\epsilon}\log\dfrac{k+1}{\beta})
\]
For this, we will use the adversary rule, with the invariant:
\[ \Phi'(j) = |T - \st[t]| \leq \dfrac{2}{\epsilon}\log\dfrac{k+1}{\beta} \wedge \forall i \leq j. \Phi(i) \]\
The first part of the invariant states how close the noisy threshold is to the original threshold $T$, while
the second part states an accuracy bound on the first $j$ queries.
The key part of the proof is showing that $\cO$ preserves the invariant for any
query, that is:
\[ \vdash \{ \Phi'(r) \} \cO(q) \{\!\{ \Phi'(r) \}\!\}_{\beta/(k+1)} \]
We will ellaborate further on this part of the proof. We go backwards from the end of the procedure,
and we will use the notation $\lfloor \cO(q) \rfloor_n$ to denote the program formed by the first
$n$ instructions of $\cO(q)$. Since the last instruction is $r:=r+1$,
we can apply the \rname{WRITE} rule, and we need to show
\[ \vdash \{ \Phi'(r) \} \lfloor \cO(q) \rfloor_5 \{\!\{|T - \st[t]| \leq \dfrac{2}{\epsilon}\log\dfrac{k+1}{\beta} \wedge \forall i \leq r+1. \Phi(i)\}\!\}_{\beta/(k+1)}  \]
Since our precondition contains already $\forall i \leq r. \Phi(i)$ and $t$ does not change, the only interesting part
is proving $\Phi(r+1)$ (the other cases can be proven by using \rname{AND-POST-U}). By applying \rname{WRITE} again and then \rname{CASE}, it is sufficient to show:
\[ \vdash \{ \Phi'(r) \} \lfloor \cO(q) \rfloor_2 \{\!\{  (\vl \geq \st[t] \wedge x \geq T - \dfrac{6}{\epsilon}\log\dfrac{k+1}{\beta})
	\vee (\vl < \st[t] \wedge x \leq T + \dfrac{6}{\epsilon}\log\dfrac{k+1}{\beta}) \}\!\}_{\beta/(k+1)}  \]
Now we can apply the \rname{LAP-UBL} rule setting $\delta$ to $\beta/(k+1)$. By computation, we can see that
\[ |\vl - x| \leq \dfrac{4}{\epsilon} \log \dfrac{k+1}{\beta} \wedge
\vl \geq \st[t] \wedge |T - \st[t]| \leq \dfrac{2}{\epsilon}\log\dfrac{k+1}{\beta} \Rightarrow x \geq T - \dfrac{6}{\epsilon}\log\dfrac{k+1}{\beta} \]
and
\[ |\vl - x| \leq \dfrac{4}{\epsilon} \log \dfrac{k+1}{\beta} \wedge
\vl <\st[t] \wedge |T - \st[t]| \leq \dfrac{2}{\epsilon}\log\dfrac{k+1}{\beta} \Rightarrow x \leq T + \dfrac{6}{\epsilon}\log\dfrac{k+1}{\beta} \]
and this completes the proof.

\section{Quasi-Borel Spaces}

Quasi-borel spaces are defined as follows:
\begin{definition}
  A quasi-Borel space is a pair $(X, M_X)$ of a set $X$ and a set
  $M_X \subseteq \RR \to X$ satisfying the following closure properties:
  \begin{enumerate}
  \item If $\alpha : \RR \to X$ is constant, then $\alpha \in M_X$.
  \item If $\alpha \in M_X$ and $f : \RR \to \RR$ is (Borel)
    measurable, then $\alpha \circ f \in M_X$.
  \item
    
    If $S \colon \RR \to \nat$ (Borel) measurable and
    $\{ \alpha_i \}_{i \in \nat} \subseteq M_X$ then
    $\lambda r. \alpha_{S(r)}(r) \in M_X$.
  \end{enumerate}
  A morphism between quasi-Borel spaces $(X, M_X)$ and $(Y, M_Y)$ is a
  function $f : X \to Y$ such that for every $\alpha \in M_X$,
  $f \circ \alpha \in M_Y$.
  Quasi-Borel spaces and morphisms between them form a category
  $\mathbf{QBS}$.
\end{definition}
\begin{lemma}
  For each standard Borel space $A$, $(A,\Meas(\RR,A))$ is a QBS.
\end{lemma}
There is a strong monad for probabilistic choices. We first define
the concept of probability measures on QBSs.
\begin{definition}
  Let $(X, M_X)$ be a quasi-Borel space. A (probability) measure over
  $(X, M_X)$ is a tuple $(A,\Sigma_A,\alpha,\mu)$ where
  $(A, \Sigma_A)$ is a standard Borel space, $\alpha \in \QBS(Z, X)$
  is a morphism and $\mu$ is a (probability) measure over
  $(A, \Sigma_A)$. We can define an equivalence relation between
  (probability) measures when they define the same integration
  operator:
  \begin{align*}
    (A,\Sigma_A,\alpha,\mu) \equiv (A',\alpha',\mu') \iff
    \forall f \in \QBS(X, \RR), 
    \int_A (f\circ \alpha) d\mu = \int_{A'} (f \circ \alpha') d\mu'
  \end{align*}
\end{definition}
\begin{definition}
  The {\em probability monad} on $\QBS$ is defined by
  \begin{align*}
    \mprob(X) &\triangleq \{ (A,\alpha,\mu)\ \text{probability measure over }X \}/\equiv  \\
    M_{\mprob(X)} &\triangleq \{ \lambda r. [D_r, \alpha(r,-), \mu_r] \mid \mu\ \sigma\text{-finite}, 
                    D \subseteq \RR \times A\ \text{measurable},
                    \alpha \in \QBS(D,X)\}
  \end{align*}
  where $D_r = \{ \omega \mid (r,\omega) \in D \}$.  The unit and
  Kleisli lifting are also defined by
  \begin{align*}
    \eta_X(x) & \triangleq (\{ \ast \}, \lambda y. x, \delta)
                \quad\text{where}\quad \delta(\{\ast\})=1.
    \\
    f^\#(A,\alpha,\mu) & \triangleq [D, \beta, (\mu \otimes \mu')|_D]
                         \quad\text{where} \quad (f \in \QBS(X,\mprob(Y)))
  \end{align*}
  where $f(\alpha(r)) = [D_r, \beta(r,-), \mu']$, and $\otimes$
  computes the product measure.
\end{definition}

\subsection{The Interpretation of Non-Monadic Part of the Language}
\label{sec:interpb}

First some preparations.  We write
$d_{A,X}:(A\cdot 1)\times X\to A\cdot X$ for the canonical
isomorphism. This exists because $\QBS$ is a bi-CCC.  Let
$\qbszero^0\triangleq\iota^{0,1}_0:1\to 0\cdot 1$ and
$\qbssucc^K\triangleq[\iota^{K+1,1}_{i+1}]_{i\in K}:K\cdot 1\to (K+1)\cdot 1$.
The interpretation of the non-monadic part of the language is given as
follows. Note that this is an induction on the {\em derivation} of
$\Gamma\vdash t:\tau$; the last rule is the case of the subtyping
rule.
\begin{align*}
  \sem{\G \vdash x \colon \tau}
  &\triangleq \pi^\G_x\qquad\qquad (x\in\dom(\G)) \\
  \sem{\G \vdash \star \colon \tunit}
  &\triangleq {!_{\sem\G}} \\
  \sem{\G\vdash 0\colon\NN[0]}
  &\triangleq \qbszero^0\circ {!_{\sem\G}} \\
  \sem{\G\vdash S~t\colon\NN[K+1]}
  &\triangleq \qbssucc^{K}\circ\sem{\G\vdash t\colon\NN[K]} \\
  \sem{\G \vdash \lambda x. t \colon \sigma \to \tau}
  &\triangleq
    \lambda (\sem{\G, x:\sigma \vdash t \colon \tau}\circ \mix\G x\tau) \\
  \sem{\G \vdash t~u \colon \tau}
  &\triangleq
    \mathit{\ev}\circ \langle \sem{\G \vdash t \colon \sigma \to \tau},
    \sem{\G \vdash u \colon \tau} \rangle \\
  \sem{\G \vdash \langle t, u \rangle \colon \tau \times \sigma}
  &\triangleq
    \langle \sem{\G \vdash t \colon \tau}, \sem{\G\vdash u\colon\sigma}\rangle \\
  \sem{\G\vdash\casebool{b}{t_1}{t_2}:\tau}
  &\triangleq [\sem{\G\vdash t_1:\tau},\sem{\G\vdash t_2:\tau}]\circ d_{\sem\G,1}\circ\langle\sem{\G\vdash b:\tbool} ,id_{\sem\G}\rangle \\
  \sem{\G \vdash \pi_i(t) \colon \tau_i}
  &\triangleq 
    \pi_i\circ \sem{\G \vdash t \colon \tau_1 \times \tau_2}\qquad\qquad(i = 1,2) \\
  \sem{\G\vdash t:\forall\alpha.\tau} &= \sem{\G\vdash t:\tau} \\
  \sem{\G\vdash t:\tau[\Sigma/\alpha]} &= \sem{\G\vdash t:\forall\alpha.\tau} \\
  \sem{\G\vdash t:\tau} &= \coe{\tau'}{\tau}\circ\sem{\G\vdash t:\tau'}\qquad\qquad(\tau'\preceq\tau)
\end{align*}
The rules Adv and Adv-Inst are not interpreted, as adversaries will
be instantiated by closed terms. \sk{Is this correct?}

\subsection{Logic }

Semantics of the logic is defined inductively by:
\begin{align*}
  \sem{\holwf{\G}{\top}}
  &\triangleq
    \top \\
  \sem{\holwf{\G}{\bot}}
  &\triangleq
    \bot \\
    \sem{\holwf{\G}{P(t_1,\dots,t_k)}}
  &\triangleq 
    (\sem{\G\vdash t_1\colon \sigma_1}\times\cdots\times \sem{\G\vdash t_k\colon\sigma_k})^*\sem{P}\\
  \sem{\holwf{\G}{\phi \wedge \psi}}
  &\triangleq
    \sem{\holwf{\G}{\phi}} \sqcap \sem{\holwf{\G}{\psi}} \\
    \sem{\holwf{\G}{\phi \vee \psi}}
  &\triangleq
    \sem{\holwf{\G}{\phi}} \sqcup \sem{\holwf{\G}{\psi}} \\
  \sem{\holwf{\G}{\phi \Rightarrow \psi}}
  &\triangleq 
    \sem{\holwf{\G}{\phi}} \Rrightarrow \sem{\holwf{\G}{\psi}} \\
    \sem{\holwf{\G}{\forall (x:\sigma). \phi}}
  &\triangleq 
    \forall_{\sem\G,\sem\sigma}(m^*_{\G,x:\sigma}\sem{\holwf{\G,x:\sigma}{\phi}})
\end{align*}
where the interpretation $\sem{P}$ for each base predicate
$P : \sigma_1\times\cdots\times\sigma_k$
is a chosen element in
$\UPred\S{\sem{\sigma_1\times\cdots\times\sigma_k}}$.

\subsection{Graded Lifting for Union Bound in Example \ref{ex:lifting-UB}}
We show that  $\glubn$ is indeeed a $\S$-valued strong $([0,\infty],\leq,+,0)$-graded lifting $\glubn$ of $\mprobn$.

  (1) We show $P(y) = \top \implies \glub[X] 0 P(\eta^\mprob_X(y)) = \top$.
  For all $f \in \QBS(X,D\{0,1\})$ such that $P \sqsubseteq |f|$, we have $f(y) = 1$,
  hence $\Pr_{x\sim\eta^\mprob(y)}[f(x)=1] = 1$.
  
  (2) Let $\Xi \in \mprob \mprob X$. We show 
  $\glub[\mprob X] {\delta_1} {\glub[X] {\delta_2} P}{(\Xi)} = \top \implies 
  \glub[X] {\delta_1 + \delta_2} P {(\mu^\mprob (\Xi))} = \top$.
  
   For all $f \in \QBS(X,D\{0,1\})$ such that $P \sqsubseteq |f|$,
   We have $\glub[X] {\delta_2} P \sqsubseteq \glub[X] {\delta_2} {|f|}$.
  \begin{align*}
      \top &=
      \glub[\mprob X] {\delta_1} {\glub[X] {\delta_2} P}{(\Xi)}\\
      &\sqsubseteq 
      \glub[\mprob X] {\delta_1} {\glub[X] {\delta_2} {|f|}}{(\Xi)}\\
      &=
      \glub[\mprob X] {\delta_1} {\lambda \nu.~\Pr_{x\sim\nu}[f(x)=1] \geq 1-\delta_2}{(\Xi)}\\
      &=
      \forall g \in \QBS(\mprob(X),D\{0,1\}),(\lambda \nu.~\Pr_{x\sim\nu}[f(x)=1] \geq 1-\delta_2) \sqsubseteq |f|~.~
      \Pr_{\nu\sim\Xi}[g(\nu)=1] \geq 1-\delta_1\\
      &=
      \Pr_{\nu\sim\Xi}[\Pr_{x\sim\nu}[f(x)=1] \geq 1-\delta_2] \geq 1-\delta_1\\
      &\qquad \{\text{ Markov's inequality (it holds for probability measures on quasi-Borel spaces). }\}\\
      &\sqsubseteq
		\begin{cases}
      \EE_{\nu\sim\Xi}[\Pr_{x\sim\nu}[f(x)=1]] \geq (1-\delta_1)(1-\delta_2) & \delta_1 < 1 \land \delta_2 < 1\\
      1 \geq 1-\delta_1 & \text{otherwise}
      \end{cases}\\
      &\sqsubseteq
      \EE_{\nu\sim\Xi}[\Pr_{x\sim\nu}[f(x)=1]] \geq 1-\delta_1-\delta_2\\
      &=
      \Pr_{x\sim \mu^\mprob_X(\Xi)}[f(x)=1] \geq 1-\delta_1-\delta_2.
  \end{align*}
  (3) Let $x \in X$ and $\mu \in \mprob Y$.
  If $P(x) = \top$ and $\glub[Y] \delta Q(\mu) = \top$ then
  $\glub[Y] \delta {P\dot\times Q } {(\theta^\mprob(x,\mu))}$.
  We have $\forall f \in \QBS(Y,D\{0,1\}),Q \sqsubseteq |f|~.~\Pr_{x\sim\nu}[f(x)=1] \geq 1-\delta$
  
  Here, for any $g \in \QBS(X \times Y,D\{0,1\})$ such that $P \dot\times Q \sqsubseteq |g|$,
  $g(x,-) \colon Y \to D\{0,1\}$ satisfy $Q \sqsubseteq g(x,-)$.
  This is a specific property of the case of $\Omega = \S$.
  
  Hence,
  \begin{align*}
  \top &= \glub[Y] \delta Q(\mu)\\
  & = \forall f \in \QBS(Y,D\{0,1\}),Q \sqsubseteq |f|~.~\Pr_{x \sim \mu}[f(x)=1] \geq 1-\delta\\
  &\sqsubseteq
  \forall g \in \QBS(X \times Y,D\{0,1\}),P \dot\times Q \sqsubseteq |g|~.~
  \Pr_{y'\sim\mu}[g(x,y')=1] \geq 1-\delta\\
  & \qquad \{ \text{Fubini theorem (coherence property of strength $\theta^\mprob$)} \}\\
  &=
  \forall g \in \QBS(X \times Y,D\{0,1\}),P \dot\times Q \sqsubseteq |g|~.~
  \Pr_{(x',y') \sim (\theta^\mprob(x,\mu)}[g(x,y)=1] \geq 1-\delta\\
  &= \glub[Y] \delta {P\dot\times Q } {(\theta^\mprob(x,\mu))}
  \end{align*}
  (4) It is obvious that 
  $\delta \leq \delta' \implies \glub[X]\delta P(\nu) \sqsubseteq \glub[X]{\delta'} P(\nu)$
  holds by definition of $\glubn$.

\subsection{Graded Lifting for Differential privacy in Example \ref{ex:DP-lifting}}
We show that $\dot\mprob^{\mathrm{dp}}$ is a $\S$-valued strong $([0,\infty],+,0,\leq)^2$-graded relational lifting of $\mprob$.

(1) 
Assume $(x,y) \models P$.
For all $(f,g) \colon P \dot\to S(\epsilon',\delta')$, we have 
\[
(f^\sharp \eta^\mprob_X(x), g^\sharp \eta^\mprob_Y(y)) = (f(x),g(y)) \models S(\epsilon',\delta').
\]
This implies that $\dot\mprob_{X,Y}^{\mathrm{dp}}(0,0)(P)(\eta^\mprob_X(x),\eta^\mprob_Y(x))$ holds.

(2) 
Consider
$(\nu_1,\nu_2) \models \dot\mprob_{X,Y}^{\mathrm{dp}}(\epsilon,\delta)(P) $,
$(f,g) \colon P \dot\to \dot\mprob_{X',Y'}^{\mathrm{dp}}(\epsilon',\delta')(Q)$ and
$(k,l) \colon Q \dot\to S(\epsilon'',\delta'')$.
Then,
$(k^\sharp f(x), l^\sharp g(y)) \models S(\epsilon'+\epsilon'',\delta'+\delta'')$ holds for all
$(x,y) \models P$.
Hence, $(k^\sharp f, l^\sharp g) \colon P \dot\to S(\epsilon'+\epsilon'',\delta'+\delta'')$.
Hence, $(k^\sharp f^\sharp, l^\sharp g^\sharp) \colon \dot\mprob_{X,Y}^{\mathrm{dp}}(\epsilon,\delta)(P) \dot\to S(\epsilon+\epsilon'+\epsilon'',\delta+\delta'+\delta'')$.
Hence  $ (k^\sharp f^\sharp \nu_1, l^\sharp g^\sharp \nu_2) \models S(\epsilon+\epsilon'+\epsilon'',\delta+\delta'+\delta'')$.
Since $(k,l)$ is arbitrary, we obtain $(f^\sharp \nu_1,g^\sharp \nu_2) \models \dot\mprob_{X',Y'}^{\mathrm{dp}}(\epsilon'+\epsilon'',\delta'+\delta'')(Q)$.
Since $(\nu_1,\nu_2)$ is arbitrary we conclude $(f^\sharp,g^\sharp) \colon \dot\mprob_{X,Y}^{\mathrm{dp}}(\epsilon,\delta)(P) \dot\to \dot\mprob_{X',Y'}^{\mathrm{dp}}(\epsilon'+\epsilon'',\delta'+\delta'')(Q)$.

We have $(\mathrm{id}_{\mprob X},\mathrm{id}_{\mprob Y}) \colon \dot\mprob_{X,Y}^{\mathrm{dp}}(\epsilon',\delta')(P) \dot\to \dot\mprob_{X,Y}^{\mathrm{dp}}(\epsilon',\delta')(P)$.
Then we have 
\[
(\mu^\mprob_X,\mu^\mprob_Y)=(\mathrm{id}_{\mprob X}^\sharp,\mathrm{id}_{\mprob Y}^\sharp)
\colon \dot\mprob_{X,Y}^{\mathrm{dp}}(\epsilon,\delta)(\dot\mprob_{X,Y}^{\mathrm{dp}}(\epsilon',\delta')(P)) \dot\to \dot\mprob_{X,Y}^{\mathrm{dp}}(\epsilon'+\epsilon'',\delta'+\delta'')(P).
\]

(3)
Consider
$(x,y)\models P$ and 
$(\nu_1,\nu_2) \models \dot\mprob_{X',Y'}^{\mathrm{dp}}(\epsilon,\delta)(Q)$.

For all $(f,g) \colon P\dot\times Q \dot\to S(\epsilon',\delta')$,
we have $(f(x,-),g(y,-)) \colon Q \dot\to S(\epsilon',\delta')$,
(it is specific in the case of $\Omega = \S$).

Then,
\[
(f(x,-)^\sharp(\nu_1),g(y,-)^\sharp(\nu_2)) \models S(\epsilon+\epsilon',\delta+\delta').
\]
Since
$f(x,-)^\sharp(\nu_1) = f^\sharp \theta^\mprob_{X,X'}(x,\nu_1)$ and 
$g(y,-)^\sharp(\nu_2) = g^\sharp \theta^\mprob_{Y,Y'}(x,\nu_2)$
\tetsuya{Moggi's uniqueness theorem of strength},
we have 
\[
(\theta^\mprob_{X,X'}(x,\nu_1),\theta^\mprob_{Y,Y'}(x,\nu_2)) \models \dot\mprob_{X\times X',Y\times Y'}^{\mathrm{dp}}(\epsilon,\delta)(P\dot \times Q).
\]

(4)
It is obvious that if $\epsilon \leq \epsilon'$ and $\delta \leq \delta'$ then $\dot\mprob_{X,Y}^{\mathrm{dp}}(\epsilon,\delta)(P)(\nu_1,\nu_2) \sqsubseteq \dot\mprob_{X,Y}^{\mathrm{dp}}(\epsilon',\delta')(P)(\nu_1,\nu_2)$.

\section{Proofs of Section~\ref{sec:semantics}}

We begin by formally stating the soundness of HOL:

\begin{theorem}[Soundness of HOL]\label{thm:hol-soundness}
  Let $\jhol{\G}{\Psi}{\phi}$ be a derivable HOL judgment. Then it
  is valid, i.e.
  $\sem{\holwf{\G}{\bigwedge \Psi}} \sqsubseteq
  \sem{\holwf{\G}{\phi}}$ holds in $\UPred\S{\sem\G}$.
\end{theorem}

\begin{proof}
The proof is done by the induction on the derivation tree of $\G\mid\Psi \vdash \phi$.
It is almost obvious.
We see the rules {Ax}, $\Rightarrow_E$ and $\Rightarrow_I$.

(Ax)
A judgment of the form $\G\mid \Psi,\phi \vdash \phi$ is always valid:
\[
\semp{\G\vdash \bigwedge (\Psi, \phi) } = 
\semp{\G\vdash (\bigwedge \Psi) \land \phi}
=
\semp{\G\vdash \bigwedge \Psi} 
\cap
\semp{\G\vdash \phi}
\subseteq \semp{\G\vdash \phi}.
\]

($\Rightarrow_I$)
Suppose that the judgment $\jhol{\G}{\Psi, \phi_1}{\phi_2}$ is derivable.
By induction hypoithesis, it is valid. 
Then we have $ \semp{\G\vdash\bigwedge \Psi} \cap \semp{\G\vdash\phi_1} \subseteq \semp{\G\vdash \phi_2}$.
This implies
\[ \semp{\G\vdash\bigwedge \Psi} \subseteq 
(|\sem\G| \setminus \semp{\G\vdash\phi_1})  \cup \semp{\G\vdash \phi_2}
= \semp{\G\vdash \phi_1 \To \phi_2}.
\] 

($\Rightarrow_E$)
Suppose that
the judgments $\jhol{\G}{\Psi}{\phi_1 \To \phi_2}$ and $\jhol{\G}{\Psi}{\phi_1}$ are derivable. By induction hypoithesis, they are valid. 
We have
$\semp{\G\vdash\bigwedge \Psi} \subseteq \semp{\G\vdash \phi_1 \To \phi_2} = 
|\sem\G| \setminus \semp{\G\vdash \phi_1} \cup \semp{\G\vdash \phi_2}
$ and 
$\semp{\G\vdash\bigwedge \Psi} \subseteq \semp{\G\vdash \phi_1}$.
We have 
\[
\semp{\G\vdash\bigwedge \Psi} \subseteq
(|\sem\G| \setminus \semp{\G\vdash \phi_1} \cup \semp{\G\vdash \phi_2}) \cap 
\semp{\G\vdash \phi_1}
\subseteq \semp{\G\vdash \phi_2}. 
\]
\end{proof}

A judgment in UHOL can
be seen a pair of a typing judgment and a logical judgment that contains an
extra distinguished variable $\res$ referring to the typed term. With this in mind,
we show:

\begin{theorem}[Soundness of UHOL]
    \label{thm:uhol-fib-sound}
    Let $\juhol{\G}{\Psi}{t}{\sigma}{\phi}$ be a derivable UHOL
    judgment. Then, for any $\g\in |\sem\G|$,
    $ \g\models\sem{\holwf{\G}{\textstyle\bigwedge \Psi}}$ implies
    $\sem{\G\vdash t:\sigma}(\g) \models
    \sem{\holwf{\G, \res : \sigma}{\phi}}_\gamma$.
\end{theorem}

Now we prove soundness of HO-UBL. We recall here the statement: 

\begin{proposition*}
  Let $\aumj{\G}{\Psi}{P}{t}{\mon{\Sigma,k}{\tau}}{Q}{\delta}$ be a
  derivable HO-UBL judgment without the adversary rule.
  Then, for any
  $\g\in|\sem{\G}|$, $\g\models\sem{\holwf{\G}{\textstyle\bigwedge \Psi}}$
  implies
  \begin{align*}
    &
      \sem{\G\vdash t : \mon{\Sigma,k}{\tau}}(\g)
      \models
      \dstmtr{\glubn}(\delta)(\sem{\G,\st:\tmem\vdash P}_\g,\sem{\G,\vl:\tau,\st:\tmem\vdash Q}_\g).
  \end{align*}
\end{proposition*}

  \begin{proof}[Proof of Proposition~\ref{prop:ho-ubl-sound}]
    In this proof we simply write $\lpstn$ for $\dstmtr{\glubn}$.
    The proof is by induction on the derivation. We show the
    more interesting cases:
    \begin{itemize}
    \item Unit. By soundness of non-monadic HO-UBL, we have that
      \begin{displaymath}
        \g \models \sem{\holwf{\G}{\bigwedge \Psi}} \To \sem{\G\vdash
          t:\tau}(\gamma) \models
        \sem{\holwf{\G,\res:\sigma}{\phi}}_\g
      \end{displaymath}
      Now consider an arbitrary $\holwf{\G, \st : \tmem}{P}$.  By
      definition,
      $\sem{\G\vdash \unit{t}:\mon{\Sigma,k}\sigma}(\g) = \eta^\mpstn \circ
      \sem{\G\vdash t:\sigma}(\g)$. By Lemma~\ref{lem:gr-lift} we
      conclude that
      \begin{displaymath}
        \eta^\mpstn \circ \sem{\G\vdash t:\sigma}(\g) \models
        \lpst[\sem\sigma]{0} { \sem{\holwf{\G, \st : \tmem}{P}}_\g,
          I\sem{\holwf{\G}{\phi}}\dtimes \sem{\holwf{\G, \st :
              \tmem}{P}}_\g }
      \end{displaymath}

    \item Bind. By I.H. we have that, for all
      $\g\models\sem{\holwf{\G}{\Psi}}$
      \begin{displaymath}
        \sem{\G\vdash t:\mon{\Sigma,k}\tau}(\g) \in
        \lpst{\delta}{\sem{\holwf{\G,\st:\tmem}{P}}_\g,
          \sem{\holwf{\G,\st:\tmem,\vl:\tau}{Q}}_\g}
      \end{displaymath} and
      for all $e\in|\sem\tau|$, $(\g,e)\models\sem{\holwf{\G,x:\tau}{\Psi}}$, so
      \begin{displaymath}
        \sem{\G, x:\tau \vdash u:\mon{\Sigma,k}\sigma}_{(\g,e)}
        \models \lpst{\delta'}{\sem{\holwf{\G,x:\tau,t:\tmem}{Q}}_{(\g,e)},
          \sem{\holwf{\G,x:\tau,\st:\tmem,\vl:\sigma}{R}}_{(\g,e)}}
      \end{displaymath}
      Since $x\not\in FV(R)$, then also
      \begin{displaymath}
        \sem{\G, x:\tau \vdash u:\mon{\Sigma,k}\sigma}(\g,e)
        \models
        \lpst{\delta'}{
          \sem{\holwf{\G,x:\tau,t:\tmem}{Q}}_{(\g,e)},
          \sem{\holwf{\G,\st:\tmem,\vl:\sigma}{R}}_\g
        }
      \end{displaymath}
      Note that
      $\lambda^{-1}(\sem{\G, x:\tau \vdash
        u:\mon{\Sigma,k}\sigma}(\g))$ is a morphism
      $\tau \times \tmem \to \mprob(\sigma\times \tmem)$ and that we can
      derive
      \begin{displaymath}
        \lambda^{-1}(\sem{\G, x:\tau \vdash u:\mon{\Sigma,k}\sigma}(\g))
        :
        \sem{\holwf{\G,x:\tau,t:\tmem}{Q}}_\g
        \dto
        \glprob{\delta'}{\sem{\holwf{\G,\st:\tmem,\vl:\sigma}{R}}_\g}
      \end{displaymath}
      By Lemma~\ref{lem:gr-lift},
      \begin{displaymath}
        \begin{array}{c} (\sem{\G, x:\tau \vdash u:\mon{\Sigma,k}\sigma}(\g))^\#:
          \lpst{\delta}{\sem{\holwf{\G,\st:\tmem}{P}}_\g, \sem{\holwf{\G,\st:\tmem,\vl:\tau}{Q}}_\g}
          \dto\\
          \lpst{\delta+\delta'}{
          \sem{\holwf{\G,\st:\tmem}{P}}_\g,
          \sem{\holwf{\G,\st:\tmem,\vl:\sigma}{R}}_\g}
        \end{array}
      \end{displaymath}
      We know that Bind can be equivalently defined as:
      \begin{displaymath}
        \sem{\G\vdash \mlet{x}{t}{u} : \mon{\Sigma,k}{\sigma}}
        \triangleq
        \sem{\G,x:\sigma\vdash u : \mon{\Sigma,k}{\sigma}}^\# \circ \theta
        \circ \langle id_{\G}, \sem{\G\vdash t : \mon{\Sigma,k}{\tau}}
        \rangle
      \end{displaymath} so we conclude.

    \item Read. Recall that
      $\sem{\G\vdash \rd{a} : \mon{\Sigma,k}{\tval}} \triangleq
      \lambda(\eta^\mprobn \circ \langle \pi_a, id \rangle \circ
      \pi_2)$, so
      $\sem{\G\vdash \rd{a} : \mon{\Sigma,k}{\tval}}(\g) \triangleq
      \eta^\mprobn \circ \langle \pi_a, id \rangle$.  Consider an
      arbitrary $\holwf{\G,\vl:\tval, \st : \tmem}{P}$. By definition
      of $\tmem$, we can see $P\subst{\vl}{\st[a]}$ as a predicate over
      $\sem\G\times \qbsmem$, and by the semantics of substitution,
      \begin{displaymath}
        \langle \pi_a, id \rangle
        :
        \sem{\holwf{\G,\st : \tmem}{P\subst{\vl}{\st[a]}}}_\g
        \dto
        \sem{\holwf{\G,\vl:\tval, \st : \tmem}{P}}_\g
      \end{displaymath}
      By the properties of $\glprobn$ and $\eta^\mprobn$, we conclude.

    \item Write. Recall that
      $\sem{\G\vdash \wrt{a}{t} : \mon{\Sigma,k}{\tunit}} \triangleq
      \lambda(\eta^\mprob\circ\langle !,id\rangle\circ
      u_a(\sem{\Gamma\vdash t:\tval}))$. As in the previous case, we use
      the semantics of substitution to show:
      \begin{displaymath}
        \langle !,id\rangle\circ u_a(\sem{\Gamma\vdash t:\tval})(\gamma,-)
        :
        \sem{\holwf{\G,\st:\tmem}{P\subst{\st}{\st[a\mapsto t]}}}_\g
        \dto
        \sem{\holwf{\G,\vl:\tunit, \st:\tmem}{P}}_\g
      \end{displaymath}

    \item Monadic Case. By I.H., Let $\g\in|\sem\G|$ such that
      $\g\models\sem{\G\vdash\bigwedge\Psi}$.
      This, together with
      the induction hypothesis on each branch
      entails
      \begin{displaymath}
        \sem{\G \vdash t_1 : \mon{\Sigma,k}{\tau}}(\g)
        \models
        \lpst{\delta}{
          \sem{\holwf{\G,\st:\tmem}{\langle b=\ttrue\rangle\sqcap P_1}}_\g,
          \sem{\holwf{\G,\st:\tmem,\vl:\tau}{Q}}_\g
        }
      \end{displaymath}
      \begin{displaymath}
        \sem{\G \vdash t_2 : \mon{\Sigma,k}{\tau}}(\g)
        \models
        \lpst{\delta}{
          \sem{\holwf{\G,\st:\tmem}{\langle b=\ffalse\rangle\sqcap P_2}}_\g,
          \sem{\holwf{\G,\st:\tmem,\vl:\tau}{Q}}_\g}
      \end{displaymath}
      Here we used the fact that
      $
        \gamma\models\Psi\wedge\Phi\implies P_\gamma\le Q_\gamma
      $
      implies
      $
        \gamma\models\Psi\implies (\pi_1^*I\Phi\sqcap P_\gamma)\le Q_\gamma.
      $
      Then from the standard
      reasoning on conditional expression we conclude
      \begin{displaymath}
        \sem{\G \vdash\casebool{b}{t_1}{t_2}}(\g)
        \models
        \lpst{\delta}{
          \sem{\holwf{\G,\st:\tmem}{(\langle b=\ttrue\rangle\sqcap P_1)\sqcup(\langle b=\ffalse\rangle\sqcap P_2)}}_\g,
          \sem{\holwf{\G,\st:\tmem,\vl:\tau}{Q}}_\g
        }
      \end{displaymath}
    \item Uniform sampling. Here we use the concrete definition of
      $\glprobn$.
      $\sem{\G \vdash \unif{\sigma} : \mon{\Sigma,k}\sigma} =
      \lambda. \theta \circ \langle \unif{\sem{\sigma}} , id \rangle
      $.  By definition of the uniform distribution, if
      $\{ x \in \sigma | x \in \phi \}/ |\sigma| = \delta$, then
      $\Pr_{x \sim \unif{\sem{\sigma}}}[x\in\phi] = \delta$, so
      $\unif{\sem{\sigma}} \in \glub_{1-\delta}(\phi)$, and therefore,
      \begin{align*}
        & \sem{\G \vdash \unif{\sigma} :
          \mon{\Sigma,k}\sigma} : \sem{\holwf{\G, \st:\tmem}{P}}\\
        &\dto
        \glub{1-\delta}{\sem{\holwf{x:\sigma}{\phi}}}\dtimes
        \sem{\holwf{\G, \st:\tmem}{P}}\\
        &\dto
        \glub{1-\delta}{\sem{\holwf{\G, x:\sigma, \st:\tmem}{\phi
            \wedge P}}}.
      \end{align*}
      %
    \end{itemize}
  \end{proof}

A judgment in RHOL can be seen as a triple formed by
two typing judgments and a logical judgment with two extra
distinguished variables.
\begin{theorem}[Soundness of RHOL]
  Let $\jrhol{\G}{\Psi}{t_1}{\sigma_1}{t_2}{\sigma_2}{\phi}$ be a
  derivable UHOL judgment.  Then, for any $\g\in |\sem\G|$,
  $\g\models\sem{\holwf{\G}{\textstyle\bigwedge \Psi}}$ implies
  \begin{displaymath}
    (\sem{\jlc{\G}{t_1}{\sigma_1}}(\g), \sem{\jlc{\G}{t_2}{\sigma_2}}(\g))
    \models
    \sem{\holwf{\G, \res_1:\sigma_1,\res_2:\sigma_2}{\phi}}_\gamma.
  \end{displaymath}
\end{theorem}

We now prove soundness of the relational logic. We first recall the statement:

\begin{proposition*}
  Let
  $\armj{\G}{\Psi}{P}{t_1}{\mon{\Sigma,k}{\tau_1}}{t_1}{\mon{\Sigma,k}{\tau_2}}{Q}{\delta}$
  be a derivable HO-PRL judgment without the \rname{ADV-R} rule.  Then
  for any $\g\in|\sem{\G}|$,
  $\gamma\models\sem{\holwf{\G}{\bigwedge \Psi}}$ implies 
  \begin{displaymath}
    (\sem{\G\vdash t_1: \mon{\Sigma,k}{\tau_1}}(\g),\sem{\G\vdash t_2:
      \mon{\Sigma,k}{\tau_2}}(\g))
    \models
    \lpst{0,\delta}{\sem P_\g,\sem Q_\g},
  \end{displaymath}
  where $\sem P_\g \triangleq \sem{\G,\st_1:\tmem,\st_2:\tmem\vdash
    P}_\g$ and  $Q_\g\triangleq\sem{\G,\st_1:\tmem,\vl_1:\sigma,\st_2:\tmem,\vl_2:\sigma\vdash
    Q}_\g$.
\end{proposition*}

    \begin{proof}[Proof of Proposition~\ref{prop:ho-rpl-sound}]
    To simplify the proof, we will use set theory notation, i.e.
    instead of $\g\in\QBS(1,\sem{\G})$ and $\langle\g^*,{id}\rangle\sem{\holwf{\G,\st\colon M}{P}}$,
    we write $\g\in\sem{\G}$ and $\sem{\holwf{\G,\st\colon M}{P}}(\g)$.
    We also use the shorthand 
    $\ddot\cT_{\sigma_1,\sigma_2,(\epsilon,\delta)} \triangleq \ddstmtr{\ddot\cP^{\it dp}}_{\sigma_1,\sigma_2,(\epsilon,\delta)}$,
    and omit $\sigma_1,\sigma_2$ when they are clear from the context.
    We only show a few interesting cases:
    \begin{itemize}
        \item \rname{UNIT-L}. By soundness of UHOL,
            \[ \langle id_{\sem{\G}},\sem{\G\vdash t\colon \tau_1}\rangle :
            \sem{\holwf{\G}{\bigwedge \Psi}}\to\sem{\holwf{\G, \res\colon \tau_1}{\phi}}\]
            so for all $\g\in\sem{\holwf{\G}{\bigwedge \Psi}}$,
            $\sem{\G\vdash t \colon  \tau_1}(\g)\in\sem{\holwf{\G, \res\colon\tau_1}{\phi}}(\g)$.
            Consider an arbitrary predicate $\holwf{\G, \stl \colon  M, \str\colon  M}{P}$. 
            By definition,
            \[\sem{\G\vdash \unit{t}\colon \cT\tau_1}(\g) = \eta^\cT (\sem{\G\vdash t\colon \tau_1}(\g))\]
            and
            \[\sem{\G\vdash {\sf skip}\colon \cT\tunit}(\g) = \eta(*) \]
            Where $*$ is the only element of the singleton set.
            By Lemma~\ref{lem:rel-gr-lift} we can conclude that
            \begin{align*}
                &{}(\eta^\cT (\sem{\G\vdash t\colon \tau_1}(\g)), \eta^\cT(*)) \in \\
                &{}\; \ddot\cT_{\tau_1,\tunit,(0,0)}(\sem{\holwf{\G, \stl \colon  M, \str}{P}}(\g), 
        \sem{\holwf{\G}{\phi}}(\g)\ddtimes \sem{\holwf{\G, \stl \colon  M, \str\colon M}{P}}(\g))
            \end{align*}

        \item \rname{BIND}. By applying I.H. to the first premise we have that,
        for all $\g\in\sem{\holwf{\G}{\Psi}}$
        \begin{align*} 
            &{}(\sem{\G\vdash t_1:\cT\tau_1}(\g),\sem{\G\vdash t_2:\cT\tau_2}(\g)) \in \\
            &{}\; \ddot\cT_{(\epsilon,\delta)}(\sem{\holwf{\G,\stl:M,\str:M}{P}}(\g), 
            \sem{\holwf{\G,\stl:M,\vll:\tau_1,\str:M,\vlr:\tau_2}{Q}}(\g))
        \end{align*}
        and for all $e_1 : \tau_2$, $e_2 : \tau_2$, we have 
        $(\g,e_1,e_2) \in \sem{\holwf{\G, x_1:\tau_1, x_2:\tau_2}{\Psi}}$, so
        by applying I.H. to the second premise,
        \begin{align*}
        &{}(\sem{\G, x_1:\tau_1,x_2:\tau_2, \vdash u_1:\cT\sigma_1}(\g,e_1,e_2),
        \sem{\G, x_1:\tau_1, x_2:\tau_2 \vdash u_2:\cT\sigma_2}(\g,e_1,e_2)) \in\\
        &{}\;\ddot\cT_{(\epsilon',\delta')}(\sem{\holwf{\G,x_1:\tau_1,x_2:\tau_2,\stl:M, \str:M}{Q}}(\g,e_1,e_2),\\
        &{}\qquad\qquad \sem{\holwf{\G,x_1:\tau_1,x_2:\tau_2,\stl:M,\vll:\sigma_1,\str:M,\vlr:\sigma_2}{R}}(\g,e_1,e_2))
        \end{align*}
        Since $x_1,x_2\not\in FV(R)$, then also
        \begin{align*}
        &{}(\sem{\G, x_1:\tau_1,x_2:\tau_2, \vdash u_1:\cT\sigma_1}(\g,e_1,e_2),
        \sem{\G, x_1:\tau_1, x_2:\tau_2 \vdash u_2:\cT\sigma_2}(\g,e_1,e_2)) \in\\
        &{}\;\ddot\cT_{(\epsilon',\delta')}(\sem{\holwf{\G,x_1:\tau_1,x_2:\tau_2,\stl:M, \str:M}{Q}}(\g,e_1,e_2),\\
        &{}\qquad\qquad \sem{\holwf{\G,\stl:M,\vll:\sigma_1,\str:M,\vlr:\sigma_2}{R}}(\g))
        \end{align*}
        
        Also, since $x_1\not\in FV(u_2)$ and $x_2\not\in FV(u_1)$, we have morphisms
        \[\lambda^{-1}(\sem{\G, x_1:\tau_1 \vdash u_1:\cT\sigma_1}(\g)) \colon 
            \tau_1 \times M \to \cP(\sigma_1\times M)\]
        and
        \[\lambda^{-1}(\sem{\G, x_2:\tau_2 \vdash u_2:\cT\sigma_2}(\g)) \colon 
            \tau_2 \times M \to \cP(\sigma_2\times M)\]
        so,
        \begin{align*} 
        &{}(\lambda^{-1}(\sem{\G, x_1:\tau_1 \vdash u_1:\cT\sigma_1}(\g)),
        \lambda^{-1}(\sem{\G, x_2:\tau_2 \vdash u_2:\cT\sigma_2}(\g))) \colon \\
        &{}\qquad\sem{\holwf{\G,x_1:\tau_1,\stl:M,x_2:\tau_2,\stl:M}{Q}}(\g) \to\\
        &{}\qquad\qquad    \ddot\cP^{\it dp}_{(\epsilon',\delta')}(
            \sem{\holwf{\G,\stl:M,\vll:\sigma_1,\str:M,\vll:\sigma_2}{R}}(\g))
        \end{align*}
        By the relational analogue of Lemma~\ref{lem:gr-lift}, we get
        \begin{align*}
            &{}((\sem{\G, x_1:\tau_1 \vdash u_1:\cT\sigma_1}(\g))^\#,
            (\sem{\G, x_2:\tau_2 \vdash u_2:\cT\sigma_2}(\g))^\#) \colon \\
            &{}\ddot\cT_{(\epsilon,\delta)}(\sem{\holwf{\G,\stl:M,\str:M}{P}}(\g), 
            \sem{\holwf{\G,\stl:M,\vll:\tau_1,\str:M,\vlr:\tau_2}{Q}}(\g)) \to\\
            &{}\ddot\cT_{(\epsilon+\epsilon',\delta+\delta')}
            (\sem{\holwf{\G,\stl:M,\str:M}{P}}(\g), 
            \sem{\holwf{\G,\stl:M,\vll:\sigma_1,\str:M,\vlr:\sigma_2}{R}(\g)})
        \end{align*}
       and from this we can conclude.
    \end{itemize}
\end{proof}



\end{document}